\documentclass[english,11pt]{article}
\usepackage[T1]{fontenc}
\usepackage[utf8]{inputenc}
\usepackage{geometry}
\usepackage{babel}
\usepackage{array}
\usepackage{float}
\usepackage{booktabs}
\usepackage{mathtools}
\usepackage{multirow}
\usepackage{amsmath}
\usepackage{amsthm}
\usepackage{amssymb}
\usepackage{graphicx}
\usepackage{setspace}
\usepackage[unicode=true,
 bookmarks=true,bookmarksnumbered=true,bookmarksopen=false,
 breaklinks=false,pdfborder={0 0 0},pdfborderstyle={},backref=false,colorlinks=false]
 {hyperref}
\hypersetup{pdftitle={Revisiting Event Study Designs: Robust and Efficient Estimation},
 pdfauthor={Borusyak, Kirill, Xavier Jaravel, and Jann Spiess}}

\makeatletter

\newcommand{\lyxmathsym}[1]{\ifmmode\begingroup\def\b@ld{bold}
  \text{\ifx\math@version\b@ld\bfseries\fi#1}\endgroup\else#1\fi}

\providecommand{\tabularnewline}{\\}

\theoremstyle{plain}
\newtheorem{assumption}{\protect\assumptionname}
\theoremstyle{plain}
\newtheorem{prop}{\protect\propositionname}
\theoremstyle{plain}
\newtheorem{thm}{\protect\theoremname}
\theoremstyle{plain}
\newtheorem{cor}{\protect\corollaryname}

\usepackage[nameinlink,nosort]{cleveref}
\Crefname{assumption}{Assumption}{Assumptions}
\Crefname{prop}{Proposition}{Propositions}
\Crefname{thm}{Theorem}{Theorems}
\Crefname{test}{Test}{Tests}
\Crefname{cor}{Corollary}{Corollaries}

\AtBeginDocument{
\Crefformat{equation}{(#2#1#3)}%
\Crefname{assumption}{Assumption}{Assumptions}%
\let\ref\Cref
}

\usepackage{mathtools}
\usepackage{pdflscape}

\usepackage[authordate, giveninits=true, backend=biber, uniquename=false, uniquelist=false, noibid, maxcitenames=2, mincitenames=1, footmarkoff]{biblatex-chicago}

\DeclareFieldFormat{citehyperref}{%
  \DeclareFieldAlias{bibhyperref}{noformat}%
  \bibhyperref{#1}}

\DeclareFieldFormat{textcitehyperref}{%
  \DeclareFieldAlias{bibhyperref}{noformat}%
  \bibhyperref{%
    #1%
    \ifbool{cbx:parens}
      {\bibcloseparen\global\boolfalse{cbx:parens}}
      {}}}

\savebibmacro{cite}
\savebibmacro{textcite}

\renewbibmacro*{cite}{%
  \printtext[citehyperref]{%
    \restorebibmacro{cite}%
    \usebibmacro{cite}}}

\renewbibmacro*{textcite}{%
  \ifboolexpr{
    ( not test {\iffieldundef{prenote}} and
      test {\ifnumequal{\value{citecount}}{1}} )
    or
    ( not test {\iffieldundef{postnote}} and
      test {\ifnumequal{\value{citecount}}{\value{citetotal}}} )
  }
    {\DeclareFieldAlias{textcitehyperref}{noformat}}
    {}%
  \printtext[textcitehyperref]{%
    \restorebibmacro{textcite}%
    \usebibmacro{textcite}}}

\AtEveryBibitem{%
  \clearfield{issn} %
  \clearfield{doi} %
  \clearfield{isbn} %
  \clearfield{eprint} %
  \clearfield{number} %

}



\allowdisplaybreaks

\makeatother

\usepackage[style=chicago-authordate]{biblatex}
\providecommand{\assumptionname}{Assumption}
\providecommand{\corollaryname}{Corollary}
\providecommand{\propositionname}{Proposition}
\providecommand{\theoremname}{Theorem}

\begin{document}
\title{Revisiting Event Study Designs:\\
Robust and Efficient Estimation{\Large{} }}
\author{\vspace{1.25cm}
}
\author{Kirill Borusyak\\
UC Berkeley and CEPR\and Xavier Jaravel\\
LSE and CEPR\and Jann Spiess\\
Stanford\thanks{Borusyak: \protect\href{mailto:k.borusyak@berkeley.edu}{k.borusyak@berkeley.edu};
Jaravel: \protect\href{mailto:x.jaravel@lse.ac.uk}{x.jaravel@lse.ac.uk};
Spiess: \protect\href{mailto:jspiess@stanford.edu}{jspiess@stanford.edu}.
This draft supersedes our 2018 manuscript, \textquotedblleft Revisiting
Event Study Designs, with an Application to the Estimation of the
Marginal Propensity to Consume.\textquotedblright{} We thank Alberto
Abadie, Isaiah Andrews, Raj Chetty, Itzik Fadlon, Ed Glaeser, Peter
Hull, Guido Imbens, Larry Katz, Jack Liebersohn, Benjamin Moll, Jonathan
Roth, Pedro Sant'Anna, Amanda Weiss, and three anonymous referees
for thoughtful conversations and comments. We are particularly grateful
to Jonathan Parker for his support in accessing and working with the
data and code from \textcite{Broda2014}. The results in the empirical
part of this paper are calculated based in part on data from Nielsen
Consumer LLC and marketing databases provided through the NielsenIQ
Datasets at the Kilts Center for Marketing Data Center at The University
of Chicago Booth School of Business. The conclusions drawn from the
NielsenIQ data are ours and do not reflect the views of NielsenIQ.
NielsenIQ is not responsible for, had no role in, and was not involved
in analyzing and preparing the results reported herein. Two accompanying
Stata commands are available from the SSC repository: \texttt{did\_imputation}
for treatment effect estimation with our imputation estimator and
pre-trend testing, and \texttt{event\_plot} for making dynamic event
study plots.}\vspace{0.25cm}
}
\date{This version: January 2024}
\maketitle
\begin{abstract}
We develop a framework for difference-in-differences designs with
staggered treatment adoption and heterogeneous causal effects. We
show that conventional regression-based estimators fail to provide
unbiased estimates of relevant estimands absent strong restrictions
on treatment-effect homogeneity. We then derive the efficient estimator
addressing this challenge, which takes an intuitive \textquotedblleft imputation\textquotedblright{}
form when treatment-effect heterogeneity is unrestricted. We characterize
the asymptotic behavior of the estimator, propose tools for inference,
and develop tests for identifying assumptions. Our method applies
with time-varying controls, in triple-difference designs, and with
certain non-binary treatments. We show the practical relevance of
our results in a simulation study and an application. Studying the
consumption response to tax rebates in the United States, we find
that the notional marginal propensity to consume is between 8 and
11 percent in the first quarter \textemdash{} about half as large
as benchmark estimates used to calibrate macroeconomic models \textemdash{}
and predominantly occurs in the first month after the rebate.

\noindent 
\global\long\def\expec#1{\mathbb{E}\left[#1\right]}%
\global\long\def\var#1{\mathrm{Var}\left[#1\right]}%
\global\long\def\cov#1{\mathrm{Cov}\left[#1\right]}%
\global\long\def\one{\mathbf{1}}%
\global\long\def\diag{\operatorname{diag}}%
\global\long\def\tr{\operatorname{tr}}%
\global\long\def\plim{\operatorname*{plim}}%
\newtheorem*{target*}{Estimation Target} \newtheorem{test}{Test} 

\noindent 
\global\long\def\E{\mathbb{E}}%
\global\long\def\Var{\mathrm{Var}}%
\global\long\def\wnorm{\|v\|_{\textnormal{H}}}%
\global\long\def\1{\textbf{1}}%
\global\long\def\0{\textbf{0}}%
$
\global\long\def\R{\mathbb{R}}%
\global\long\def\I{\mathbb{I}}%
\global\long\def\Null{\mathbb{O}}%
\global\long\def\N{\mathcal{N}}%
$

\noindent \newpage{}
\end{abstract}

\section{Introduction}

Event studies are one of the most popular tools in applied economics
and policy evaluation. An event study is a difference-in-differences
(DiD) design in which a set of units in the panel receive treatment
at different points in time. In this paper, we investigate the robustness
and efficiency of estimators of causal effects in event studies, with
a focus on the role of treatment effect heterogeneity. We first develop
a simple econometric framework that delineates the identification
assumptions from each other and from the estimation target, defined
as some average of heterogeneous causal effects. We then apply this
framework in three ways. First, we analyze the conventional practice
of implementing event studies via two-way fixed effect Ordinary Least
Squares (TWFE OLS) regressions and show how the implicit conflation
of different assumptions leads to biases. Second, leveraging event
study assumptions in an explicit and principled way allows us to derive
the robust and efficient estimator, along with appropriate inference
methods and tests. The estimator takes an intuitive ``imputation''
form when treatment-effect heterogeneity is unrestricted. Finally,
we illustrate the practical relevance of our approach in an application
estimating the marginal propensity to spend (MPX) out of tax rebates;
our MPX estimates are lower than in prior work, implying that fiscal
stimulus is less powerful than commonly thought.

Event studies are frequently used to estimate treatment effects when
treatment is not randomized, but the researcher has panel data allowing
them to compare outcome trajectories before and after the onset of
treatment, as well as across units treated at different times. By
analogy to conventional DiD designs without staggered rollout, event
studies are commonly implemented by two-way fixed effect regressions,
such as
\begin{align}
Y_{it} & =\alpha_{i}+\beta_{t}+\tau D_{it}+\varepsilon_{it},\label{eq:baselineTEmodel}
\end{align}
where outcome $Y_{it}$ and binary treatment $D_{it}$ are measured
in periods $t$ and for units $i$, $\alpha_{i}$ are unit fixed effects
(FEs) that allow for different baseline outcomes across units, and
$\beta_{t}$ are period fixed effects that accommodate overall trends
in the outcome. Specifications like \ref{eq:baselineTEmodel} are
meant to isolate a treatment effect $\tau$ from unit- and period-specific
confounders. A commonly-used dynamic version of this regression includes
``lags'' and ``leads'' of the indicator for the onset of treatment,
to capture treatment effects for different ``horizons'' since the
onset of treatment and test for the parallel trajectories of the pre-treatment
outcomes.

To understand the problems with conventional two-way fixed effect
estimators in event-study designs and provide a principled econometric
approach to overcoming these issues, in \ref{sec:Setting} we develop
a simple framework that makes the estimation targets and underlying
assumptions explicit and clearly isolated. We suppose that the researcher
chooses a particular weighted average (or weighted sum) of heterogeneous
treatment effects they are interested in estimating. We make (and
later test) two standard DiD identification assumptions: that potential
outcomes without treatment are characterized by parallel trends and
that there are no anticipatory effects. We also allow for \textemdash{}
but do not require \textemdash{} an auxiliary assumption that the
treatment effects themselves follow some model that restricts their
heterogeneity for \emph{a priori }specified economic reasons. This
explicit approach is in contrast to regression specifications like
\ref{eq:baselineTEmodel}, both static and dynamic, which implicitly
conflate choices of estimation target and identification assumptions.
Our framework covers a broad class of empirically relevant estimands
beyond the standard average treatment-on-the-treated (ATT), including
heterogeneous treatment effects by observed covariates and ATTs at
different horizons that hold the composition of units fixed.

Through the lens of this framework, in \ref{sec:Conventional-Practice}
we uncover a set of challenges with conventional event-study estimation
methods and trace them back to a mismatch between estimation target,
identification assumptions, and the flexibility of the regression
specification. First, we note that failing to rule out anticipation
effects in ``fully-dynamic'' specifications (with all leads and
lags of the event included) leads to an underidentification problem
when there are no never-treated units, such that the dynamic path
of anticipation and treatment effects over time is not point-identified.
We conclude that it is important to separate out testing the assumptions
about pre-trends from the estimation of dynamic treatment effects
under those assumptions. Second, implicit assumptions of homogeneous
treatment effects embedded in static DiD regressions like \ref{eq:baselineTEmodel}
may lead to estimands that put negative weights on some long-run treatment
effects. With staggered rollout, regression-based estimation leverages
comparisons between groups that got treated over a period of time
and reference groups which had been treated \emph{earlier}. We label
such cases ``forbidden comparisons.'' Indeed, these comparisons
are only valid when the homogeneity assumption is true; when it is
violated, they can substantially distort the weights the estimator
places on treatment effects, or even make them negative. Third, in
dynamic specifications, implicit assumptions about treatment effect
homogeneity across groups first treated at different times lead to
the spurious identification of long-run treatment effects for which
no DiD comparisons valid under heterogeneous treatment effects are
available. The last two challenges highlight the danger of imposing
implicit treatment effect homogeneity assumptions instead of allowing
for heterogeneity and explicitly specifying the target estimand. We
show that these challenges are not resolved by trimming the sample
to a fixed window around the event date.

From the above discussion, the reader should not conclude that event
study designs are plagued by fundamental problems. On the contrary,
these challenges only arise due to a mismatch between treatment effect
heterogeneity and specifications which restrict it. We therefore use
our framework to circumvent these issues and derive robust and efficient
estimators.

In \ref{sec:Imputation-Solution}, we first establish a simple characterization
for the most efficient linear unbiased estimator of any pre-specified
weighted sum of treatment effects, in the baseline case of spherical
errors, i.e. homoskedasticity with no serial correlation. This estimator
explicitly incorporates the researcher's estimation goal and assumptions
about parallel trends, anticipation effects, and restrictions on treatment
effect heterogeneity. It is constructed by estimating a flexible high-dimensional
regression that differs from conventional event study specifications,
and aggregating its coefficients appropriately. While spherical errors
are a natural starting point, the principled construction of this
estimator more generally ensures unbiasedness and yields attractive
efficiency properties, as we later confirm in simulations.

In our leading case where the heterogeneity of treatment effects is
not restricted, the efficient robust estimator can be implemented
using a transparent ``imputation'' procedure. First, the unit and
period fixed effects $\hat{\alpha}_{i}$ and $\hat{\beta}_{t}$ are
fitted by regressions using untreated observations only. Second, these
fixed effects are used to impute the untreated potential outcomes
and therefore obtain an estimated treatment effect $\hat{\tau}_{it}=Y_{it}-\hat{\alpha}_{i}-\hat{\beta}_{t}$
for each treated observation. Finally, a weighted sum of these treatment
effect estimates is taken, with weights corresponding to the estimation
target.

To relate our efficient imputation estimator to other unbiased estimators
that have been proposed in the literature, we derive two additional
results showing the generality of the imputation structure. First,
any other linear estimator that is unbiased in our framework with
unrestricted causal effects can be represented as \emph{an} imputation
estimator, albeit with an inefficient way of imputing untreated potential
outcomes. Second, even when assumptions that restrict treatment effect
heterogeneity are imposed, any unbiased estimator can still be understood
as an imputation estimator for an adjusted estimand. Together, these
two results allow us to characterize estimators of treatment effects
in event studies as a combination of how they impute unobserved potential
outcomes and which weights they put on treatment effects.

For the efficient estimator in our framework, we provide tools for
valid inference. Specifically, we derive conditions under which the
estimator is consistent and asymptotically normal and propose standard
error estimates. Inference is challenging under arbitrary treatment
effect heterogeneity, because causal effects cannot be separated from
the error terms. We instead show how asymptotically conservative standard
errors can be derived, by attributing some variation in estimated
treatment effects to the error terms.\footnote{\label{fn:comparable estimators}While the generality of our setting
only allows for conservative inference (on any robust estimator, including
ours), we obtain asymptotically exact standard errors in the special
case that received the most attention in the literature: when units
are randomly sampled from a population and the estimand consists of
average treatment effects by period\textendash cohort pairs.} Our inference results apply under mild conditions in short panels.
Advancing the existing literature on DiD estimation with staggered
adoption, we also provide conditions for consistency and inference
that extend to panels where the number of time periods grows, as long
as growth is not too fast. We also propose a leave-one-out modification
to our conservative variance estimates with improved finite-sample
performance.

Another important practical advantage of our approach is that it provides
a principled way of testing the identifying assumptions of parallel
trends and no anticipation effects, based on OLS regressions with
untreated observations only. Compared to conventional specifications
with leads and lags of treatment that implicitly restrict treatment
effects, this approach avoids the contamination of the tests by treatment
effect heterogeneity shown by \textcite{Abraham2018}. Moreover, our
strategy circumvents the inference problems after pre-testing that
were pointed out by \textcite{Roth2018a}, under spherical errors.
These attractive properties result from the clear separation of estimation
and testing.

It is also useful to point out two limitations of our analysis. First,
all event study designs assume a restrictive parametric model for
untreated outcomes. We do not evaluate when these assumptions may
be applicable, and therefore when the event study design are \emph{ex
ante }appropriate, as \textcite{Roth2020} do. We similarly do not
consider estimation that is robust to violations of parallel-trend
type assumptions, as \textcite{Roth2019} propose, although our framework
allows relaxing those assumptions by including unit-specific trends
and time-varying covariates. We instead take the standard assumptions
of event study designs as given and derive optimal estimators, valid
inference, and practical tests to assess whether parallel-trend assumptions
hold. Second, we also do not consider event studies as understood
in the finance literature, based on high-frequency panel data, which
typically do not use period fixed effects \parencite{MacKinlay1997}.

In \ref{sec:Application}, we illustrate the practical relevance of
our theoretical insights by revisiting the estimation of the marginal
propensity to spend out of tax rebates in the event study of \textcite{Broda2014}.
First, we show that the choice of a binned specification used by \textcite{Broda2014}
leads to a substantial upward bias in estimated MPXs. Indeed, we find
that the binned specification puts a large weight on the effects happening
in the first week after the rebate receipt, and negative weights on
some longer-run effects, biasing the estimate upwards because the
spending response quickly decays over time. Second, we highlight that,
due to the implicit extrapolation of treatment effects in specifications
restricting treatment effect heterogeneity, some dynamic specifications
could be mistakenly interpreted as evidence for a large and persistent
increase in spending. Our imputation estimator eliminates unstable
patterns found across such specifications. Finally, we illustrate
the underidentification problem with the fully-dynamic specification:
the dynamic path of estimates is very sensitive to the choice of leads
to drop.

Our findings deliver several insights for the macroeconomics literature.
While commonly used estimates of the quarterly MPX covering all expenditures
range from 50-90\% and estimates of the quarterly MPX for nondurable
expenditure range from 15-25\%,\footnote{\textcite{Broda2014}, \textcite{Parker2013} and \textcite{Johnson2006a}
estimate different versions of the MPX out of tax rebates. \textcite{Laibson2022},
\textcite{kaplan2020marginal} and \textcite{di2020stock} provide
recent reviews of the literature on the estimation of the marginal
propensity to spend and consume.} our estimates, when appropriately rescaled, are about half as large,
at 25\textendash 37\% for the MPX one quarter after tax rebate receipt
for all expenditures and 8\textendash 11\% for nondurables. Using
the scaling methodology of \textcite{Laibson2022}, we estimate that
the model-consistent, or ``notional,'' MPC in the quarter following
the tax rebate ranges between 7.8\% and 11.4\%, compared with 15.9\%
to 23.4\% in the original estimation of \textcite{Broda2014}. Furthermore,
our preferred estimates are much more short-lived than benchmark estimates,
falling to a statistical zero beyond the first month after receiving
the tax rebate. Thus, our new estimates imply that fiscal stimulus
may be less potent than predicted by leading macroeconomic models
targeting benchmark estimates.\footnote{\textcite{Orchard2022} apply our imputation estimator to the \textcite{Parker2013}
quarterly data, covering the full consumption basket, and also obtain
estimates around half as large as in the original study. Our analysis
complements their results since, thanks to the high-frequency data,
it allows us to investigate the dynamics of the effect and explain
the source of the bias of conventional approaches. See also \textcite{Baker2021}
for evidence that using robust event study estimation methods matters
in other empirical contexts.}

For convenient application of our results, we supply a Stata command,
\texttt{did\_imputation}, which implements the imputation estimator
and inference for it in a computationally efficient way. Our command
handles a variety of practicalities which are also covered by our
theoretical results, such as time-varying covariates, triple-difference
designs, and repeated cross-sections. We also provide a second command,
\texttt{event\_plot}, for producing ``event study plots'' that visualize
the estimates with both our estimator and the alternative ones.

Our paper contributes to a growing methodological literature on event
studies. To the best of our knowledge, our paper is the first and
only one to characterize the underidentification and spurious identification
of long-run treatment effects that arise in conventional implementations
of event study designs. The negative weighting problem has received
more attention. It was first shown by \textcite[Supplement 1]{DeChaisemartin2015}.
The earlier manuscript of our paper \parencite{Borusyak2017} independently
pointed it out and additionally explained how it arises because of
forbidden comparisons and why it affects long-run effects in particular,
which we now discuss in \ref{subsec:Negative-Weighting} below. The
issue has since been further investigated by \textcite{Goodman-Bacon2021},
\textcite{Strezhnev2018}, and \textcite{DeChaisemartin2018}, while
\textcite{Abraham2018} have shown similar problems with dynamic specifications.
\textcite{Abraham2018} and \textcite{Roth2018a} have further uncovered
problems with conventional pre-trend tests, and \textcite{Schmidheiny2018}
have characterized the problems which arise from binning multiple
lags and leads in dynamic specifications. Besides being the first
to point out some of these issues, our paper provides a unifying econometric
framework which explicitly relates these issues to the conflation
of the target estimand and the underlying identification assumptions.

Several papers have proposed ways to address these problems, introducing
estimators that remain valid when treatment effects can vary arbitrarily
\parencite{DeChaisemartin2020,Abraham2018,Callaway2018,Marcus2020,Cengiz2019}.
An important limitation of these robust estimators is that their efficiency
properties are not known.\footnote{There are three notable exceptions. \textcite{Marcus2020} consider
a two-stage generalized method of moments (GMM) estimator and establish
its semiparametric efficiency under heteroskedasticity in a large-sample
framework with a fixed number of periods. However, they find this
estimator to be impractical, as it involves many moments, e.g. almost
as many as the number of observations in the application they consider.
Second, \textcite{Roth2021} characterize the efficient DiD estimator
which leverages random timing of the treatment, rather than a more
conventional parallel trends assumption, as we do. Finally, \textcite{harmon_DiD}
builds on our framework to characterize the efficiency properties
of difference-in-differences estimators when error terms follow a
random walk \textemdash{} the opposite case from our benchmark analysis
of efficiency which imposes no serial correlation of errors. In \ref{subsec:appx-GLS},
we generalize our results to intermediate cases, allowing for models
of heteroskedasticity and serial correlation.} A key contribution of our paper is to derive a practical, robust,
and finite-sample efficient estimator from first principles. We show
that this estimator takes a particularly transparent form under unrestricted
treatment effect heterogeneity, while our construction also yields
efficiency when some restrictions on treatment effects are imposed.
By clearly separating the testing of underlying assumptions from the
estimation step imposing these assumptions, we simultaneously increase
estimation efficiency and avoid problems with inference after pre-testing
under spherical errors. Our estimator uses all pre-treatment periods
for imputation, as appropriate under the standard DiD assumptions,
while alternative estimators use more limited information.\footnote{This efficiency gain relative to \textcite{DeChaisemartin2020} and
\textcite{Abraham2018} is obtained without stronger assumptions.
The \textcite{Callaway2018} assumptions are also equivalent to ours
when there is only one period before any unit is treated and there
are no covariates (see \textcite{Marcus2020}).}

In the MPX application, we find large gains of our imputation estimator:
the confidence interval is about 50\% longer for each week relative
to the rebate for \textcite{DeChaisemartin2020}, and 2\textendash 3.5
times longer for \textcite{Abraham2018} (which without extra controls
are equivalent to the two versions of the \textcite{Callaway2018}
estimator). We confirm these gains in a simulation study, finding
that the standard deviations of alternative robust estimators are
1.3\textendash 3.6 times higher with spherical errors, and that these
gains are generally preserved under heteroskedasticity and serial
correlation of errors.

Finally, our paper is related to a nascent literature that develops
robust estimators similar to the imputation estimator. To the best
of our knowledge, this idea has been first proposed for factor models
\parencite{Gobillon2016a,Xu2017}. \textcite{Athey2018b} consider
a general class of ``matrix-completion'' estimators for panel data
that first impute untreated potential outcomes by regularized factor-
and fixed-effects models and then average over the implied treatment-effect
estimates. The imputation idea has been explicitly applied to fixed-effect
estimators in event studies by \textcite{Liu2020a}, \textcite{Gardner2020a}, \textcite{Thakral2020}, and \textcite{thakral2023two}.
Specifically, the counterfactual estimator of \textcite{Liu2020a},
the two-stage estimator of \textcite{Gardner2020a}, \textcite{Thakral2020}, and \textcite{thakral2023two}, and a version
of the matrix-completion estimator from \textcite{Athey2018b} without
factors or regularization coincide with the imputation estimator in
our model for the specific class of estimands their papers consider.
Relative to these papers, we make four contributions: we derive a general imputation estimator from first principles, show its efficiency,
provide tools for valid asymptotic inference when unit fixed effects
are included, and show its robustness to pre-testing. Subsequently
to our work, \textcite{Wooldridge2021} derives a two-way Mundlak
estimator, which is also equivalent to the imputation estimator for
a restricted class of estimands in complete panels with controls that
are not allowed to change over time (but that may have time-varying
effects). The robustness and efficiency properties of our estimator
are not limited to those situations.

\section{Setting\label{sec:Setting}}

We consider estimation of causal effects of a binary treatment $D_{it}$
on an outcome $Y_{it}$ in a panel of units $i$ and periods $t$.
We focus on ``staggered rollout'' designs in which being treated
is an absorbing state. For each unit there is an event date $E_{i}$
when $D_{it}$ switches from 0 to 1 forever: $D_{it}=\one\left[K_{it}\ge0\right]$,
where $K_{it}=t-E_{i}$ is the number of periods since the event date
(``horizon''). Some units may never be treated, denoted by $E_{i}=\infty$.
Units with the same event date are referred to as a cohort.

We do not make any random sampling assumptions and work with a set
of observations $it\in\Omega$ of total size $N$, which may or may
not form a complete panel. We similarly view the event date for each
unit, and therefore all treatment indicators, as fixed. We define
the set of treated observations by $\Omega_{1}=\left\{ it\in\Omega\colon\ D_{it}=1\right\} $
of size $N_{1}$ and the set of untreated (i.e., never-treated and
not-yet-treated) observations by $\Omega_{0}=\left\{ it\in\Omega\colon\ D_{it}=0\right\} $
of size $N_{0}$.\footnote{Viewing the set of observations and event times as non-stochastic
is not essential. In \ref{subsec:Stochastic-Regressors}, we show
how this framework can be derived from one in which both are stochastic,
by appropriate conditioning. Our conditional framework avoids random
sampling assumptions made in other work on DiD designs (e.g. \cite{DeChaisemartin2018},
\cite{Abraham2018}, and \cite{Callaway2018}).}

We denote by $Y_{it}(0)$ the period-$t$ stochastic potential outcome
of unit $i$ if it is never treated. Causal effects on the treated
observations $it\in\Omega_{1}$ are denoted $\tau_{it}=\expec{Y_{it}-Y_{it}(0)}$.
We suppose a researcher is interested in a statistic which sums or
averages treatment effects $\tau=\left(\tau_{it}\right)_{it\in\Omega_{1}}$
over the set of treated observations with pre-specified non-stochastic
weights $w_{1}=\left(w_{it}\right)_{it\in\Omega_{1}}$ that can depend
on treatment assignment and timing, but not on realized outcomes:

\begin{target*}$\tau_{w}=\sum_{it\in\Omega_{1}}w_{it}\tau_{it}\equiv w_{1}^{\prime}\tau$.\end{target*}

\noindent For notation brevity, we consider scalar estimands.

Different weights are appropriate for different research questions.
The researcher may be interested in the overall ATT, formalized by
$w_{it}=1/N_{1}$ for all $it\in\Omega_{1}$. In event study analyses
a common estimand is the average effect $h$ periods since treatment
for a given horizon $h\ge0$: $w_{it}=\one\left[K_{it}=h\right]/\lvert\Omega_{1,h}\rvert$
for $\Omega_{1,h}=\left\{ it\colon\ K_{it}=h\right\} $. Our approach
also allows researchers to specify target estimands that place unequal
weights on units within the same cohort-by-horizon cell. For example,
one may be interested in weighting units by their size, or in estimating
a ``balanced'' version of horizon-average effects: the ATT at horizon
$h$ computed only for the subset of units also observed at horizon
$h^{\prime}$, such that the gap between two or more estimates is
not confounded by compositional differences. Finally, we do not require
the $w_{it}$ to add up to one; for example, a researcher may be interested
in the difference between average treatment effects at different horizons
or across some groups of units (e.g. women and men), corresponding
to $\sum_{it\in\Omega_{1}}w_{it}=0$.\footnote{More broadly, the choice of weights allows for estimation of treatment
effect heterogeneity by observed characteristics $R_{it}$. Indeed,
the slope of the linear projection of $\tau_{it}$ on some observable
$R_{it}$ (which may or may not be time-varying) is a weighted sum
of treatment effects, $\sum_{it\in\Omega_{1}}w_{it}\tau_{it}$ for
$w_{it}=\left(R_{it}-\bar{R}\right)/\sum_{js\in\Omega_{1}}(R_{js}-\bar{R})^{2}$
and $\bar{R}=\frac{1}{\lvert\Omega_{1}\rvert}\sum_{js\in\Omega_{1}}R_{js}$.
The same logic generalizes when $R_{it}$ is a vector, via the Frisch\textendash Waugh\textendash Lowell
theorem. This approach also allows for tests of restrictions on treatment
effect heterogeneity, e.g. to assess whether ATTs vary across time
horizons.}

To identify $\tau_{w}$, we consider three assumptions. We start with
the parallel-trends assumption, which imposes a two-way fixed effect
(TWFE) model on the untreated potential outcomes.
\begin{assumption}[Parallel trends]
\label{assu:A1} There exist non-stochastic $\alpha_{i}$ and $\beta_{t}$
such that $\expec{Y_{it}(0)}=\alpha_{i}+\beta_{t}$ for all $it\in\Omega$.\footnote{In estimation, we will set the fixed effect of either one unit or
one period to zero, such as $\beta_{1}=0$. This is without loss of
generality, since the TWFE model is otherwise over-parameterized.}
\end{assumption}
An equivalent formulation requires $\expec{Y_{it}(0)-Y_{it'}(0)}$
to be the same across units $i$ for all periods $t$ and $t'$ (whenever
$it$ and $it'$ are observed).

Parallel trend assumptions are standard in DiD designs, but their
details may vary. First, we impose the TWFE model on the entire sample.
Although weaker assumptions can be sufficient for identification of
$\tau_{w}$ \parencite[e.g.,][]{Callaway2021}, those alternative
restrictions depend on the realized treatment timing. Since parallel
trends is an assumption on \emph{potential} outcomes, we prefer its
stronger version which can be made \emph{a priori}.\footnote{Specifically, Assumption 4 in \textcite{Callaway2018} requires that
the TWFE model only holds for all treated observations ($D_{it}=1$),
observations directly preceding the treatment onset ($K_{it}=-1$),
and in all periods for never-treated units. Similarly, \textcite{Goodman-Bacon2021}
proposes to impose parallel trends on a ``variance-weighted'' average
of units, as the weakest assumption under which static specifications
we discuss in \ref{sec:Conventional-Practice} identify some average
of causal effects. While technically weaker, this assumption may be
hard to justify \emph{ex ante} without imposing parallel trends on
all units as it is unlikely that non-parallel trends will cancel out
by averaging.} Moreover, \ref{assu:A1} can be tested by using pre-treatment data,
while minimal assumptions cannot. Second, we impose \ref{assu:A1}
at the unit level, while sometimes it is imposed on cohort-level averages.
Our approach is in line with the practice of including unit, rather
than cohort, FEs in DiD analyses and allows us to avoid biases in
incomplete panels where the composition of units changes over time.
Moreover, we show in \ref{subsec:appx-Sampling} that, under random
sampling and without compositional changes, assumptions on cohort-level
averages imply \ref{assu:A1}.

Our framework extends immediately to richer models of $Y_{it}(0)$:\renewcommand{\theassumption}{1$^\prime$} 
\begin{assumption}[General model of $Y(0)$]
\label{assu:A1prime}For all $it\in\Omega$, $\expec{Y_{it}(0)}=A_{it}^{\prime}\lambda_{i}+X_{it}^{\prime}\delta$,
where $\lambda_{i}$ is a vector of unit-specific nuisance parameters,
$\delta$ is a vector of nuisance parameters associated with common
covariates, and $A_{it}$ and $X_{it}$ are known non-stochastic vectors.
\end{assumption}
\renewcommand{\theassumption}{\arabic{assumption}}\setcounter{assumption}{1} The
first term in this model of $Y_{it}(0)$ nests unit FEs, but also
allows to interact them with some observed covariates unaffected by
the treatment status, e.g. to include unit-specific trends. This term
looks similar to a factor model, but differs in that regressors $A_{it}$
are observed. The second term nests period FEs but additionally allows
any time-varying covariates, i.e. $X_{it}^{\prime}\delta=\beta_{t}+\tilde{X}_{it}^{\prime}\tilde{\delta}$.
In \ref{subsec:Stochastic-Regressors} we clarify that $X_{it}$ have
to be unaffected by treatment and strictly exogenous to be included
in the specification.

We next rule out anticipation effects, i.e. the causal effects of
being treated in the future on current outcomes (e.g. \cite{Abbring2003}):
\begin{assumption}[No anticipation effects]
\label{assu:A2}$Y_{it}=Y_{it}(0)$ for all $it\in\Omega_{0}$.
\end{assumption}
\ref{assu:A1,assu:A2} together imply that the observed outcomes $Y_{it}$
for untreated observations follow the TWFE model. It is straightforward
to weaken this assumption, e.g. by allowing anticipation for some
$k$ periods before treatment: this simply requires redefining event
dates to earlier ones. However, some form of this assumption is necessary
for DiD identification, as there would be no reference periods for
treated units otherwise.

Finally, researchers sometimes impose restrictions on causal effects,
explicitly or implicitly. For instance, $\tau_{it}$ may be assumed
to be homogeneous for all units and periods, or only depend on the
number of periods since treatment (but be otherwise homogeneous across
units and calendar periods). We will consider such restrictions as
a possible auxiliary assumption:
\begin{assumption}[Restricted causal effects]
\label{assu:A3} $B\tau=0$ for a known $M\times N_{1}$ matrix $B$
of full row rank.
\end{assumption}
It will be more convenient for us to work with an equivalent formulation
of \ref{assu:A3}, based on $N_{1}-M$ free parameters driving treatment
effects rather than $M$ restrictions on them:\renewcommand{\theassumption}{3$^\prime$} 
\begin{assumption}[Model of causal effects]
\label{assu:A3prime}$\tau=\Gamma\theta$, where $\theta$ is a $\left(N_{1}-M\right)\times1$
vector of unknown parameters and $\Gamma$ is a known $N_{1}\times\left(N_{1}-M\right)$
matrix of full column rank.
\end{assumption}
\renewcommand{\theassumption}{\arabic{assumption}}\setcounter{assumption}{3} \ref{assu:A3prime}
imposes a parametric model of treatment effects. For example, the
assumption that treatment effects all be the same, $\tau_{it}\equiv\theta_{1}$,
corresponds to $N_{1}-M=1$ and $\Gamma=\left(1,\dots,1\right)^{\prime}$.
Conversely, a ``null model'' $\tau_{it}\equiv\theta_{it}$ that
imposes no restrictions is captured by $M=0$ and $\Gamma=\mathbb{I}_{N_{1}}$.

If restrictions on the treatment effects are implied by economic theory,
imposing them will increase estimation power. Often, however, such
restrictions are implicitly imposed without an \emph{ex ante }justification,
but just because they yield a simple model for the outcome. We will
show in \ref{sec:Conventional-Practice} how estimators that rely
on this assumption can fail to estimate reasonable averages of treatment
effects, let alone the specific estimand~$\tau_{w}$, when the assumption
is violated.\footnote{We view the null \ref{assu:A3} as a conservative default. We note,
however, that this makes the assumptions inherently asymmetric in
that they impose restrictive models on potential control outcomes
$Y_{it}(0)$ (\ref{assu:A1}), but not on treatment effects $\tau_{it}$.
This asymmetry reflects the standard practice in staggered rollout
DiD designs and is natural when the structure of treatment effects
is \emph{ex ante }unknown, while our framework also accommodates the
case where the researcher is willing to impose structure. Restrictions
on treatment effects, when appropriate, are also useful for external
validity: unless some structure is imposed on treatment effects, one
cannot use estimates from past data to inform future policy, for instance
extending a given treatment to currently untreated units. However,
one can use our framework without restrictions to\textit{ }learn\textit{
}\textit{\emph{about the structure of treatment effects, e.g. whether
they vary across cohorts for each horizon}}.}

While we formulated our setting for staggered-adoption DiD designs
with binary treatments in panel data, our framework applies without
change in many related research designs. In \emph{repeated cross-sections},
a different random sample of units $i$ (e.g., individuals) from the
same groups $g(i)$ (e.g., regions) is observed in each period. Unit
FEs are not possible to include but can be replaced with group FEs
in \ref{assu:A1prime}: $\expec{Y_{it}(0)}=\alpha_{g(i)}+\beta_{t}$.
In \emph{triple-differences designs}, the data have two dimensions
in addition to periods, e.g. $i$ corresponds to a pair of region
$j(i)$ and demographic group $g(i)$. \ref{assu:A1prime} can be
specified as $\expec{Y_{it}(0)}=\alpha_{j(i)g(i)}+\alpha_{j(i)t}+\alpha_{g(i)t}$.\footnote{Another variation is when the outcome is measured in a single period
but across \emph{two cross-sectional dimensions}, such as regions
$i$ and birth cohorts $g$, with the treatment implemented in a set
of regions for the cohorts born after some cutoff period $E_{i}$
(e.g., \textcite{Hoynes2016}). Then one may write $\expec{Y_{ig}(0)}=\alpha_{i}+\beta_{g}$.} With \emph{non-binary treatment intensity}, our setting applies if
each unit is observed untreated before $E_{i}$ and treated with heterogenous
intensity $R_{it}\ne0$ (that may or may not vary over time) from
period $E_{i}$. \ref{assu:A1,assu:A2} can apply, and the researcher
can consider estimands such as the ``ATT per unit of intensity'',
$\frac{1}{\left|\Omega_{1}\right|}\expec{\sum_{it\in\Omega_{1}}\left(Y_{it}-Y_{it}(0)\right)/R_{it}}$,
by setting $w_{it}$ proportionally to $1/R_{it}$. The challenges
we describe in \ref{sec:Conventional-Practice} for standard staggered
DiDs and the solutions of \ref{sec:Imputation-Solution} directly
apply in all of these cases.\footnote{This is also the case of \emph{non-staggered DiD} designs, in which
units receive treatment in a single period or never. Our insights
in \ref{sec:Conventional-Practice} and \ref{sec:Imputation-Solution}
are still relevant if continuous covariates or unit-specific trends
are included (see \textcite{SantAnna2020} and \textcite{Wolfers2003}
for related ideas).}

\section{Challenges Pertaining to Conventional Practice\label{sec:Conventional-Practice}}

In this section, we first introduce the common two-way fixed effects
regressions with restricted treatment effect heterogeneity that have
traditionally been used in DiD designs. We then discuss several estimation
challenges that pertain to these specifications, including underidentification
in certain dynamic specifications, negative weighting, and spurious
identification of long-run causal effects. We conclude the section
by discussing how our framework also relates to other problems that
have been pointed out by \textcite{Roth2018a} and \textcite{Abraham2018}.

\subsection{Conventional Restrictive Specifications in Staggered Adoption DiD}

Causal effects in staggered adoption DiD designs have traditionally
been estimated via OLS regressions with two-way fixed effects, using
specifications that implicitly restrict treatment effect heterogeneity
across units. While details may vary, the following specification
covers many studies:
\begin{equation}
Y_{it}=\tilde{\alpha}_{i}+\tilde{\beta}_{t}+\sum_{\substack{h=-a\\
h\ne-1
}
}^{b-1}\tau_{h}\one\left[K_{it}=h\right]+\tau_{b+}\one\left[K_{it}\ge b\right]+\varepsilon_{it},\label{eq:dynamicOLS}
\end{equation}
Here $\tilde{\alpha}_{i}$ and $\tilde{\beta}_{t}$ are the unit and
period (``two-way'') fixed effects, $a\ge0$ and $b\ge0$ are the
numbers of included ``leads'' and ``lags'' of the event indicator,
respectively, and $\varepsilon_{it}$ is the error term. The first
lead, $\one\left[K_{it}=-1\right]$, is often excluded as a normalization,
while the coefficients on the other leads (if present) are interpreted
as measures of ``pre-trends,'' and the hypothesis that $\tau_{-a}=\dots=\tau_{-2}=0$
is tested visually or statistically. Conditionally on this test passing,
the coefficients on the lags are interpreted as a dynamic path of
causal effects: at $h=0,\dots,b-1$ periods after treatment and, in
the case of $\tau_{b+}$, at longer horizons binned together. We will
refer to this specification as \emph{``dynamic''} (as long as $a+b>0$)
and, more specifically, \emph{``fully-dynamic''} if it includes
all available leads and lags except $h=-1$, or \emph{``semi-dynamic''}
if it includes all lags but no leads.

Viewed through the lens of the \ref{sec:Setting} framework, these
specifications make implicit assumptions on untreated potential outcomes,
anticipation and treatment effects, and the estimand of interest.
First, they make \ref{assu:A1} but, for $a>0$, do not fully impose
\ref{assu:A2}, allowing for anticipation effects for $a$ periods
before treatment.\footnote{One can alternatively view this specification as imposing \ref{assu:A2}
but making a weaker \ref{assu:A1} which includes some pre-trends
into $Y_{it}(0)$. This difference in interpretation is immaterial
for our results.} Typically this is done as a means to \emph{test} \ref{assu:A2} rather
than to \emph{relax} it, but the resulting specification is the same.
Second, equation \ref{eq:dynamicOLS} imposes strong restrictions
on causal effect heterogeneity (\ref{assu:A3}), with treatment (and
anticipation) effects assumed to only vary by horizon $h$ and not
across units and periods otherwise. Most often, this is done without
an \emph{a priori }justification. If the lags are binned into the
term with $\tau_{b+}$, the effects are further assumed to be time-invariant
once $b$ periods have elapsed since the event. Finally, dynamic specifications
do not explicitly define the estimands $\tau_{h}$ as particular averages
of heterogeneous causal effects, even though researchers often consider
that effects may vary across observations, as evidenced by a literature
on the interpretation of OLS estimands going back to at least \textcite{Angrist1998a}
and \textcite{humphreys2009bounds}.

Besides dynamic specifications, equation \ref{eq:dynamicOLS} also
nests a very common specification used when a researcher is interested
in a single parameter summarizing all causal effects. With $a=b=0$,
we have the \emph{``static''} specification in which a single treatment
indicator is included:
\begin{equation}
Y_{it}=\tilde{\alpha}_{i}+\tilde{\beta}_{t}+\tau^{\text{static}}D_{it}+\varepsilon_{it}.\label{eq:staticOLS}
\end{equation}
In line with our \ref{sec:Setting} setting, the static equation imposes
the parallel trends and no anticipation \ref{assu:A1,assu:A2}. However,
it also makes a particularly strong version of \ref{assu:A3} \textemdash{}
that all treatment effects are the same. Moreover, the target estimand
is again not written out as an explicit average of potentially heterogeneous
causal effects.

In the rest of this section we turn to the challenges associated with
OLS estimation of equations \ref{eq:dynamicOLS} and \ref{eq:staticOLS}.
We explain how these issues result from the conflation of the target
estimand, \ref{assu:A2} and \ref{assu:A3}, providing a new and unified
perspective on the problems of static and dynamic specifications with
restricted treatment effect heterogeneity.

\subsection{Under-Identification of the Fully-Dynamic Specification\label{subsec:Underidentification}}

The first problem pertains to fully-dynamic specifications and arises
because a strong enough \ref{assu:A2} is not imposed. We show that
those specifications are under-identified if there is no never-treated
group:
\begin{prop}
\label{prop:Underid}If there are no never-treated units, the path
of $\left\{ \tau_{h}\right\} _{h\ne-1}$ coefficients is not point-identified
in the fully-dynamic specification. In particular, for any $\kappa\in\mathbb{R}$,
the path $\left\{ \tau_{h}+\kappa\left(h+1\right)\right\} $ fits
the data equally well, with the fixed effect coefficients appropriately
modified.
\end{prop}
\begin{proof}
All proofs are given in \ref{sec:appx-proofs}.
\end{proof}
To illustrate this result with a simple example, \ref{fig:UnderId}
plots the outcomes for a simulated dataset with two units (or equal-sized
cohorts), one treated at $t=2$ and the other at $t=4$. Both units
exhibit linear growth in the outcome, starting from different levels.
There are two interpretations of these dynamics. First, treatment
could have no impact on the outcome, in which case the level difference
corresponds to the unit FEs, while trends are just a common feature
of the environment, through period FEs. Alternatively, note that the
outcome equals the number of periods since the event for both groups
and all time periods: it is zero at the moment of treatment, negative
before, and positive after. A possible interpretation is that the
outcome is entirely driven by causal effects and anticipation of treatment.
Thus, one cannot hope to distinguish between unrestricted dynamic
causal effects and a combination of unit effects and time trends.\footnote{Formally, the problem arises because a linear time trend $t$ and
a linear term in the cohort $E_{i}$ (subsumed by the unit FEs) can
perfectly reproduce a linear term in horizon $K_{it}=t-E_{i}$. Therefore,
a complete set of treatment leads and lags, which is equivalent to
the horizon FEs, is collinear with the unit and period FEs.}

\begin{figure}
\caption{Underidentification of Fully-Dynamic Specification\label{fig:UnderId}}

\begin{centering}
\includegraphics[width=0.35\textwidth]{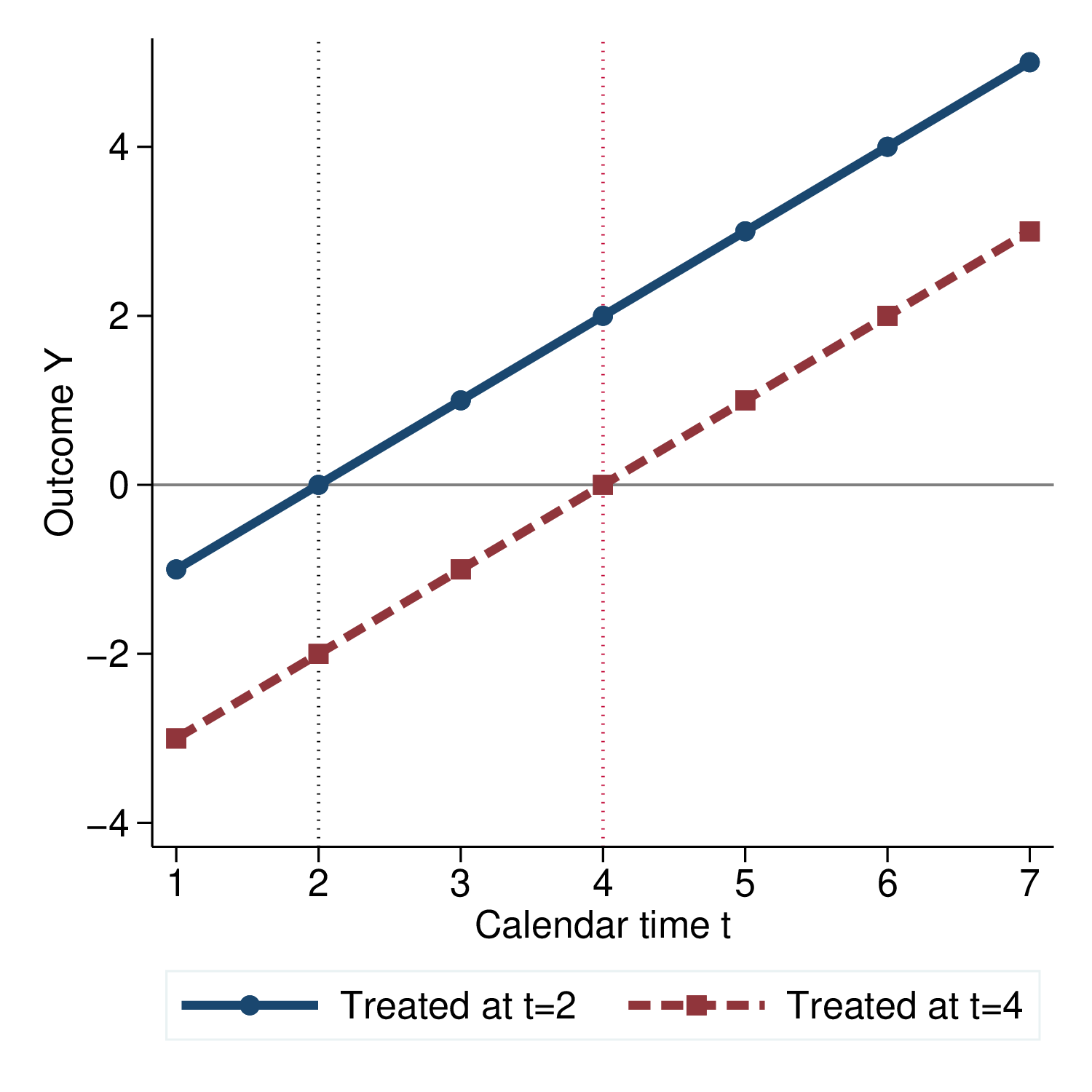}
\par\end{centering}
\raggedright{}\emph{\footnotesize{}Notes}{\footnotesize{}: This figure
shows the evolution of outcomes over seven periods for two units (or
cohorts), for the illustrative example of \ref{subsec:Underidentification}.
Vertical lines mark the periods in which the two units are first treated}{\small{}.}{\small\par}
\end{figure}

The problem may be important in practice, as statistical packages
may resolve this collinearity by dropping an arbitrary unit or period
indicator. Some estimates of $\left\{ \tau_{h}\right\} $ would then
be produced, but because of an arbitrary trend in the coefficients
they may suggest a violation of parallel trends even when the specification
is in fact correct, i.e.~\ref{assu:A1,assu:A2} hold and there is
no heterogeneity of treatment effects for each horizon (\ref{assu:A3}).

To break the collinearity problem, stronger restrictions on anticipation
effects, and thus on $Y_{it}$ for untreated observations, have to
be introduced. One could consider imposing minimal restrictions on
the specification that would make it identified. In typical cases,
only a linear trend in $\left\{ \tau_{h}\right\} $ is not identified
in the fully dynamic specification, while nonlinear paths cannot be
reproduced with unit and period fixed effects. Therefore, just one
additional normalization, e.g. $\tau_{-a}=0$ in addition to $\tau_{-1}=0$,
breaks multicollinearity.\texttt{}\footnote{Additional collinearity arises, e.g., when treatment is staggered
but happens at periodic intervals.}

However, minimal identified models rely on \emph{ad hoc }identification
assumptions which are \emph{a priori }unattractive. For instance,
just imposing $\tau_{-a}=\tau_{-1}=0$ means that anticipation effects
are assumed away $1$ and $a$ periods before treatment, but not in
other pre-periods. This assumption therefore depends on the realized
event times. Instead, a systematic approach is to impose the assumptions
\textemdash{} some forms of no anticipation effects and parallel trends
\textemdash{} that the researcher has an \emph{a priori }argument
for and which motivated the use of DiD. Such assumptions also give
much stronger identification power.\footnote{Our suggestion to impose identification assumptions at the estimation
stage does not mean that those assumptions should not also be tested;
we discuss testing in detail in \ref{subsec:Testing-PTA}.}

\subsection{Negative Weighting in the Static Regression\label{subsec:Negative-Weighting}}

We now show how, by imposing \ref{assu:A3} instead of specifying
the estimation target, the static TWFE specification does not identify
a reasonably-weighted average of heterogeneous treatment effects:
the underlying weights may be negative, particularly for the long-run
causal effects. The issues we discuss here also arise in dynamic specifications
that bin multiple lags together.

First, we note that, if the parallel-trends and no-anticipation assumptions
hold, the static specification identifies \emph{some }weighted average
of treatment effects:\footnote{This result was previously stated in Theorem 1 of \textcite{DeChaisemartin2018}
for general designs, and later in Appendix C of \textcite{Borusyak2017}
for staggered adoption designs.}
\begin{prop}
\label{prop:OLS-linear}If \ref{assu:A1,assu:A2} hold, then the estimand
of the static specification in \ref{eq:staticOLS} satisfies $\tau^{\text{static}}=\sum_{it\in\Omega_{1}}w_{it}^{\text{static}}\tau_{it}$
for some weights $w_{it}^{\text{static}}$ that do not depend on the
outcome realizations and add up to one, $\sum_{it\in\Omega_{1}}w_{it}^{\text{static}}=1$.
\end{prop}
The underlying weights $w_{it}^{\text{static}}$ can be computed from
the data using the Frisch\textendash Waugh\textendash Lovell theorem
(see equation \ref{eq:weights-FWL} in the proof of \ref{prop:OLS-linear})
and only depend on the timing of treatment for each unit and the set
of observed units and periods. The static specification's estimand,
however, cannot be interpreted as a \emph{proper} weighted average,
as some weights can be negative, which we illustrate with a simple
example:
\begin{prop}
\label{prop:example-static}Suppose \ref{assu:A1,assu:A2} hold and
the data consist of two units (or equal-sized cohorts), $A$ and $B$,
treated in periods 2 and 3, respectively, both observed in periods
$t=1,2,3$ (as shown in \ref{tab:twobythree}). Then the estimand
of the static specification \ref{eq:staticOLS} can be expressed as
$\tau^{\text{static}}=\tau_{A2}+\frac{1}{2}\tau_{B3}-\frac{1}{2}\tau_{A3}$.
\end{prop}
\begin{table}[H]
\begin{centering}
\caption{Two-Unit, Three-Period Example{\small{}\label{tab:twobythree}}}
\medskip{}
\par\end{centering}
\begin{centering}
{\footnotesize{}}%
\begin{tabular}{cll}
\toprule 
{\footnotesize{}$\expec{Y_{it}}$} & \multicolumn{1}{c}{{\footnotesize{}$i=A$}} & \multicolumn{1}{c}{{\footnotesize{}$i=B$}}\tabularnewline
\midrule
{\footnotesize{}$t=1$} & {\footnotesize{}$\alpha_{A}$} & {\footnotesize{}$\alpha_{B}$}\tabularnewline
{\footnotesize{}$t=2$} & {\footnotesize{}$\alpha_{A}+\beta_{2}+\tau_{A2}$} & {\footnotesize{}$\alpha_{B}+\beta_{2}$}\tabularnewline
{\footnotesize{}$t=3$} & {\footnotesize{}$\alpha_{A}+\beta_{3}+\tau_{A3}$} & {\footnotesize{}$\alpha_{B}+\beta_{3}+\tau_{B3}$}\tabularnewline
\midrule
{\footnotesize{}Event date} & \multicolumn{1}{c}{{\footnotesize{}$E_{i}=2$}} & \multicolumn{1}{c}{{\footnotesize{}$E_{i}=3$}}\tabularnewline
\bottomrule
\end{tabular}{\footnotesize\par}
\par\end{centering}
\begin{raggedright}
\medskip{}
\par\end{raggedright}
\raggedright{}\emph{\footnotesize{}Notes}{\footnotesize{}: This table
shows the evolution of expected outcomes over three periods for two
units (or cohorts), for the illustrative example of \ref{prop:example-static}.
Without loss of generality, we normalize $\beta_{1}=0$.}{\footnotesize\par}
\end{table}

This example illustrates the severe short-run bias of the static specification:
the long-run causal effect, corresponding to the early-treated unit
$A$ and the late period 3, enters with a negative weight ($-1/2$).
Thus, larger long-run effects make the coefficient smaller.

This problem results from what we call ``forbidden comparisons''
performed by the static specification. Recall that the original idea
of DiD estimation is to compare the evolution of outcomes over some
time interval for the units which got treated during that interval
relative to a reference group of units which didn't, identifying the
period FEs. In the \ref{prop:example-static} example, such an ``admissible''
comparison is between units $A$ and $B$ in periods 2 and 1, $\left(Y_{A2}-Y_{A1}\right)-\left(Y_{B2}-Y_{B1}\right)$.
However, panels with staggered treatment timing also lend themselves
to a second type of comparisons \textemdash{} which we label ``forbidden''
\textemdash{} in which the reference group has been treated throughout
the relevant period. For units in this group, the treatment indicator
$D_{it}$ does not change over the relevant period, and so the restrictive
specification uses them to identify period FEs, too. The comparison
between units $B$ and $A$ in periods 3 and 2, $\left(Y_{B3}-Y_{B2}\right)-\left(Y_{A3}-Y_{A2}\right)$,
in \ref{prop:example-static} is a case in point. While a comparison
like this is appropriate and increases efficiency when treatment effects
are homogeneous (which the static specification was designed for),
forbidden comparisons are problematic under treatment effect heterogeneity.
For instance, subtracting $\left(Y_{A3}-Y_{A2}\right)$ not only removes
the gap in period FEs, $\beta_{3}-\beta_{2}$, but also deducts the
evolution of treatment effects $\tau_{A3}-\tau_{A2}$, placing a negative
weight on $\tau_{A3}$. The restrictive specification leverages comparisons
of both types and estimates the treatment effect by $\hat{\tau}^{static}=\left(Y_{B2}-Y_{A2}\right)-\frac{1}{2}\left(Y_{B1}-Y_{A1}\right)-\frac{1}{2}\left(Y_{B3}-Y_{A3}\right)$.\footnote{The proof of \ref{prop:OLS-linear} shows why long-run effects in
particular are subject to the negative weights problem. In general,
negative weights arise for the treated observations, for which the
residual from an auxiliary regression of $D_{it}$ on the two-way
FEs is negative. \textcite{DeChaisemartin2018} show that, in complete
panels, the unit FEs are higher for early-treated units (which are
observed treated for a larger shares of periods) and period FEs are
higher for later periods (in which a larger shares of units are treated).
The early-treated units observed in later periods correspond to the
long-run effects.}

Fundamentally, this problem arises because the specification imposes
very strong restrictions on treatment effect homogeneity, i.e. \ref{assu:A3},
instead of acknowledging the heterogeneity and specifying a particular
target estimand (or perhaps a class of estimands that the researcher
is indifferent between).

With a large number of never-treated units or a large number of periods
before any unit is treated (relative to other units and periods),
our setting becomes closer to a classical non-staggered DiD design,
and therefore negative weights disappear, as our next result illustrates:
\begin{prop}
\label{prop:anynegweights}Suppose all units are observed for all
periods $t=1,\dots,T$ and the earliest treatment happens at $E_{\text{first}}>1$.
Let $N_{1}^{*}$ be the number of observations for never-treated units
before period $E_{\text{first}}$ and $N_{0}^{*}$ be the number of
untreated observations for ever-treated units since $E_{\text{first}}$.
Then there is no negative weighting, i.e. $\min_{it\in\Omega_{1}}w_{it}^{\text{static}}\ge0$,
if and only if $N_{1}^{*}\ge N_{0}^{*}$.\footnote{$N_{1}^{\ast}$ and $N_{0}^{\ast}$ respectively correspond to the
numbers of admissible and forbidden 2x2 DiD comparisons available
for the earliest-treated units in the latest period $T$. The gap
between them drives negative weights with complete panels, as in \textcite[Proposition 1]{Strezhnev2018}.}
\end{prop}
Even when weights are non-negative, they may remain highly unequal
and diverge from the estimands that the researcher is interested in.
Our preferred strategy is therefore to commit to the estimation target
and explicitly allow for treatment effect heterogeneity, except when
some form of \ref{assu:A3} is \emph{ex ante }appropriate.

\subsection{Spurious Identification of Long-Run Effects in Dynamic Specifications\label{subsec:Spurious}}

Another consequence of inappropriately imposing \ref{assu:A3} concerns
estimation of long-run causal effects. Conventional dynamic specifications
(except those subject to the underidentification problem) yield \emph{some
}estimates for all $\tau_{h}$ coefficients. Yet, for large enough
$h$, no averages of treatment effects are identified under \ref{assu:A1,assu:A2}
with unrestricted treatment effect heterogeneity. Therefore, estimates
from restrictive specifications are fully driven by unwarranted extrapolation
of treatment effects across observations and may not be reliable,
unless strong \emph{ex ante }reasons for \ref{assu:A3} exist.

This issue is well illustrated in the example of \ref{prop:example-static}.
To identify the long-run effect $\tau_{A3}$ under \ref{assu:A1,assu:A2},
one needs to form an admissible DiD comparison, of the outcome growth
over some period between unit $A$ and another unit not yet treated
in period 3. However, by period 3 both units have been treated. Mechanically,
this problem arises because the period fixed effect $\beta_{3}$ is
not identified separately from the treatment effects $\tau_{A3}$
and $\tau_{B3}$ in this example, absent restrictions on treatment
effects. Yet, the semi-dynamic specification
\[
Y_{it}=\tilde{\alpha}_{i}+\tilde{\beta}_{t}+\tau_{0}\one\left[K_{it}=0\right]+\tau_{1}\one\left[K_{it}=1\right]+\tilde{\varepsilon}_{it}
\]
will produce an estimate $\hat{\tau}_{1}$ via extrapolation. Specifically,
two different parameters, $\tau_{A3}-\tau_{B3}$ and $\tau_{A2}$,
are identified by comparing the two units in periods 2 or 3, respectively,
with period 1. Therefore, when imposing homogeneity of short-run effects
across units, $\tau_{A2}=\tau_{B3}\equiv\tau_{0}$, we estimate the
long-run effect $\tau_{A3}\equiv\tau_{1}$ as the sum of $\tau_{1}-\tau_{0}$
and $\tau_{0}$:
\[
\hat{\tau}_{1}=\left[\left(Y_{A3}-Y_{A1}\right)-\left(Y_{B3}-Y_{B1}\right)\right]+\left[\left(Y_{A2}-Y_{A1}\right)-\left(Y_{B2}-Y_{B1}\right)\right].
\]
However, when $\tau_{A2}\ne\tau_{B3}$, this estimator is biased.

In general, the gap between the earliest and the latest event times
observed in the data provides an upper bound on the number of dynamic
coefficients that can be identified without extrapolation of treatment
effects. This result, which follows by the same logic of non-identification
of the later period effects, is formalized by our next proposition:
\begin{prop}
\label{prop:No-Longrun}Suppose there are no never-treated units and
let $\bar{H}=\max_{i}E_{i}-\min_{i}E_{i}.$ Then, for any non-negative
weights $w_{it}$ defined over the set of observations with $K_{it}\ge\bar{H}$
(that are not identically zero), the weighted sum of causal effects
$\sum_{it\colon K_{it}\ge\bar{H}}w_{it}\tau_{it}$ is not identified
by \ref{assu:A1,assu:A2}.\footnote{The requirement that the weights are non-negative rules out some estimands
on the \emph{gaps }between treatment effects for $K_{it}\ge\bar{H}$
which are in fact identified. For instance, adding period $t=4$ to
the \ref{tab:twobythree} example, the difference $\tau_{A4}-\tau_{B4}$
would be identified (by $\left(Y_{A4}-Y_{B4}\right)-\left(Y_{A1}-Y_{B1}\right)$),
even though neither $\tau_{A4}$ nor $\tau_{B4}$ is identified.}
\end{prop}
\noindent Robust estimators, including the one we characterize in
\ref{sec:Imputation-Solution}, can only be computed for identified
estimands, never resulting in spurious estimates.

We finally note that the challenges described in \ref{sec:Conventional-Practice}
apply even if the sample is ``trimmed'' to a fixed window around
the event time; see \ref{subsec:Trimming}.

\section{Imputation-Based Estimation and Testing\label{sec:Imputation-Solution}}

To overcome the challenges affecting conventional practice, we now
derive the robust and efficient estimator and show that it takes a
particularly transparent ``imputation'' form when no restrictions
on treatment-effect heterogeneity are imposed. We then perform asymptotic
analysis, establishing the conditions for the estimator to be consistent
and asymptotically normal, derive conservative standard error estimates
for it, and discuss appropriate pre-trend tests.

Throughout, we continue to suppose that the researcher chose the estimation
target $\tau_{w}$ and assumed a model of $Y_{it}(0)$ (\ref{assu:A1prime})
and no anticipation. Some model of treatment effects (\ref{assu:A3})
may also be assumed, although our main focus is on the null model,
under which treatment effect heterogeneity is unrestricted. Letting
$\varepsilon_{it}=Y_{it}-\expec{Y_{it}}$ for $it\in\Omega$, we thus
have under \ref{assu:A1prime,assu:A2,assu:A3prime}:
\begin{equation}
Y_{it}=A_{it}^{\prime}\lambda_{i}+X_{it}^{\prime}\delta+D_{it}\Gamma_{it}^{\prime}\theta+\varepsilon_{it}.\label{eq:reduced-model}
\end{equation}
We assume throughout that $\tau_{w}=w_{1}'\Gamma\theta$ is identified.
\ref{prop:identification} provides conditions for identification.
First, we provide a general rank condition on the matrices of unit-specific
and other covariates that requires that the covariate space of treated
observations is spanned by that of the untreated ones. This assumption
allows us to estimate from the untreated observations those nuisance
parameters that are necessary to impute control outcomes of the treated
observations, thus providing identification. Second, we derive specific
conditions for the case where the parameter $\delta$ represents time
fixed effects and $A_{it}$ may vary over time but not across units
(as with unit FEs and unit-specific linear trends). In this case,
we show that there is identification if (i) the $A_{it}$ are not
collinear for any relevant unit and (ii) there is at least one untreated
unit at the end of the time period of interest.

\subsection{Efficient Estimation}

For our efficiency result, we impose an additional assumption on the
error variances:
\begin{assumption}[Spherical errors]
\label{assu:homosk}Error terms $\varepsilon_{it}$ are spherical,
i.e. homoskedastic and mutually uncorrelated across all $it\in\Omega$:
$\expec{\varepsilon\varepsilon^{\prime}}=\sigma^{2}\mathbb{I}_{N}$.
\end{assumption}
While this assumption is strong, our efficiency results also apply
without change under dependence that is due to unit random effects,
i.e. if $\varepsilon_{it}=\eta_{i}+\tilde{\varepsilon}_{it}$ for
$\tilde{\varepsilon}_{it}$ that satisfy \ref{assu:homosk} and for
some $\eta_{i}$. Moreover, these results are straightforward to relax
to any known form of heteroskedasticity or mutual dependence.\footnote{For instance, if error terms are uncorrelated and have known variances
$\sigma_{it}^{2}$ (up to a scaling factor), efficiency requires the
estimation step (Step 1) of \ref{thm:OLS-BLUE,thm:imputation} to
be performed with weights proportional to $\sigma_{it}^{-2}$. One
example for this is when the data are aggregated from $n_{it}$ individuals
randomly drawn from group $i$ in period $t$ and spherical individual-level
errors, in which case efficiency is obtained with weights proportional
to $n_{it}$.} Under \ref{assu:homosk} and allowing for restrictions on causal
effects, we have:
\begin{thm}[Efficient estimator]
\label{thm:OLS-BLUE}Suppose \ref{assu:A1prime,assu:A2,assu:A3prime,assu:homosk}
hold. Then among the linear unbiased estimators of $\tau_{w}$, the
(unique) efficient estimator $\hat{\tau}_{w}^{\ast}$ can be obtained
with the following steps:
\begin{enumerate}
\item Estimate $\theta$ by the OLS solution $\hat{\theta}^{*}$ from the
regression \ref{eq:reduced-model} (where we assume that $\theta$
is identified);
\item Estimate the vector of treatment effects $\tau$ by $\hat{\tau}^{*}=\Gamma\hat{\theta}^{*}$;
\item Estimate the target $\tau_{w}$ by $\hat{\tau}_{w}^{\ast}=w_{1}^{\prime}\hat{\tau}^{*}$.
\end{enumerate}
Moreover, this estimator $\hat{\tau}_{w}^{*}$ is unbiased for $\tau_{w}$
under \ref{assu:A1prime,assu:A2,assu:A3prime} alone, even when error
terms are not spherical.
\end{thm}
Under \ref{assu:A1prime,assu:A2,assu:A3prime}, regression \ref{eq:reduced-model}
is correctly specified. Thus, this estimator for $\theta$ is unbiased
by construction, and efficiency under spherical error terms is a direct
consequence of the Gauss\textendash Markov theorem. Moreover, OLS
yields the most efficient estimator for any linear combination of
$\theta$, including $\tau_{w}=w_{1}^{\prime}\Gamma\theta$. While
assuming spherical errors may be unrealistic in practice, we think
of this assumption as a natural conceptual benchmark to decide between
the many unbiased estimators of $\tau_{w}$.\footnote{This benchmark appears natural as it parallels the Gauss-Markov theorem
which also relies on spherical errors. In Monte Carlo simulations
(\ref{subsec:Monte-Carlo-BP-Appx-NEW}), the estimator performs well
even under deviations from spherical errors. In \ref{subsec:appx-GLS}
we generalize the results to parametric models of heteroskedasticity
and serial correlation, in the spirit of generalized least squares
(GLS) and relating to \textcite{Wooldridge2021}.}

In the important special case of unrestricted treatment effect heterogeneity,
$\hat{\tau}_{w}^{*}$ has a useful ``imputation'' representation.
The idea is to estimate the model of $Y_{it}(0)$ using the untreated
observations $it\in\Omega_{0}$ and leverage it to impute $Y_{it}(0)$
for treated observations $it\in\Omega_{1}$. Then, observation-specific
causal effect estimates can be averaged appropriately. Perhaps surprisingly,
the estimation and imputation steps are identical regardless of the
target estimand. Applying any weights to the imputed causal effects
yields the efficient estimator for the corresponding estimand. We
have:
\begin{thm}[Imputation representation for the efficient estimator]
\label{thm:imputation}With a null \ref{assu:A3prime} (that is,
if $\Gamma=\I_{N_{1}}$), the unique efficient linear unbiased estimator
$\hat{\tau}_{w}^{\ast}$ of $\tau_{w}$ from \ref{thm:OLS-BLUE} can
be obtained via an imputation procedure:
\begin{enumerate}
\item Within the untreated observations only ($it\in\Omega_{0}$), estimate
the $\lambda_{i}$ and $\delta$ (by $\hat{\lambda}_{i}^{*},\hat{\delta}^{*}$)
by OLS in
\begin{equation}
Y_{it}=A_{it}^{\prime}\lambda_{i}+X_{it}^{\prime}\delta+\varepsilon_{it};\label{eq:imputation-regression}
\end{equation}
\item For each treated observation ($it\in\Omega_{1}$) with $w_{it}\ne0$,
set $\hat{Y}_{it}(0)=A_{it}^{\prime}\hat{\lambda}_{i}^{*}+X_{it}^{\prime}\hat{\delta}^{*}$
and $\hat{\tau}_{it}^{*}=Y_{it}-\hat{Y}_{it}(0)$ to obtain the estimate
of $\tau_{it}$;
\item Estimate the target $\tau_{w}$ by a weighted sum $\hat{\tau}_{w}^{\ast}=\sum_{it\in\Omega_{1}}w_{it}\hat{\tau}_{it}^{*}$.
\end{enumerate}
\end{thm}
The imputation representation offers computational and conceptual
benefits. First, it is computationally efficient as it only requires
estimating a simple TWFE model, for which fast algorithms are available
\parencite{Guimaraes2010,Correia2017}. This is in contrast to the
OLS estimator from \ref{thm:OLS-BLUE}, as equation \ref{eq:reduced-model}
has regressors $\Gamma_{it}D_{it}$ in addition to the fixed effects,
which are high-dimensional unless a low-dimensional model of treatment
effect heterogeneity is imposed.

Second, the imputation approach is intuitive and transparently links
the parallel trends and no-anticipation assumptions to the estimator.
Indeed, \textcite{imbens2015causal} write: \emph{``At some level,
all methods for causal inference can be viewed as imputation methods,
although some more explicitly than others'' }(p.~141). We formalize
this statement in the next proposition, which shows that \emph{any}
estimator unbiased for $\tau_{w}$ can be represented in the imputation
way, but the way of imputing the $Y_{it}(0)$ may be less explicit
and no longer efficient.
\begin{prop}[Imputation representation for all unbiased estimators]
 \label{prop:generalimputation}Under \ref{assu:A1prime,assu:A2},
any linear estimator $\hat{\tau}_{w}$ of $\tau_{w}$ that is unbiased
under arbitrary treatment-effect heterogeneity (that is, a null \ref{assu:A3})
can be obtained via imputation:
\begin{enumerate}
\item For every treated observation, estimate expected untreated potential
outcomes $A_{it}'\lambda_{i}+X_{it}'\delta$ by some unbiased linear
estimator $\hat{Y}_{it}(0)$ using data from the untreated observations
only;
\item For each treated observation, set $\hat{\tau}_{it}=Y_{it}-\hat{Y}_{it}(0)$;
\item Estimate the target by a weighted sum $\hat{\tau}_{w}=\sum_{it\in\Omega_{1}}w_{it}\hat{\tau}_{it}$.
\end{enumerate}
\end{prop}
This result establishes an imputation representation when treatment
effects can vary arbitrarily. \ref{prop:model-imputation} in the
appendix establishes that the imputation structure applies even when
restrictions $\tau=\Gamma\theta$ are imposed, albeit with an additional
step in which the weights $w_{1}$ defining the estimand are adjusted
in a way that does not change $\tau_{w}$ under the imposed model.\footnote{As a special case, we can still write the efficient estimator from
\ref{thm:OLS-BLUE} as an imputation estimator from \ref{thm:imputation}
with alternative weights $v_{1}^{*}$ on the imputed treatment effects.
\ref{prop:imputation-model} shows that these adjusted weights $v_{1}^{*}$
solve a quadratic variance-minimization problem with a linear constraint
that preserves unbiasedness under \ref{assu:A3}. We also provide
an explicit formula for the resulting weights in \ref{prop:olsweights}.} In this sense, unbiased causal inference is equivalent to imputation
in our framework.

\subsection{Asymptotic Properties\label{subsec:Asymptotic-Properties}}

Having derived the linear unbiased estimator $\hat{\tau}_{w}^{\ast}$
for $\tau_{w}$ in \ref{thm:OLS-BLUE} that is also efficient under
spherical error terms, we now consider its asymptotic properties without
imposing that assumption. We study convergence along a sequence of
panels indexed by the sample size $N$, where randomness stems from
the error terms $\varepsilon_{it}$ only, as in \ref{sec:Setting}.
Our approach applies to asymptotic sequences where both the number
of units and the number of time periods may grow, but the assumptions
are least restrictive when the number of time periods remains constant
or grows slowly, as in short panels.

Instead of assuming that error terms are spherical, we now assume
that error terms are clustered by units $i$.
\begin{assumption}[Clustered error terms]
\label{assu:clustered}Error terms $\varepsilon_{it}$ are independent
across units $i$ and have bounded variance, $\var{\varepsilon_{it}}\leq\bar{\sigma}^{2}$
for $it\in\Omega$ uniformly.
\end{assumption}
The key role in our results is played by the weights that the \ref{thm:OLS-BLUE}
estimator places on each observation. Since the estimator is linear
in the observed outcomes $Y_{it}$, we can write it as $\hat{\tau}_{w}^{*}=\sum_{it\in\Omega}v_{it}^{*}Y_{it}$
with non-stochastic weights $v_{it}^{*}$, derived in \ref{prop:olsweights}
in the appendix.

We now formulate high-level conditions on the sequence of weight vectors
that ensure consistency, asymptotic normality, and will later allow
us to provide valid inference. These results apply to any unbiased
linear estimator $\hat{\tau}_{w}=\sum_{it\in\Omega}v_{it}Y_{it}$
of $\tau_{w}$, not just the efficient estimator $\hat{\tau}_{w}^{*}$
from \ref{thm:OLS-BLUE} \textendash{} that is, if the respective
conditions are fulfilled for the weights $v_{it}$, then consistency,
asymptotic normality, and valid inference follow as stated. For the
specific estimator $\hat{\tau}_{w}^{*}$ introduced above, we then
provide sufficient low-level conditions for short panels.

First, we obtain consistency of $\hat{\tau}_{w}$ under a Herfindahl
condition on the weights $v$ that takes the clustering structure
of error terms into account.

\begin{assumption}[Herfindahl condition]
\label{assu:Herfindahl}Along the asymptotic sequence, $\wnorm^{2}\equiv\sum_{i}\left(\sum_{t;it\in\Omega}\left|v_{it}\right|\right)^{2}\rightarrow0$,
for weights $v_{it}$ in the unbiased linear estimator $\hat{\tau}_{w}=\sum_{it\in\Omega}v_{it}Y_{it}$.
\end{assumption}
The condition on the clustered Herfindahl index $\wnorm^{2}$ states
that the sum of squared weights vanishes, where weights are aggregated
by units. One can think of the inverse of the sum of squared weights,
$n_{H}=\wnorm^{-2}$, as a measure of effective sample size, which
\ref{assu:Herfindahl} requires to grow large along the asymptotic
sequence. If it is satisfied, and variances are uniformly bounded,
we obtain consistency of $\hat{\tau}_{w}$:\footnote{\label{fn:R-criterion}The Herfindahl condition can be restrictive
since it allows for a worst-case correlation of error terms within
units. When such correlations are limited, other sufficient conditions
may be more appropriate instead, such as $R\left(\sum_{it\in\Omega}v_{it}^{2}\right)\rightarrow0$
with $R=\max_{i}\left(\text{largest eigenvalue of \ensuremath{\Sigma_{i}}}\right)/\bar{\sigma}^{2}$,
where $\Sigma_{i}=\left(\cov{\varepsilon_{it},\varepsilon_{is}}\right)_{t,s}$.
Here $R$ is a measure of the maximal joint covariation of all observations
for one unit. If error terms are uncorrelated, then $R\leq1$, since
the maximal eigenvalue of $\Sigma_{i}$ corresponds to the maximal
variance of an error term $\varepsilon_{it}$ in this case, which
is bounded by $\bar{\sigma}^{2}$. An upper bound for $R$ is the
maximal number of periods for which we observe a unit, since the maximal
eigenvalue of $\Sigma_{i}$ is bounded by the sum of the variances
on its diagonal.}
\begin{prop}[Consistency of $\hat{\tau}_{w}$]
\label{prop:consistency}Under \ref{assu:A1prime,assu:A2,assu:A3prime,assu:clustered,assu:Herfindahl},
$\hat{\tau}_{w}-\tau_{w}\stackrel{\mathcal{L}_{2}}{\rightarrow}0$
for an unbiased linear estimator $\hat{\tau}_{w}$ of $\tau_{w}$,
such as $\hat{\tau}_{w}^{*}$ in \ref{thm:OLS-BLUE}.
\end{prop}
We note that the large number of unit-specific parameters $\left\{ \lambda_{i}\right\} _{i}$
cannot generally be estimated consistently in panels with a small
number of time periods, raising a potential incidental-parameters
problem. However, consistency of $\hat{\tau}_{w}$ does not rely on
consistency for unit-specific parameters, since our estimator averages
over many units.

We next consider the asymptotic distribution of the estimator around
the estimand.
\begin{prop}[Asymptotic normality]
\label{prop:asymptoticnormality}If the assumptions of \ref{prop:consistency}
hold, there exists $\kappa>0$ such that $\E\left[\left|\varepsilon_{it}\right|^{2+\kappa}\right]$
is uniformly bounded, the weights are not too concentrated in the
sense that $\sum_{i}\left(\sqrt{n_{H}}\sum_{t;it\in\Omega}\left|v_{it}\right|\right)^{2+\kappa}\rightarrow0$,
and the variance does not vanish, $\liminf n_{H}\sigma_{w}^{2}>0$
for $\sigma_{w}^{2}=\var{\hat{\tau}_{w}}$, then we have that $\sigma_{w}^{-1}(\hat{\tau}_{w}-\tau_{w})\stackrel{d}{\rightarrow}\N(0,1)$.
\end{prop}
This result establishes conditions under which the difference between
estimator and estimand is asymptotically normal. Besides regularity,
this proposition requires that the estimator variance $\sigma_{w}^{2}$
does not decline faster than $1/n_{H}$. It is violated if the clustered
Herfindahl formula is too conservative: for instance, if the number
of periods is growing along the asymptotic sequence while the within-unit
over-time correlation of error terms remains small. Alternative sufficient
conditions for asymptotic normality can be established in such cases,
e.g. along the lines of \ref{fn:R-criterion}.

So far, we have formulated high-level conditions on the weights $v_{it}$
of any linear unbiased estimator of $\tau_{w}$. \ref{subsec:appx-sufficient}
presents low-level sufficient conditions for consistency and asymptotic
normality of the imputation estimator $\hat{\tau}_{w}^{*}$ for the
benchmark case of a panel with unit and period FEs, a fixed or slowly
growing number of periods, and no restrictions on treatment effects.
Unlike \ref{prop:consistency,prop:asymptoticnormality}, these conditions
are imposed directly on the weights $w_{1}$ chosen by the researcher,
and not on the the implied weights $v_{it}^{\ast}$, such that the
researcher can assess more directly whether the asymptotic approximation
is likely to be precise. In particular, the estimator achieves consistency
and asymptotic normality in the common case where the number of time
periods is fixed, the size of all cohorts increases, the weights on
treatment effects do not vary within the same period and cohort, and
the sum of (absolute) weights is bounded. In addition, the sufficient
conditions are also fulfilled when the number of periods grows slowly
and when weights differ across observations within the same cohort
and period, but not by too much. With covariates other than unit and
period FEs, e.g. with unit-specific linear trends, the general weight
conditions in \ref{assu:Herfindahl} and \ref{prop:asymptoticnormality}
can also be used to verify consistency and asymptotic normality. In
those cases, the sufficient conditions are typically fulfilled for
convex combinations of cohort-average treatment effects whenever the
size of cohorts grows sufficiently fast relatively to the number of
periods (see \ref{subsec:appx-sufficient}).

\subsection{\label{subsec:Conservative-Inference}Conservative Inference}

We next estimate the variance of $\hat{\tau}_{w}=\sum_{it\in\Omega}v_{it}Y_{it}$,
which equals $\sigma_{w}^{2}=\expec{\sum_{i}\left(\sum_{t;it\in\Omega}v_{it}\varepsilon_{it}\right)^{2}}$
with clustered error terms (\ref{assu:clustered}). We start with
the case where treatment effect heterogeneity is unrestricted (i.e.
$\Gamma=\I$). \textit{\emph{As in \textcite{DeChaisemartin2018},
exact inference becomes infeasible when treatment effects are heterogeneous}},
but conservative inference is possible. Following \ref{subsec:Asymptotic-Properties},
the inference tools we propose apply to a generic linear unbiased
estimator but we use them for the efficient estimator $\hat{\tau}_{w}^{\ast}$.
Our strategy is to estimate individual error terms by some $\tilde{\varepsilon}_{it}$
and then use a plug-in estimator,
\begin{align}
\hat{\sigma}_{w}^{2} & =\sum_{i}\left(\sum_{t;it\in\Omega}v_{it}\tilde{\varepsilon}_{it}\right)^{2}.\label{eq:plugin}
\end{align}

Estimating the error terms presents two challenges, which become apparent
when we consider the benchmark choice $\tilde{\varepsilon}_{it}=\hat{\varepsilon}_{it}$
based on the regression residuals $\hat{\varepsilon}_{it}=Y_{it}-A_{it}^{\prime}\hat{\lambda}_{i}^{*}-X_{it}^{\prime}\hat{\delta}^{*}-D_{it}\hat{\tau}_{it}^{*}$
in the regression \ref{eq:reduced-model}. The first challenge is
the incidental-parameters problem in estimating $\lambda_{i}$. However,
by using cluster-robust variance estimates, our inference does not
suffer from this problem since the variance estimator $\hat{\sigma}_{w}^{2}$
does not rely on the consistent estimation of $\lambda_{i}$ any more,
similar to the insight of \textcite{Stock2008a}.

A second challenge arises from unrestricted treatment-effect heterogeneity.
In \ref{thm:imputation}, treatment effects are estimated by fitting
the corresponding outcomes $Y_{it}$ perfectly, with residuals $\hat{\varepsilon}_{it}\equiv0$
for all treated observations. This issue is not specific to our estimation
procedure: one generally cannot distinguish between $\tau_{it}$ and
$\varepsilon_{it}$ from observations of $Y_{it}=A_{i}^{\prime}\lambda_{i}+X_{it}^{\prime}\delta+\tau_{it}+\varepsilon_{it}$
for treated observations, making it impossible to produce unbiased
estimates of $\sigma_{w}^{2}$ (see Lemma 1 in \textcite{Kline2020}
for a similar impossibility result).

While unbiased estimation of $\sigma_{w}^{2}$ is not possible, we
show that this variance can be estimated conservatively. Our variance
estimator is based on an auxiliary parsimonious model of treatment
effects. We do not require this model to be correct, in the sense
that inference is weakly asymptotically conservative under misspecification.
However, auxiliary models which better approximate $\tau_{it}$ will
make confidence intervals tighter and closer to asymptotically exact.
In the computation of $\hat{\sigma}_{w}^{2}$ we set $\tilde{\varepsilon}_{it}$
for the treated observations equal to the residuals of the auxiliary
model. We require the model to be parsimonious, such that it does
not overfit and the residuals include $\varepsilon_{it}$. When the
model is incorrect, $\tilde{\varepsilon}_{it}$ also include a component
due to the misspecification of $\tau_{it}$, leading to conservative
inference.

We formalize the auxiliary model by considering estimators $\tilde{\tau}_{it}$
for each $it\in\Omega_{1}$ which satisfy two properties: (1) $\tilde{\tau}_{it}$
converges to \emph{some }non-stochastic limit $\bar{\tau}_{it}$ and
(2) if the auxiliary model is correct, $\bar{\tau}_{it}=\tau_{it}$.
The following theorem presents conditions under which our construction
yields asymptotically conservative inference:
\begin{thm}[Conservative clustered standard error estimates]
\label{thm:se}Assume that the assumptions of \ref{prop:consistency}
hold, that the model of treatment effects is trivial ($\Gamma=\I$),
that the estimates $\tilde{\tau}_{it}$ converge to some non-random
$\bar{\tau}_{it}$ in the sense that $\wnorm^{-2}\sum_{i}\left(\sum_{t;it\in\Omega_{1}}v_{it}(\tilde{\tau}_{it}-\bar{\tau}_{it})\right)^{2}\stackrel{p}{\rightarrow}0$,
that $\hat{\delta}^{*}$ from \ref{thm:OLS-BLUE} is sufficiently
close to $\delta$ in the sense that $\wnorm^{-2}\sum_{i}\left(\sum_{t;it\in\Omega}v_{it}X_{it}'(\hat{\delta}^{\ast}-\delta)\right)^{2}\stackrel{p}{\rightarrow}0$,
and that $\left|\tau_{it}\right|,$ $\left|\bar{\tau}_{it}\right|$
and $\expec{\varepsilon_{it}^{4}}$ are uniformly bounded and the
weights are not too concentrated in the sense that $\sum_{i}\left(\frac{\sum_{t;it\in\Omega}\lvert v_{it}\rvert}{\wnorm}\right)^{4}\rightarrow0$.
Then the variance estimate
\begin{align}
\hat{\sigma}_{w}^{2} & =\sum_{i}\left(\sum_{t;it\in\Omega}v_{it}\tilde{\varepsilon}_{it}\right)^{2}, & \tilde{\varepsilon}_{it} & =Y_{it}-A_{it}^{\prime}\hat{\lambda}_{i}^{*}-X_{it}^{\prime}\hat{\delta}^{*}-D_{it}\tilde{\tau}_{it}\label{eq:SE}
\end{align}
is asymptotically conservative: $\wnorm^{-2}(\hat{\sigma}_{w}^{2}-\sigma_{w}^{2}-\sigma_{\tau}^{2})\stackrel{p}{\rightarrow}0$
where $\sigma_{\tau}^{2}=\sum_{i}\left(\sum_{t;D_{it}=1}v_{it}(\tau_{it}-\bar{\tau}_{it})\right)^{2}\geq0$.
If $\bar{\tau}_{it}=\tau_{it}$ for all $it\in\Omega_{1}$, $\sigma_{\tau}^{2}=0$,
meaning that the variance estimate is asymptotically exact.
\end{thm}
The theorem shows that the proposed variance estimate addresses the
two challenges laid out above. First, the estimates remain valid even
though we may not be able to estimate the unit-specific parameters
$\lambda_{i}$ consistently. This is because unit-specific parameter
estimates drop out when summing over all observations of one unit
in \ref{eq:SE}, as shown in the proof. Second, by using estimates
$\tilde{\tau}_{it}$ that fulfill the convergence condition of the
theorem, we avoid the issue of obtaining trivial residuals for the
treated observations. The resulting variance estimates are asymptotically
conservative. From these estimates we can also obtain conservative
confidence intervals (that asymptotically have coverage that is at
least nominal) if the estimator is also asymptotically normal, such
as under the sufficient conditions of \ref{prop:asymptoticnormality}.

It remains to choose the estimates $\tilde{\tau}_{it}$. We focus
on auxiliary models that impose the equality of treatment effects
across large groups of treated observations: for a partition $\Omega_{1}=\bigcup_{g}G_{g}$,
$\tau_{it}\equiv\tau_{g}$ for all $it\in G_{g}$. The $\tau_{g}$
can then be estimated by some weighted average of $\hat{\tau}_{it}^{*}$
among $it\in G_{g}$. Specifically, we propose averages of the form
\begin{equation}
\tilde{\tau}_{g}=\frac{\sum_{i}\left(\sum_{t;it\in G_{g}}v_{it}\right)\left(\sum_{t;it\in G_{g}}v_{it}\hat{\tau}_{it}^{*}\right)}{\sum_{i}\left(\sum_{t;it\in G_{g}}v_{it}\right)^{2}}.\label{eq:taubar-G}
\end{equation}
In \ref{subsec:appx-SE-Weights}, we show that this choice of weights
leads to minimal excess variance $\sigma_{\tau}^{2}$ in the case
where there is only a single group $g$, corresponding to a conservative
auxiliary model which requires all treatment effects to be the same.
The choice of the partition aims to maintain a balance between avoiding
overly conservative variance estimates and ensuring consistency. If
the sample is large enough, one may want to partition $\Omega_{1}$
into multiple groups of observations such that treatment effect heterogeneity
is expected to be smaller within them than across. For instance, with
many units, a group may consist of observations corresponding to the
same horizon relative to treatment onset. If cohorts are large, one
can further partition observations into groups defined by cohort and
period, which we use as the default in our Stata command.

While sufficiently large groups in \ref{eq:taubar-G} avoid overfitting
asymptotically (under appropriate conditions), in finite samples these
$\tilde{\tau}_{it}$ still use $\hat{\tau}_{it}^{*}$ and thus partially
overfit to $\varepsilon_{it}$. In \ref{subsec:appx-Leave-Out} we
therefore also consider leave-out versions of these $\tilde{\tau}_{it}$.

We make four final remarks on \ref{thm:se}. First, our strategy for
estimating the variance extends directly to conservative estimation
of variance-covariance matrices for vector-valued estimands, e.g.
for average treatment effects at multiple horizons $h$. Second, the
result applies in short panels under the low-level conditions of \ref{subsec:appx-sufficient}
(see \ref{prop:se-short-sufficient}). Third, while we have focused
here on the case of unrestricted heterogeneity ($\Gamma=\I$), \ref{thm:se}
can be extended to the case with a non-trivial treatment-effect model
imposed in \ref{assu:A3}.\footnote{By \ref{prop:model-imputation}, the general efficient estimator can
be represented as an imputation estimator for a modified estimand,
i.e. by changing $w_{1}$ to some $v_{1}$. \ref{thm:se} then yields
a conservative variance estimate for it. We note that under sufficiently
strong restrictions on treatment effects, asymptotically exact inference
may be possible, as the residuals $\hat{\varepsilon}_{it}$ in \ref{eq:reduced-model}
may be estimated consistently even for treated observations (except
for the inconsequential noise in $\hat{\lambda}_{i}$), alleviating
the need for an additional auxiliary model.} Finally, computation of $\hat{\sigma}_{w}^{2}$ for the estimator
$\hat{\tau}_{w}^{*}$ from \ref{thm:OLS-BLUE} involves the implied
weights $v_{it}^{*}$, which becomes computationally challenging with
multiple sets of high-dimensional FEs. In \ref{subsec:appx-Weights-Imputation}
we develop a computationally efficient algorithm for computing $v_{it}^{*}$
based on the iterative least squares algorithm for conventional regression
coefficients \parencite{Guimaraes2010}.

\subsection{Testing for Parallel Trends\label{subsec:Testing-PTA}}

In this section, we discuss testing the (generalized) parallel-trend
and no-anticipation assumptions \ref{assu:A1prime,assu:A2}. We propose
a testing procedure based on OLS regressions with untreated observations
only, departing from both traditional regression-based tests and
more recent placebo tests. This procedure is robust to treatment effect
heterogeneity and, under spherical errors, has attractive power properties
and avoids the problem of inference after pre-testing explained by
\textcite{Roth2018a}. We propose:

\begin{test}\label{test:pretrends}\emph{(Robust OLS-based pre-trend
test)}
\begin{enumerate}
\item Choose an alternative model for $Y_{it}$ for untreated observations
$it\in\Omega_{0}$ that is richer than that imposed by \ref{assu:A1prime,assu:A2}:
for an observable vector $W_{it}$ (which we consider non-stochastic,
like $A_{it}$ and $X_{it}$),
\begin{equation}
Y_{it}=A_{it}^{\prime}\lambda_{i}+X_{it}^{\prime}\delta+W_{it}^{\prime}\gamma+\varepsilon_{it};\label{eq:richer-alt}
\end{equation}
\item Estimate $\gamma$ by $\hat{\gamma}$ in \ref{eq:richer-alt} using
OLS on untreated observations only;
\item Test $\gamma=0$ using the heteroskedasticity- and cluster-robust
Wald\emph{ }test.
\end{enumerate}
\end{test}

\noindent This test is valid because equation \ref{eq:richer-alt}
is implied by \ref{assu:A1prime,assu:A2} if the null $\gamma=0$
holds.\footnote{There is a natural alternative test of the null $\gamma=0$ in the
model \ref{eq:richer-alt}, namely the Hausman test based on the difference
between the imputation estimator $\hat{\tau}_{w}^{W}$ based on the
model in \ref{eq:richer-alt} and the efficient imputation estimator
$\hat{\tau}_{w}^{*}$ that is only valid when $\gamma=0$. Like our
test, this test only uses untreated observations and avoids the \textcite{Roth2018a}
pre-testing problem for spherical errors. The Hausman approach has
the advantage of quantifying the magnitude of bias from omitting $W_{it}$,
while \ref{test:pretrends} has the advantage that it is also informative
about violations that cancel out in the Hausman test. The two tests
are equivalent for scalar $W_{it}$.}

The test requires choosing $W_{it}$ to parametrize the possible violation
of \ref{assu:A1prime,assu:A2}. A natural choice for $W_{it}$, which
parallels conventional pre-trend tests, is a set of indicators for
observations $1,\dots,k$ periods before the onset of treatment for
some $k$, with periods before $E_{i}-k$ serving as the reference
group.\footnote{The optimal choice of $k$ is a challenging question. As usual with
Wald tests, choosing a $k$ that is too large can lead to low power
against many alternatives, in particular those that generate large
biases in treatment effect estimates that impose invalid \ref{assu:A1prime}.} This choice is appropriate, for instance, if the researcher's main
worry is the possible effects of treatment anticipation, i.e. violations
of \ref{assu:A2}. This choice of $W_{it}$ also lends itself to making
``event study plots,'' which combine the ATT estimates by horizon
$h\ge0$ with a series of pre-trend coefficients; we supply the \texttt{event\_plot}
Stata command for this goal. Alternatively, the researcher may focus
on possible violations of \ref{assu:A1prime}. For instance, with
data spanning many years one could test for the presence of a structural
break in unit FEs.

\ref{test:pretrends} can be contrasted with two existing strategies
to test parallel trends. Traditionally, researchers estimated a dynamic
specification including lags and leads of treatment onset, and tested
\textemdash{} visually or statistically \textemdash{} that the coefficients
on leads are equal to zero. More recent papers \parencite[e.g.][]{DeChaisemartin2018,Liu2020a}
replace it with a placebo strategy: pretend that treatment happened
$k$ periods earlier for all eventually treated units, and estimate
the average effects $h=0,\dots,k-1$ periods after the placebo treatment
using the same estimator as for actual estimation. 

Both of these alternatives strategies have drawbacks. Because the
traditional regression-based test uses the full sample, including
treated observations, and imposes restrictions on treatment effects
(which are assumed homogeneous within each horizon), it is \emph{not
}a test for \ref{assu:A1prime,assu:A2} only. Rather, it is a joint
test that is sensitive to violations of the implicit \ref{assu:A3}
\parencite{Abraham2018}. Even if a researcher has reasons to impose
a non-trivial \ref{assu:A3} in estimation, a robust test for parallel
trends and no anticipation \emph{per se }should avoid those restrictions
on treatment effect heterogeneity. With a null \ref{assu:A3}, treated
observations are not useful for testing, and our test only uses the
untreated ones.\footnote{\textcite{Wooldridge2021} shows that as long as treatment effects
are allowed to vary flexibly, tests based on specifications estimated
on the full sample do not use treated observations. Therefore, such
tests are also not contaminated by treatment effect heterogeneity.}

Tests based on placebo estimates appropriately use untreated observations
only and may have intuitive appeal. However, mimicking the estimator
does not generally correspond to an efficient test of a class of plausible
alternatives. In contrast, \ref{test:pretrends} possesses well-known
asymptotic efficiency properties when $W_{it}$ is correctly specified.
For example, when $\varepsilon_{it}$ are spherical and normal, it
is asymptotically equivalent to the homoskedastic $F$-test, which
is a uniformly most powerful invariant test \parencite[ch. 7.6]{lehmann2006testing}.

Finally, we show an additional advantage of \ref{test:pretrends}:
if the researcher conditions on the test passing (i.e. does not report
the results otherwise), inference on $\hat{\tau}_{w}^{\ast}$ is still
asymptotically valid under the null of no violations of \ref{assu:A1prime,assu:A2}
and under spherical errors. This avoids the issue pointed out by \textcite[Proposition 4]{Roth2018a}
in the context of restrictive dynamic event study regressions: that
variance estimates which do not take pre-testing into account are
inflated, leading to unnecessarily conservative inference.\footnote{\textcite{Roth2018a} points out another issue, which we also avoid
under the assumptions of \ref{prop:pretesting}: when pre-trend and
treatment effect estimators are correlated, the bias arising from
violations of \ref{assu:A1prime,assu:A2} is affected by pre-testing;
it is exacerbated in specific cases \parencite[Proposition 2]{Roth2018a}.}
\begin{prop}[Pre-test robustness]
\label{prop:pretesting}Suppose the model in \ref{eq:richer-alt}
and \ref{assu:homosk} hold. Then $\hat{\tau}_{w}^{\ast}$ constructed
as in \ref{thm:OLS-BLUE} is uncorrelated with any vector $\hat{\gamma}$
constructed as in \ref{test:pretrends}. If the error terms are also
normally distributed,\footnote{The normality assumption is not essential; in the proof, we show that
a similar, asymptotic result holds generally under regularity conditions.} then $\hat{\tau}_{w}^{\ast}$ and $\hat{\gamma}$ are independent,
and inference on $\tau_{w}$ based on $\hat{\tau}_{w}^{\ast}$ is
unaffected by pre-tests based on $\hat{\gamma}$.\footnote{An early version of \textcite{Roth2018a} shows how to construct
an adjustment that removes the dependence when it exists, provided
the covariance matrix between $\hat{\tau}^{\ast}$ and $\hat{\gamma}$
can be estimated. By \ref{prop:pretesting}, this adjustment is not
needed for the \ref{thm:OLS-BLUE} estimator under spherical errors.}
\end{prop}

\section{Application\label{sec:Application}}

Having derived the attractive theoretical properties of the imputation
estimator, we now illustrate their practical relevance by revisiting
the estimation of the marginal propensity to spend in the event study
of \textcite{Broda2014}. We also use this empirical setting to verify
the properties of the imputation estimator in a simulation study.

\subsection{Setting}

The marginal propensity to spend out of tax rebates is a crucial parameter
for economic policy. In the US, the Economic Stimulus Act of 2008
consisted primarily of a 100 billion dollar program that sent tax
rebates to approximately 130 million tax filers. \textcite{Parker2013}
and \textcite{Broda2014} estimate the marginal propensity to spend
(MPX) out of the 2008 tax rebates. The rebate was disbursed using
two methods: either via direct deposit to a bank account, if known
by the IRS, or with a mailed paper check. For each method, the week
in which the funds were disbursed depended on the second-to-last digit
of the taxpayer's Social Security number (SSN). This number provides
a source of quasi-experimental variation because the last four digits
of a SSN are assigned sequentially to applicants within geographic
areas.

\textcite[henceforth BP]{Broda2014} use an event study design to
examine the response of nondurable spending to tax rebate receipt,
leveraging the quasi-experimental variation in the \emph{timing} of
the receipt. The quasi-random assignment of the last digits of the
SSN makes the parallel-trends assumption for expenditures \emph{a
priori} plausible.\footnote{\textcite{Thakral2020} point out that for the paper check group pre-rebate
household characteristics (in \emph{levels}) are not balanced with
respect to the timing of the receipt. While this is problematic for
randomization-based approaches to DiD (e.g. \textcite{Arkhangelsky2019}
and \textcite{Roth2021}), parallel \emph{trends }in expenditures
may still hold. Indeed, we fail to reject them with pre-trend tests
below.} The no-anticipation assumption may also be expected to hold: although
the disbursement schedule was known in advance, households were directly
notified by mail only several days before disbursement.

We estimate the performance of various estimators at estimating the
impulse response function of nondurable spending to tax rebate receipt
using the same data as BP. While earlier work by \textcite{Parker2013}
estimates the impulse responses using quarterly spending data from
the Consumer Expenditure Survey, BP leverage more detailed data from
the Nielsen Homescan Consumer Panel. The Nielsen dataset tracks transactions
at a much higher (in principle, daily) frequency, which is why we
choose it for our analysis. The Nielsen data cover expenditures on
consumer packaged goods (food, beverages, beauty and health products,
household supplies, and general merchandise), representing around
15\% of total household expenditures. Our dataset, identical to that
of BP, is a complete panel of 21,760 households (including 21,690
with non-missing disbursement method information) observed over 52
weeks of year 2008.

\subsection{Comparison between Robust and Conventional Estimates\label{subsec:BP-robust-and-OLS}}

We show how BP's estimates of the MPX suffer from an upward bias
in the short-run due to the choice of a binned specification (\ref{subsec:Biased-Weighting-with-Binning})
and how they may be spurious in the long-run (\ref{subsec:Spurious-Identification-Application}).
In \ref{subsec:Implications} we present our preferred robust estimates
and discuss implications for the macroeconomics literature.

\subsubsection{Negative Weighting and Upward Bias with Binning\label{subsec:Biased-Weighting-with-Binning}}

We replicate BP's estimates, focusing on the first three months since
the receipt, while leaving longer-run effects to \ref{subsec:Spurious-Identification-Application}.
BP estimate conventional dynamic specifications of the form:
\begin{equation}
Y_{it}=\alpha_{i}+\beta_{t}+\sum_{h=-a}^{b}\tau_{h}\one\left[K_{it}=h\right]+\varepsilon_{it},\label{eq:BP_spec}
\end{equation}
where $Y_{it}$ is the dollar amount of spending in calendar week
$t$ for household $i$, $\alpha_{i}$ are household FEs, and $\beta_{t}$
are week FEs. In some specifications, week FEs are interacted with
the disbursement method $m(i)$ (i.e., $\beta_{m(i)t}$ is included
instead of $\beta_{t}$ in \ref{eq:BP_spec}) to leverage the variation
in timing only within each disbursement method; we refer to those
specifications as ``with disbursement method FEs.'' The set of $\one\left[K_{it}=h\right]$
are the lead/lag indicator variables tracking the number of weeks
$K_{it}=t-E_{i}$ since the week of the tax rebate receipt for the
household, $E_{i}$; $b$ is chosen such that all possible lags in
the sample are covered; $a$ varies as discussed below. MPXs for each
horizon, as well as pre-trend coefficients, are captured by $\tau_{h}$.
Regressions are weighted by the Nielsen projection weights.

BP's preferred specification is a binned version of \ref{eq:BP_spec}
which constrains $\tau_{h}$ to be constant across four-week periods
\textemdash{} ``months'' \textemdash{} around the event, starting
with the week of tax rebate receipt: e.g., $\tau_{0}=\dots=\tau_{3}$.
This specification also includes one monthly pre-trend coefficient,
i.e. $a=4$ with $\tau_{-1}=\dots=\tau_{-4}.$ These estimates, without
and with disbursement method FEs, are replicated in \ref{tab:replication_BP},
columns 1 and 2, suggesting that tax rebate receipt led to an increase
in spending in the contemporaneous month of \$42.6 (s.e. 7.2) in col.
1 to \$47.6 (s.e. 9.2) in col. 2, and a cumulative increase over three
months of \$60.5 (s.e. 25.7) in col. 1 to \$94.4 (s.e. 33.5) in col.
2. As we will discuss in \ref{subsec:Implications}, extrapolating
these estimates from Nielsen products to all consumption implies very
large total MPX.

\begin{table}
\caption{Estimates of the Monthly and Quarterly MPX out of Tax Rebates{\small{}\label{tab:replication_BP}}}

\medskip{}

\begin{centering}
{\small{}}%
\begin{tabular}{lcccccccc}
\toprule 
 & \multicolumn{8}{c}{{\small{}Dollars spent after tax rebate receipt}}\tabularnewline
\cmidrule{2-9} \cmidrule{3-9} \cmidrule{4-9} \cmidrule{5-9} \cmidrule{6-9} \cmidrule{7-9} \cmidrule{8-9} \cmidrule{9-9} 
 & \multicolumn{2}{c}{{\small{}OLS}} &  & \multicolumn{2}{c}{{\small{}OLS}} &  & \multicolumn{2}{c}{{\small{}Imputation}}\tabularnewline
 & \multicolumn{2}{c}{{\small{}Monthly binned}} &  & \multicolumn{2}{c}{{\small{}No binning}} &  & \multicolumn{2}{c}{{\small{}Estimator}}\tabularnewline
 & {\small{}(1)} & {\small{}(2)} &  & {\small{}(3)} & {\small{}(4)} &  & {\small{}(5)} & {\small{}(6)}\tabularnewline
\midrule
{\footnotesize{}Contemporaneous month} & {\footnotesize{}42.59} & {\footnotesize{}47.57} &  & {\footnotesize{}35.02} & {\footnotesize{}27.88} &  & {\footnotesize{}38.13} & {\footnotesize{}30.54}\tabularnewline
 & {\footnotesize{}(7.19)} & {\footnotesize{}(9.15)} &  & {\footnotesize{}(5.75)} & {\footnotesize{}(7.75)} &  & {\footnotesize{}(5.68)} & {\footnotesize{}(9.08)}\tabularnewline
{\footnotesize{}First month after} & {\footnotesize{}9.31} & {\footnotesize{}26.26} &  & {\footnotesize{}$-2.28$} & {\footnotesize{}$-4.48$} &  & {\footnotesize{}\textendash 2.47} & {\footnotesize{}7.43}\tabularnewline
 & {\footnotesize{}(9.00)} & {\footnotesize{}(11.95)} &  & {\footnotesize{}(7.59)} & {\footnotesize{}(12.48)} &  & {\footnotesize{}(7.81)} & {\footnotesize{}(16.17)}\tabularnewline
{\footnotesize{}Second month after} & {\footnotesize{}8.63} & {\footnotesize{}20.52} &  & {\footnotesize{}$-5.96$} & {\footnotesize{}$-13.82$} &  & {\footnotesize{}13.08} & {\footnotesize{}4.01}\tabularnewline
 & {\footnotesize{}(11.17)} & {\footnotesize{}(14.57)} &  & {\footnotesize{}(10.06)} & {\footnotesize{}(16.38)} &  & {\footnotesize{}(22.51)} & {\footnotesize{}(29.89)}\tabularnewline
 &  &  &  &  &  &  &  & \tabularnewline
{\footnotesize{}Three-month total} & {\footnotesize{}60.53} & {\footnotesize{}94.35} &  & {\footnotesize{}26.79} & {\footnotesize{}9.58} &  & {\footnotesize{}48.75} & {\footnotesize{}41.97}\tabularnewline
 & {\footnotesize{}(25.73)} & {\footnotesize{}(33.54)} &  & {\footnotesize{}(21.43)} & {\footnotesize{}(34.42)} &  & {\footnotesize{}(30.97)} & {\footnotesize{}(46.56)}\tabularnewline
 &  &  &  &  &  &  &  & \tabularnewline
{\footnotesize{}Disbursement method FE} & {\footnotesize{}No} & {\footnotesize{}Yes} &  & {\footnotesize{}No} & {\footnotesize{}Yes} &  & {\footnotesize{}No} & {\footnotesize{}Yes}\tabularnewline
\emph{\footnotesize{}N}{\footnotesize{} observations} & {\footnotesize{}1,131,520} & {\footnotesize{}1,127,880} &  & {\footnotesize{}1,131,520} & {\footnotesize{}1,127,880} &  & {\footnotesize{}631,040} & {\footnotesize{}536,553}\tabularnewline
\emph{\footnotesize{}N}{\footnotesize{} households} & {\footnotesize{}21,760} & {\footnotesize{}21,690} &  & {\footnotesize{}21,760} & {\footnotesize{}21,690} &  & {\footnotesize{}21,760} & {\footnotesize{}21,690}\tabularnewline
\bottomrule
\end{tabular}{\small\par}
\par\end{centering}
\medskip{}

\emph{\footnotesize{}Notes}{\footnotesize{}: Columns 1 and 2 estimate
the binned version of specification \ref{eq:BP_spec} with $a=4$
and imposing that the coefficients are the same in each month, i.e.
four weeks since the rebate receipt. Columns 3 and 4 estimate the
same specification without binning, with $a=1$. These specifications
are identical to \textcite{Broda2014}, Tables 3 and 4, columns 1
and 4. Columns 5 and 6 report the efficient imputation estimator.
All columns aggregate coefficients by month for the first three months
after the rebate receipt and suppress the other coefficients. Columns
1, 3, and 5 use household and week FEs, while columns 2, 4, and 6
additionally interact week FEs with disbursement method dummies. The
estimates in column 6 exclude the last week of the quarter ($h=11$)
due to insufficient sample size. All estimates use projection weights
from the Nielsen Consumer Panel, and standard errors are clustered
by household.}{\footnotesize\par}
\end{table}

Next, we show that the MPX estimates are much smaller without binning.
In columns 3 and 4 of \ref{tab:replication_BP} we report the estimates
from the conventional specification \ref{eq:BP_spec} without binning
and with one weekly lead ($a=1$), as in BP's Table 3. We report the
coefficients aggregated to the monthly level. Compared to columns
1 and 2, there is a large fall in the cumulative three-month MPX,
from \$60.5 (s.e. 25.7) to \$26.8 (s.e. 21.4) without disbursement
method FEs and from \$94.4 (s.e. 33.5) to \$9.6 (s.e. 34.4) with these
FEs. In columns 5 and 6, we use the robust and efficient imputation
estimator to estimate weekly average responses and aggregate them
to the monthly level. The point estimates are similar to columns 3
and 4 for the contemporaneous month response, while for the quarterly
MPX they are in between the results obtained with a binned specification
and without binning.\footnote{\ref{tab:Coef-differences} reports the differences between the estimates
from OLS with no binning or imputation and the binned OLS specification,
with standard errors and p-values. All binned estimates are significantly
different from those without binning at the 10\% (5\%) significance
level without (with) disbursement method FEs. The difference between
the binned and imputation estimates is only significant at the 10\%
level for the contemporaneous month with disbursement method FEs.
We explain below that the difference in estimates is due to the difference
in estimands, rather than statistical noise.}

Could the difference between binned and other estimates indicate a
violation of the DiD assumptions? \ref{fig:OLS vs. robust imputation estimator}
provides evidence against this possibility, showing that there is
no sign of pre-trends.\footnote{This figure reports the imputation estimates for 8 weeks since the
rebate along with the pre-trend coefficients from the \ref{subsec:Testing-PTA}
test that allows for 8 weeks of anticipation effects. The figure also
reports conventional specifications without binning augmented to include
$a=8$ weeks of pre-trends, dropping observations more than 8 weeks
since the rebate.} The Wald test confirms this finding: the p-value for the null of
no pre-trends is 0.185 (0.403) without (with) disbursement method
FEs.

\begin{figure}
\caption{Dynamic Specifications and Pre-Trends\label{fig:OLS vs. robust imputation estimator}}

\medskip{}

\begin{centering}
\begin{tabular}{ccc}
A: Without disbursement method FEs &  & B: With disbursement method FEs\tabularnewline
\includegraphics[width=0.35\textwidth]{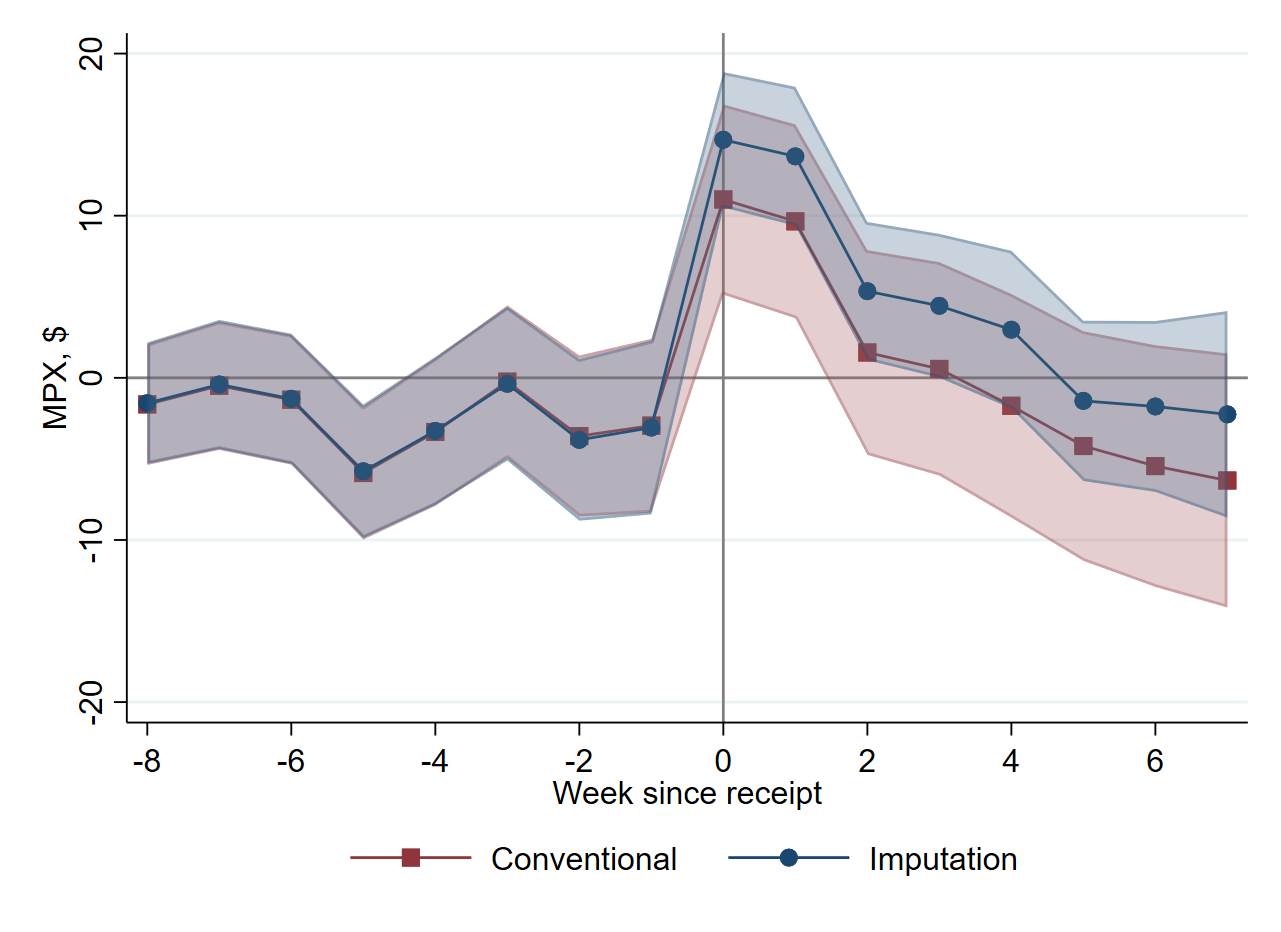} &  & \includegraphics[width=0.35\textwidth]{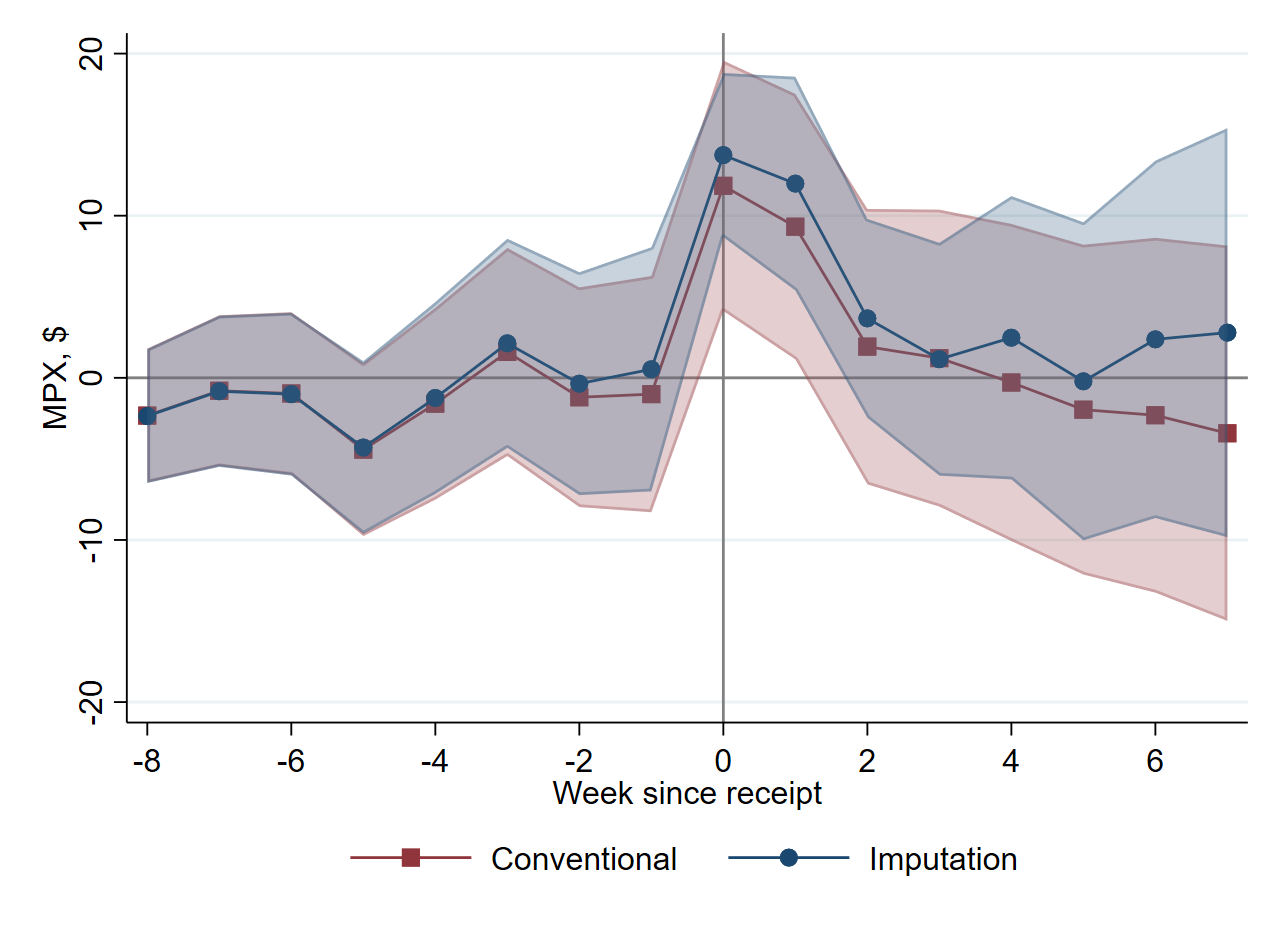}\tabularnewline
\end{tabular}
\par\end{centering}
\emph{\footnotesize{}Notes}{\footnotesize{}: Panel A reports estimates
of the response of spending to tax rebate receipts and pre-trend coefficients,
using specification \ref{eq:BP_spec} with $a=8$ and without binning
(``Conventional'') and with the efficient imputation estimator and
the pre-trend test from \ref{subsec:Testing-PTA} (``Imputation'').
Panel B additionally interacts week FEs with the disbursement method.
Observations 8 or more weeks since the rebate receipt are excluded.
Estimation is weighted by the projection weights from the Nielsen
Consumer Panel. 95\% confidence bands are shown, using standard errors
clustered by household.}{\footnotesize\par}
\end{figure}

We find instead that the higher estimates from binned specifications
are explained by the estimand they implicitly choose. Specifically,
this estimand places a very large weight on the first weeks after
the rebate, when the effects are the largest, and negative weights
on other weeks. \ref{fig:OLS vs. robust imputation estimator} shows
that the increase in spending after the receipt is concentrated in
the first weeks since the rebate. \ref{fig:binned_short_term_bias}
in turn shows the weights with which the quarterly MPX estimated from
the monthly binned specification of \ref{tab:replication_BP} aggregates
the MPXs at each weekly horizon. These weights show how the estimand
of the binned specification diverges from the true quarterly MPX,
which is a simple sum of the effects at each horizon $h=0,\dots,11$
weeks, i.e. with constant weights of one on each week.\footnote{The binned specification's estimand also diverges from the true MPX
in how it weights different households for the same weekly horizon
(similar to the issues studied theoretically by \textcite{Abraham2018}
for dynamic specifications without binning). We focus on the variation
across horizons here because MPXs have a very strong dynamic pattern.} In the specification without disbursement method FEs, the weight
placed on the first-week response is three times larger than it would
be with an equally weighted sum; it is five times larger with the
FEs. Furthermore, within each month the weights become negative for
the last weeks of the month. Applying the weights of the binned specification
across weeks from \ref{fig:OLS vs. robust imputation estimator} to
the estimates without binning (underlying col. 3 of \ref{tab:replication_BP}),
we obtain a point estimate of \$42.6 for the contemporaneous month
\textemdash{} indistinguishable from col. 1, instead of \$35.0 in
col. 3. Similarly, we get \$60.4 for the quarter, nearly identical
to \$60.5 in col. 1, instead of \$26.8 in col. 3. Thus, the short-run
biased weighting scheme due to binning explains nearly all the difference
between columns 1 and 3 of \ref{tab:replication_BP}.\footnote{Short-run biased weighting also explains the majority, although not
all, of the difference between the specification with disbursement
method FEs in columns 2 and 4 of \ref{tab:replication_BP}. Applying
the binned specification weights to the specifications without binning,
we get an estimate of \$40.0 for the contemporaneous month and \$69.0
for the quarter, thus reducing the discrepancy between columns 2 and
4 by 62\% and 70\% for the month and quarter, respectively.}

\begin{figure}
\caption{Short-term Bias in Weights for Binned Specifications\label{fig:binned_short_term_bias}}

\begin{centering}
\includegraphics[width=0.35\textwidth]{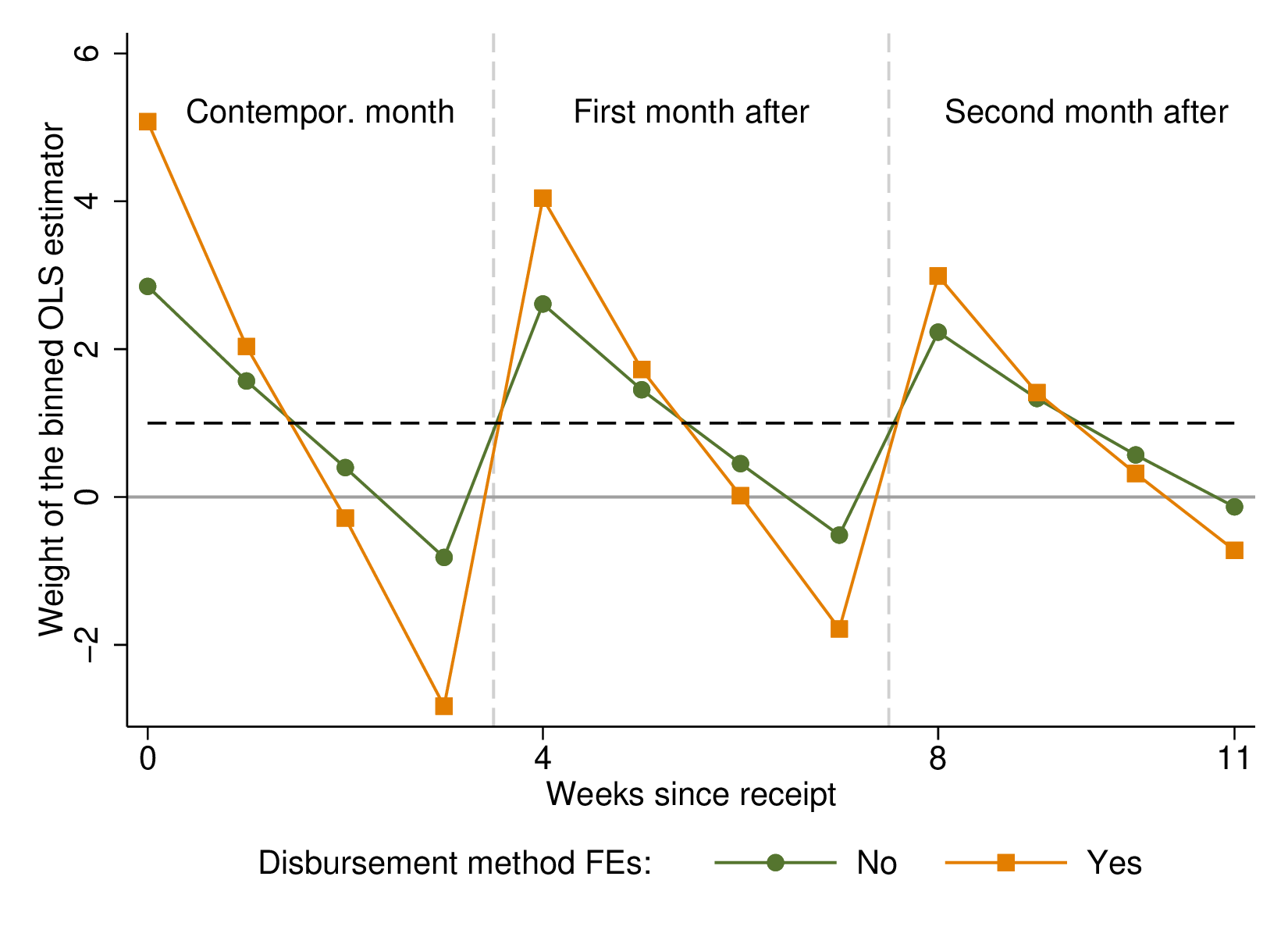}
\par\end{centering}
\emph{\footnotesize{}Notes}{\footnotesize{}: This figure reports the
cumulative weight that the monthly binned OLS estimator of the quarterly
MPX from \ref{tab:replication_BP}, with or without disbursement methods
FEs, places on the true effects at each horizon $h=0,\dots,11$ weeks
since the rebate receipt. These weights are computed using the Frisch\textendash Waugh\textendash Lovell
theorem, analogously to equation \ref{eq:weights-FWL}, and aggregated
across the first three months since the rebate receipt. The black
dashed line indicates the weight corresponding to the true quarterly
MPX, i.e. a simple sum of the effects at each horizon.}{\footnotesize\par}
\end{figure}

\subsubsection{Spurious Identification of Long-Run Causal Effects\label{subsec:Spurious-Identification-Application}}

We now examine the long-run dynamics of MPXs obtained with conventional
specifications and the imputation estimator. The timing of the tax
rebate is such that we simultaneously observe treated and untreated
households for at most 13 weeks.\footnote{The first treated households received the rebate during week 17 of
2008 (week ending April 26), while the last treated households received
it during week 30 (week ending July 26).} Per \ref{prop:No-Longrun}, without restrictive assumptions on treatment
effect heterogeneity it is not possible to estimate causal effects
beyond 12 weeks. Yet conventional dynamic specifications produce estimates
for longer horizons via extrapolation. We examine whether the estimates
obtained in this way could paint a misleading picture of the long-run
dynamics of MPXs.

\begin{figure}
\caption{Long-run MPXs, Conventional Specifications vs. Robust Imputation Estimator\label{fig:extrapolation}}

\medskip{}

\begin{centering}
\begin{tabular}{ccc}
{\small{}A: Conventional} & {\small{}B: Conventional} & {\small{}C: Robust}\tabularnewline
{\small{}Binned Estimates} & {\small{}Dynamic Estimates} & {\small{}Imputation Estimates}\tabularnewline
\includegraphics[width=0.3\textwidth]{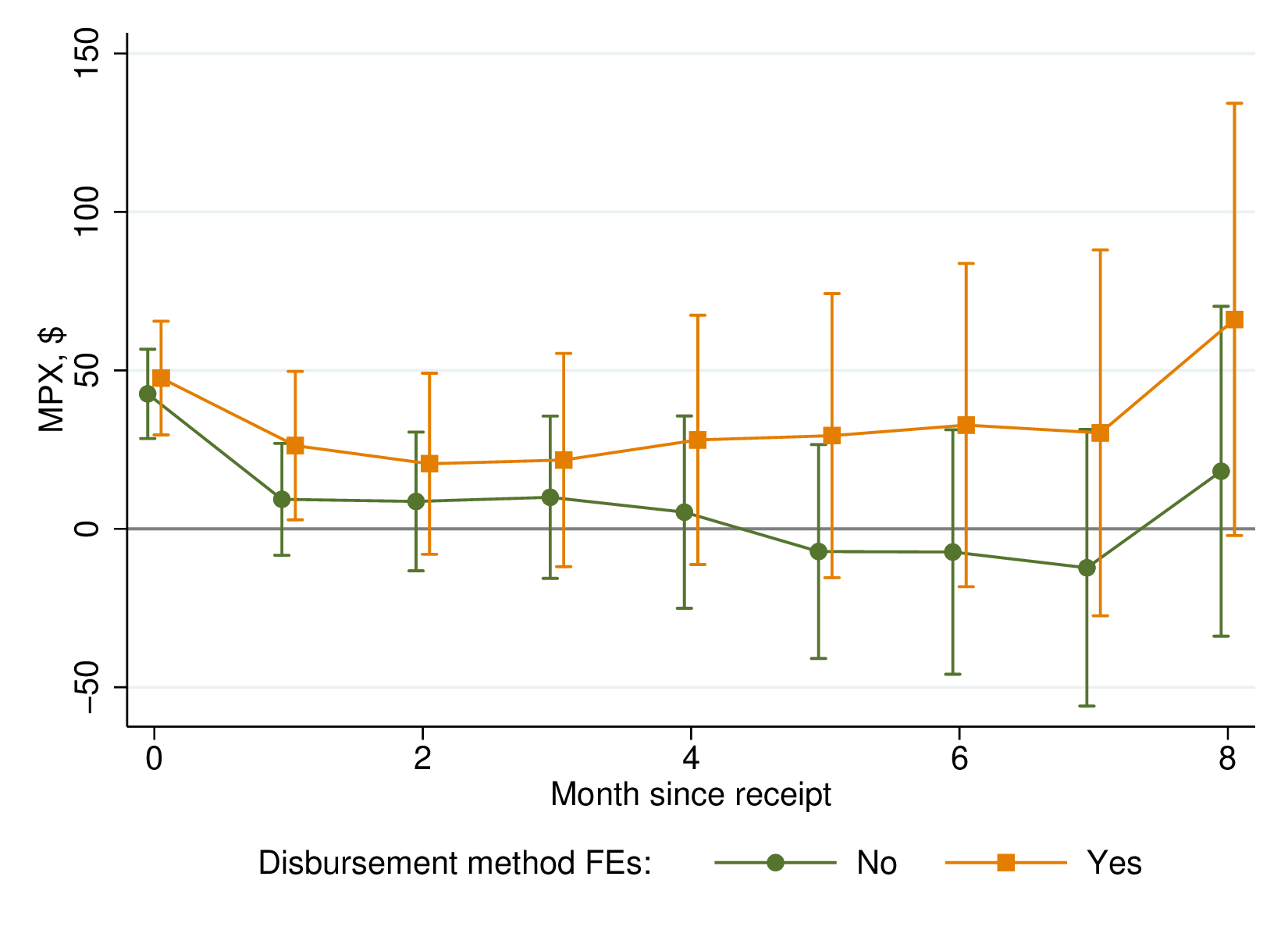} & \includegraphics[width=0.3\textwidth]{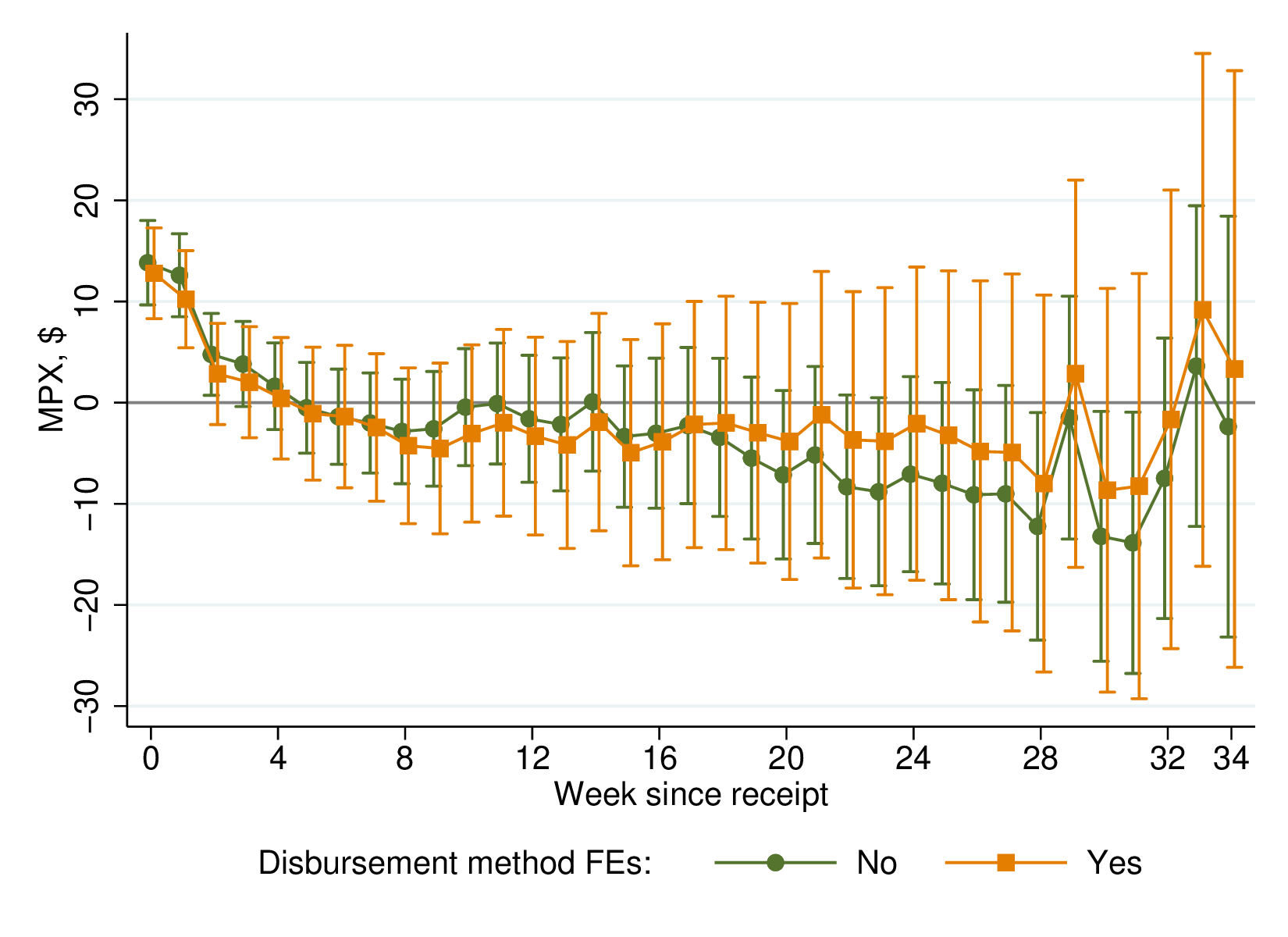} & \includegraphics[width=0.3\textwidth]{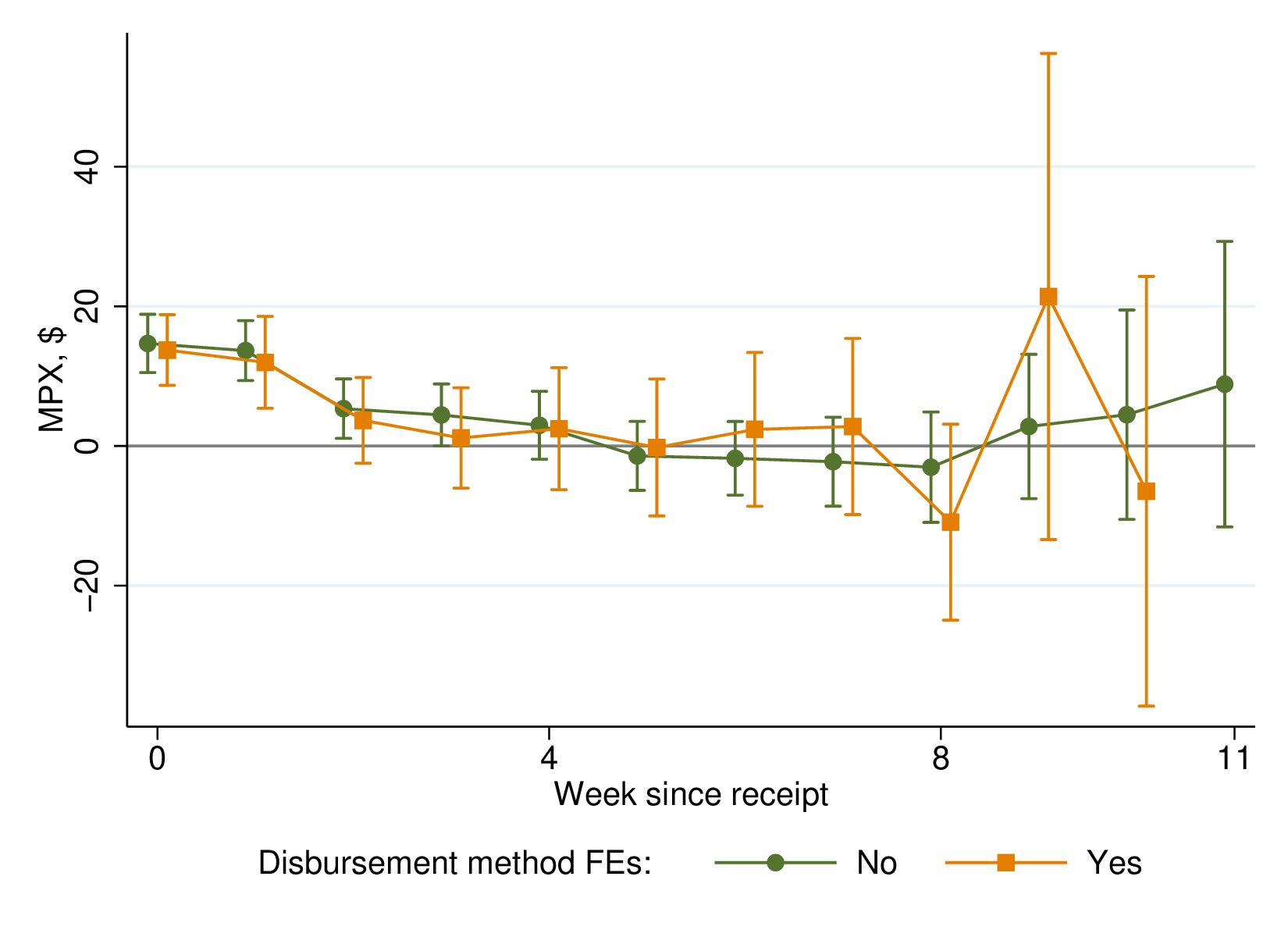}\tabularnewline
\end{tabular}
\par\end{centering}
\emph{\footnotesize{}Notes}{\footnotesize{}: Panels A\textendash C
plot MPX coefficients and 95\% confidence bands using the same specifications
as in \ref{tab:replication_BP}. Coefficients on the leads of treatment
are not shown. The last horizon in Panel B ($h=35$ weeks) and Panel
C ($h=12$ weeks without disbursement methods FEs or $h=11$ with
FEs) are suppressed because of the very large standard errors, due
to a limited sample size. Standard errors are clustered by household.}{\footnotesize\par}
\end{figure}

In \ref{fig:extrapolation}, we use the same specifications as in
\ref{tab:replication_BP} but we report the full set of dynamic estimates
for the treatment effects. Panel A reports the estimates from the
binned specification. With disbursement method FEs, the point estimates
are large and positive for all nine months following the receipt of
the tax rebate. Thus, due to the extrapolation resulting from binning,
this specification could be mistakenly interpreted as evidence for
a very large and persistent increase in spending. Without these FEs,
the estimates tend to hover around zero.

In Panel B, we show estimates with the conventional dynamic specification
without binning. Both specifications with and without disbursement
method FEs yield point estimates that are almost all negative in the
long run. Taken at face value, these estimates could misleadingly
suggest that households intertemporally substitute consumption by
making purchases at the time of tax rebate that they would have made
20 to 30 weeks later. As in Panel A, these point estimates are noisy
but could lend themselves to some economic interpretation.

In contrast, Panel C describes the results from the robust imputation
estimator, which does not allow extrapolation in the absence of an
explicit control group. This panel shows that, for the horizons for
which imputation is possible, there is no evidence of any impact on
spending beyond two to four weeks after tax rebate receipt. The patterns
are the same both with and without disbursement FEs. These results
highlight the practical relevance of the insights from \ref{subsec:Spurious}:
the imputation estimators avoid extrapolation, thus eliminating seemingly
unstable patterns found across conventional specifications.

\begin{figure}
\caption{Underidentification of the Fully-Dynamic Specification\label{fig:OLS underidentification}}

\begin{centering}
\medskip{}
\par\end{centering}
\begin{centering}
\begin{tabular}{ccc}
A: Without disbursement method FEs &  & Panel B: With disbursement method FEs\tabularnewline
\includegraphics[width=0.35\textwidth]{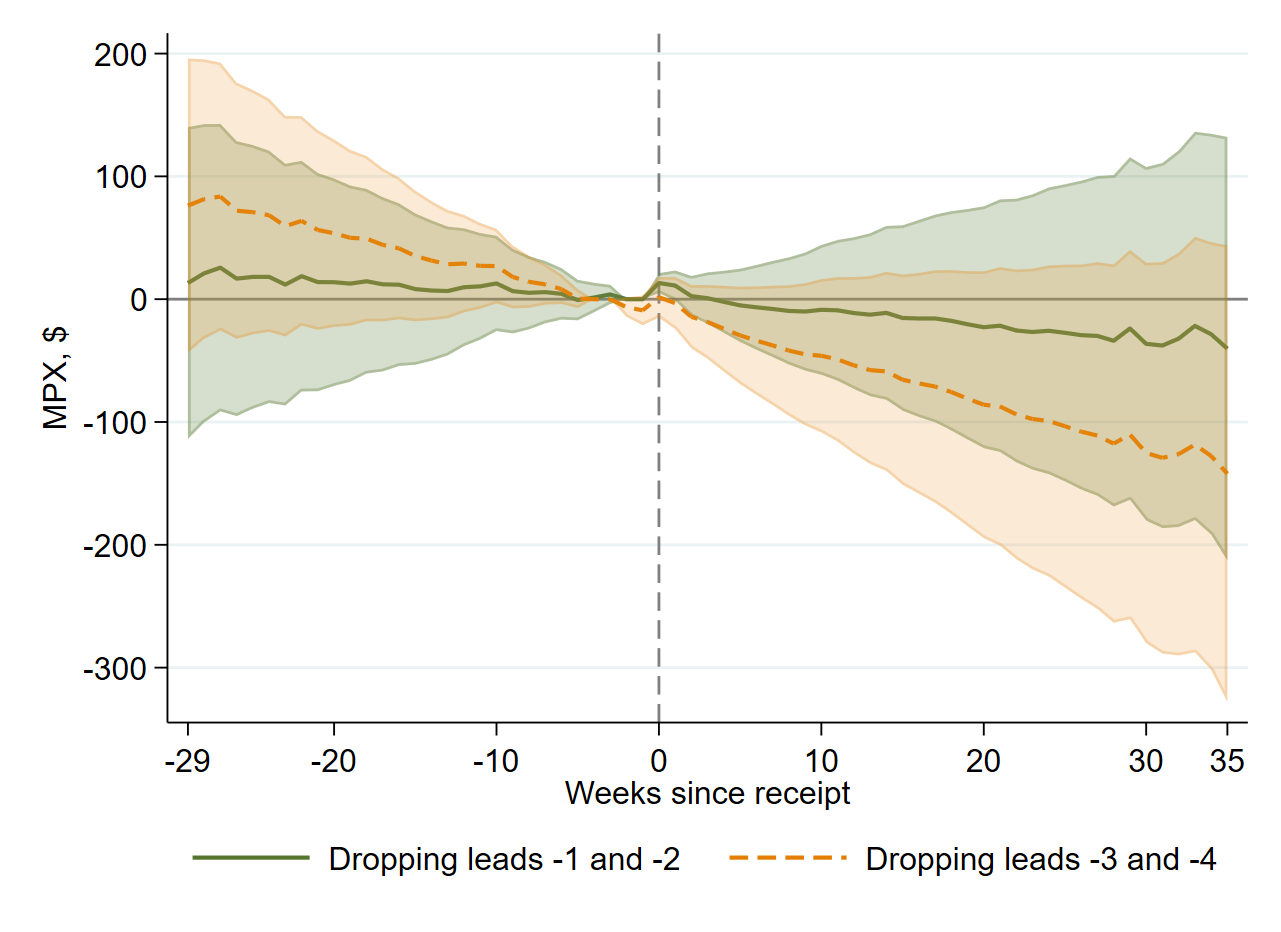} &  & \includegraphics[width=0.35\textwidth]{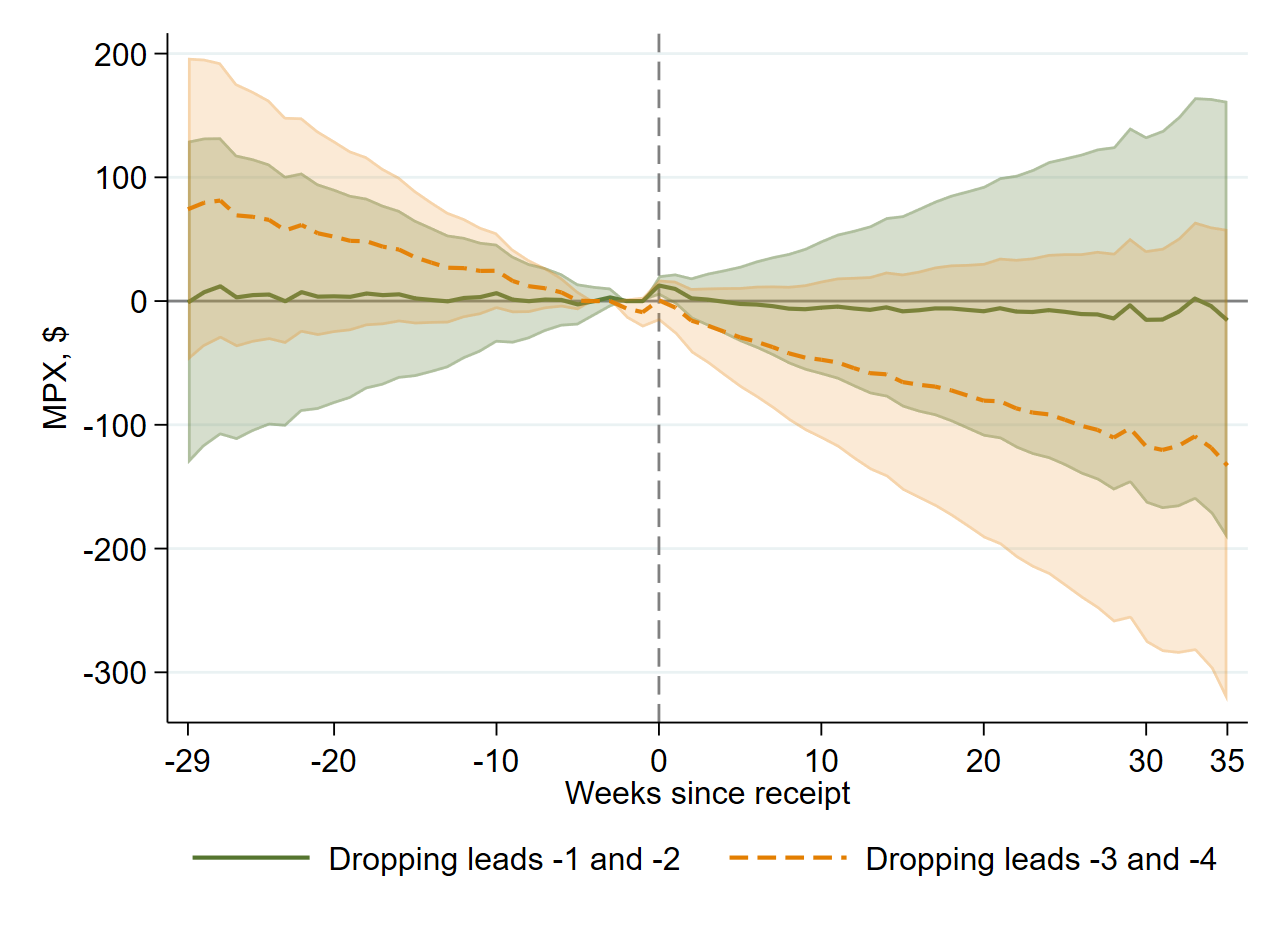}\tabularnewline
\end{tabular}
\par\end{centering}
\emph{\footnotesize{}Notes}{\footnotesize{}: This figure reports MPX
and pre-trend estimates and 95\% confidence bands for specification
\ref{eq:BP_spec} with all leads and lags of the tax receipt included,
except for two chosen as indicated. Standard errors are clustered
by household.}{\footnotesize\par}
\end{figure}

Finally, in \ref{fig:OLS underidentification} we illustrate the importance
of the insights on the underidentification of fully-dynamic specifications
from \ref{subsec:Underidentification}. Unlike earlier specifications,
which only included a small number of treatment leads, here we run
the specification \ref{eq:BP_spec} with a full set of weekly leads
and lags around tax rebate receipt. We drop two leads since the set
of lead and lag coefficients is only identified up to a linear trend,
as discussed in \ref{subsec:Underidentification}. We find that the
fully-dynamic estimates change drastically depending on which two
leads are dropped. We illustrate this by comparing the MPXs when dropping
leads $-1$ and $-2$ or $-3$ and $-4$. This shows another source
of instability in conventional practice, which the imputation estimator
directly avoids.

\subsubsection{Preferred Robust Estimates and Macroeconomic Implications \label{subsec:Implications}}

We now discuss the implications of our findings for the macroeconomics
literature. We proceed in two steps: selecting our preferred MPX estimate
from \ref{subsec:Biased-Weighting-with-Binning} for the Nielsen products
and then extrapolating it to broader consumption baskets, following
the strategy of BP.

Our preferred estimate for the average cumulative MPX out of the tax
rebate for the Nielsen products is \$30.5, corresponding to the imputation
estimator with disbursement method FEs (\ref{tab:replication_BP},
column 6) in the first month since the rebate. This constitutes 3.4\%
of the average rebate amount. We choose the specification with disbursement
method FEs because the variation in timing is more plausibly exogenous
within disbursement methods. We focus on the first (i.e., contemporaneous)
month and impose zero effects for the following months based on the
evidence from \ref{fig:OLS vs. robust imputation estimator,fig:extrapolation}
that the MPXs rapidly decay to zero, while estimation noise increases.\footnote{Our preferred estimate is robust to the choice of the time window:
the cumulative MPX would have been similar (at \$25.7 instead of \$30.5)
if we focused on the first two weeks only.} Finally, we choose the imputation estimator over conventional specifications
for its robustness properties. In contrast to columns 1\textendash 2
of \ref{tab:replication_BP}, it avoids the short-term bias due to
binning. Moreover, in contrast to column 3\textendash 4 it avoids
extrapolation of long-run effects (the estimates are similar for the
contemporaneous month). Robustness to treatment effect heterogeneity
is gained without an efficiency loss in this application: the standard
errors are similar across columns of \ref{tab:replication_BP}.

To obtain MPX estimates covering the full consumption basket, BP propose
to rescale the estimates obtained with the Nielsen data. This scaling
is done in three different ways: (i) by the ratio of spending per
capita in the National Income and Product Account (NIPA) and Nielsen
data; (ii) by the ratio of the self-reported change in spending on
all goods after the rebate relative to that on Nielsen goods alone;
(iii) by a factor based on the relative shares of spending and relative
responsiveness to the rebate across subcategories of goods as measured
in Consumer Expenditure Survey (CE). Using these three approaches
and BP's preferred MPX estimate (reproduced in our \ref{tab:replication_BP},
col. 1), they estimate that the tax rebate raised the annualized expenditure
growth rate by 1.3\textendash 1.9 percentage points (p.p.) in 2008Q2
and by 0.6\textendash 0.9 p.p. in 2008Q3, depending on the choice
of rescaling.

Applying the same scaling methods to our preferred MPX estimate for
the contemporaneous month and assuming zero response in the following
months paints a very different picture, with an increase in annualized
expenditure growth of only 0.8\textendash 1.1p.p. in 2008Q2 and 0.15\textendash 0.22p.p.
in 2008Q3. Our estimate implies a 40\% smaller response of consumption
expenditures in 2008Q2, and 75\% smaller in 2008Q3. Correspondingly,
while BP conclude that the propensity to spend at the individual level
from a tax rebate over three months since the rebate is between 51
and 75 percent, our preferred estimates are half as large, between
25 and 37 percent.\footnote{We obtain these estimates by replicating the first row of Panel A
of BP's Table 5 and using our preferred estimates.}

In \ref{tab:macro_models}, we summarize the MPX estimates for the
first quarter after tax rebate obtained with BP's and our preferred
specification. The first row reports the observed marginal propensity
to spend on products included in the Nielsen sample during that quarter,
as a fraction of the average rebate amount. The next rows rescale
these estimates to extrapolate the marginal propensity to spend to
broader samples, i.e. the full consumption basket (second row) and
nondurables (third row). For the full consumption basket, we implement
the three scaling procedures from BP and report the lower and upper
bounds; for nondurables we leverage the scaling method of \textcite{Laibson2022}.
The fourth row reports the model-consistent, or ``notional,'' marginal
propensity to consume (MPC) that can be used as a target for macroeconomic
models, also following the methodology of \textcite{Laibson2022}.\footnote{Standard macroeconomic models assume a notional consumption flow that
does not distinguish between nondurable and durable consumption. Prior
to \textcite{Laibson2022} showing that the notional MPC should be
the relevant target, state-of-the-art macroeconomic models targeted
nondurable MPX estimates. For instance, \textcite{Kaplan2014} targeted
the estimates from \textcite{Johnson2006a}, which are quantitatively
similar to those from BP when rescaled as in \ref{tab:macro_models},
despite using more aggregated data and a different rebate episode.} The estimates based on BP in column 1 are closely in line with the
literature: typical estimates of the quarterly MPX for all expenditures
range from 50-90\%, while estimates of the quarterly MPX for nondurable
expenditure range from 15-25\%.\footnote{\textcite{Laibson2022} provide a recent review of the literature.
\textcite{kaplan2020marginal} review nondurable MPX, and \textcite{di2020stock}
review total MPX.} In contrast, the imputation estimator in column 2 of \ref{tab:macro_models}
delivers estimates that are about half as large in all rows. These
smaller MPC estimates imply a lower effectiveness of fiscal stimulus.

{\small{}}
\begin{table}[t]
{\small{}\caption{{\small{}First-quarter MPX and MPC Estimates for Calibration of Macroeconomic
Models\label{tab:macro_models}}}
\medskip{}
}{\small\par}
\begin{centering}
{\small{}}%
\begin{tabular}{lcc}
\toprule 
\multirow{2}{*}{{\small{}Statistic}} & {\small{}Replication of} & \multirow{2}{*}{{\small{}Imputation Estimator}}\tabularnewline
 & {\small{}\textcite{Broda2014}} & \tabularnewline
 & {\small{}(1)} & {\small{}(2)}\tabularnewline
\midrule
{\small{}Nielsen MPX} & {\small{}6.7\%} & {\small{}3.4\%}\tabularnewline
{\small{}Total MPX} & {\small{}50.8\% to 74.8\%} & {\small{}24.8\% to 36.6\%}\tabularnewline
{\small{}Nondurable MPX} & {\small{}14.1\% to 20.8\%} & {\small{}6.9\% to 10.2\%}\tabularnewline
{\small{}Notional MPC} & {\small{}15.9\% to 23.4\%} & {\small{}7.8\% to 11.4\%}\tabularnewline
\bottomrule
\end{tabular}{\small{}\medskip{}
}{\small\par}
\par\end{centering}
\emph{\footnotesize{}Notes:}{\footnotesize{} This table reports the
first-quarter MPX and MPC using the preferred binned specification
of \textcite{Broda2014} and our preferred specification based on
the imputation estimator. The first row reports the marginal propensity
to spend on products included in the Nielsen sample, as a fraction
of the average rebate amount. The second row rescales these estimates
to extrapolate them to the marginal propensity to spend on all goods
using the three rescaling methods from \textcite{Broda2014}. The
ranges correspond to the lowest and highest values among the three
rescaling methods. To obtain the estimate for the nondurables MPX
in the third row, we use the scaling factor of \textcite{Laibson2022},
who show that the total MPX is equal to 3.6 times the nondurables
MPX. The fourth row also follows the methodology of \textcite{Laibson2022}
and reports the model-consistent (``notional'') MPC that can be
used as a target for macroeconomic models, equal to the total MPX
divided by 3.2.}{\footnotesize\par}
\end{table}
{\small\par}

Thus, our new estimates for the impact of the 2008 fiscal stimulus
on the U.S. economy yield two lessons for the calibration of macroeconomic
models: (1) that the targeted MPC should be significantly smaller
\textemdash{} about half as large \textemdash{} and (2) that it is
best to calibrate the model using weekly-level estimates of the MPC,
as we report in \ref{fig:OLS vs. robust imputation estimator}, rather
than monthly or, especially, quarterly MPC estimates, which are much
noisier. Indeed, models should reflect that most of the spending response
occurs in the very short run, in the first two to four weeks after
tax rebate receipt. 

\subsection{Efficiency Gains Relative to Alternative Robust Estimators \label{subsec:Efficiency-Gains-MPC}}

Finally, we compare the efficiency of the imputation estimator to
the alternative robust estimators of \textcite{DeChaisemartin2020}
and \textcite{Abraham2018}, abbreviated dCDH and SA. We document
the in-sample efficiency gains in \ref{fig:efficiency_real} by showing
the point estimates and confidence intervals for weekly average MPXs
based on the imputation estimator and the two alternatives in Panel
A. We use the specification without disbursement method FEs.\footnote{We implement the dCDH method using the\texttt{ csdid} Stata command
developed for the \textcite{Callaway2018} estimator: the two estimators
are identical absent additional controls, and \texttt{csdid} allows
for projection weights.} The point estimates are very similar for dCDH and the imputation
estimator, but they differ from those of SA, because this estimator
uses a much smaller control group (only the households who received
the rebate in the latest possible week) and is therefore much noisier.

Panel B zooms in on the efficiency comparison by reporting the lengths
of the confidence intervals for SA and dCDH relative to that of the
imputation estimator. The differences are large: the confidence interval
from dCDH is about 50\% longer for all periods, and 2\textendash 3.5
times longer for SA.

In \ref{subsec:Monte-Carlo-BP-Appx-NEW} we confirm these efficiency
gains, obtained from a single sample, in a Monte Carlo study based
on the BP data, for several data-generating processes. We find that
the imputation estimator has sizable efficiency advantages over alternative
robust estimators not only with spherical errors but also in presence
of heteroskedasticity, serial correlation, or both. Moreover, these
gains do not come at a cost of systematically higher sensitivity to
parallel trend violations. We also confirm that our analytical standard
errors have correct coverage.

\begin{figure}
\caption{Alternative Robust MPX Estimates and In-Sample Efficiency\label{fig:efficiency_real}}

{\small{}\medskip{}
}{\small\par}
\begin{centering}
\begin{tabular}{ccc}
{\small{}A: Point Estimates} &  & {\small{}B: Confidence Interval Lengths,}\tabularnewline
{\small{}and Confidence Intervals} &  & {\small{}Relative to the Imputation Estimator}\tabularnewline
\includegraphics[width=0.35\textwidth]{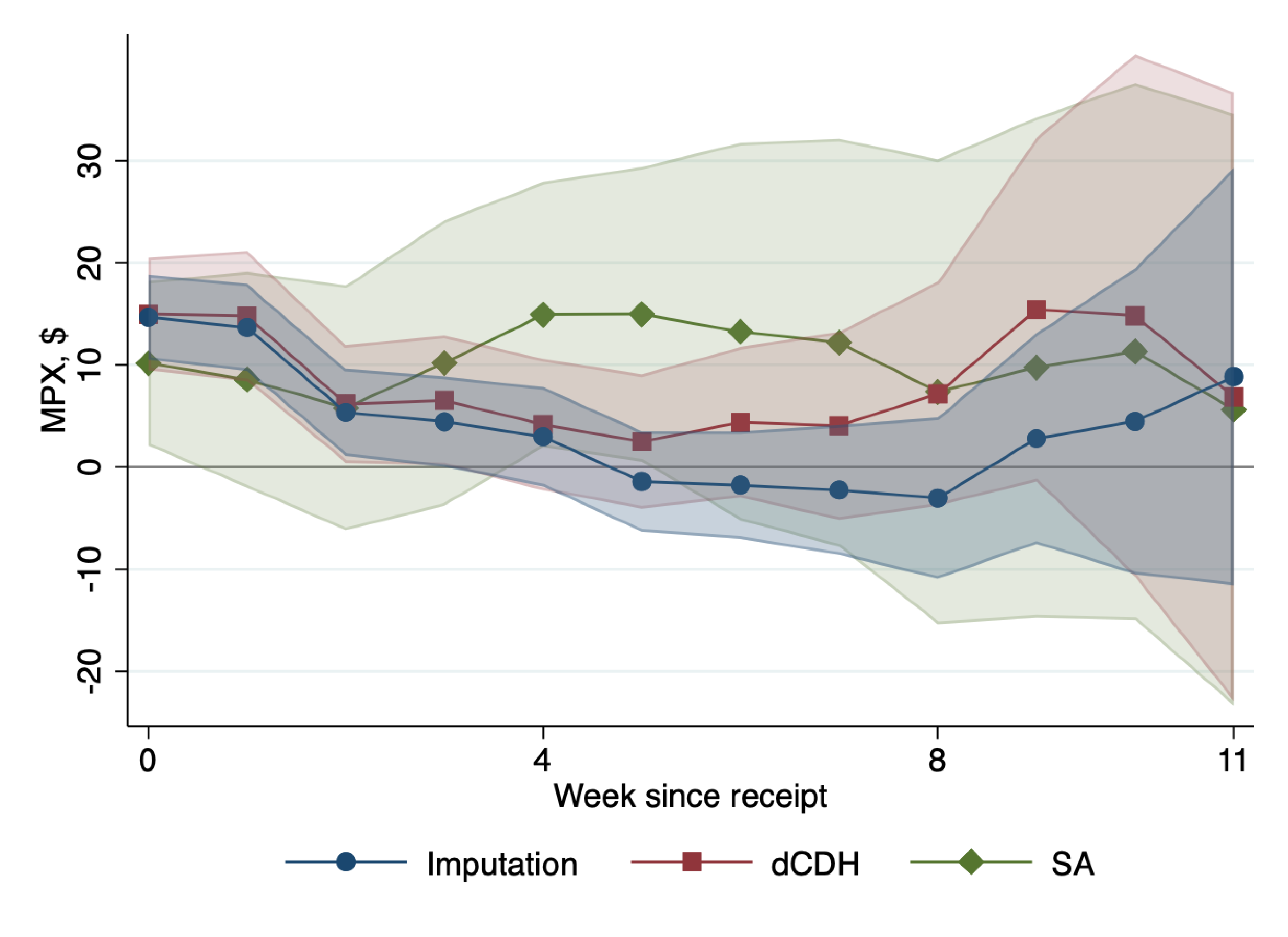} &  & \includegraphics[width=0.35\textwidth]{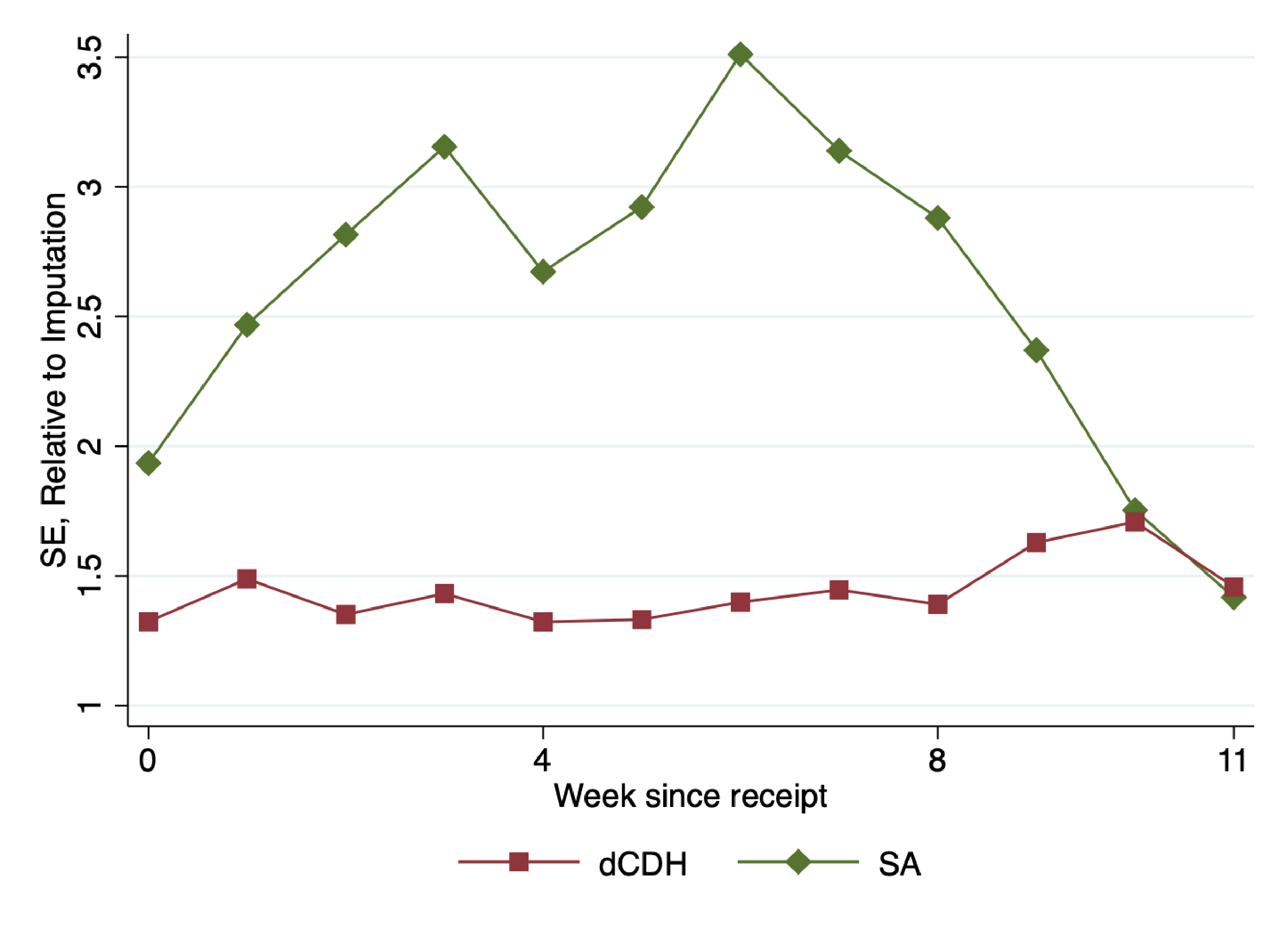}\tabularnewline
\end{tabular}
\par\end{centering}
\emph{\footnotesize{}Notes:}{\footnotesize{} Panel A shows the estimates
and 95\% confidence bands for the average MPXs by week since rebate
using three robust estimators: the imputation estimator, \textcite{DeChaisemartin2020}
(dCDH) and \textcite{Abraham2018} (SA). The specifications do not
include disbursement method fixed effects. Panel B reports the ratios
of the lengths of confidence intervals for dCDH and SA relative to
the imputation estimator. Standard errors are clustered by household.}{\footnotesize\par}
\end{figure}

\section{Conclusion}

In this paper, we provided a unified framework that formalizes an
explicit set of goals and assumptions underlying event study designs,
reveals and explains challenges with conventional practice, and yields
an efficient estimator. In a benchmark case where treatment-effect
heterogeneity remains unrestricted, this robust and efficient estimator
takes a particularly simple ``imputation'' form that estimates fixed
effects among the untreated observations only, imputes untreated outcomes
for treated observations, and then forms treatment-effect estimates
as weighted averages over the differences between actual and imputed
outcomes. We developed results for asymptotic inference and testing
and compared our approach to other estimators. We also highlighted
the importance of separating testing of identification assumptions
from estimation, which increases estimation efficiency and helps address
inference biases due to pre-testing. We demonstrated the practical
relevance of these insights in an empirical application documenting
that the notional marginal propensity to consume is between 8 and
11 percent in the first quarter, about half as large as benchmark
estimates.

\paragraph{Data Availability Statement}

The data underlying this article are publicly available on Zenodo, at \href{https://doi.org/10.5281/zenodo.10037585}{doi.org/10.5281/zenodo.10037585}.

\printbibliography

\appendix
\newpage\counterwithin{assumption}{section}\renewcommand{\theassumption}{A\arabic{assumption}} \counterwithin{prop}{section}\renewcommand{\theprop}{A\arabic{prop}}\counterwithin{cor}{section}\renewcommand{\thecor}{A\arabic{cor}}\pagenumbering{arabic}\renewcommand*{\thepage}{A\arabic{page}} \newrefsection

\part*{Online Appendix to \textquotedblleft Revisiting Event Study Designs:
Robust and Efficient Estimation\textquotedblright}

Kirill Borusyak, Xavier Jaravel, and Jann Spiess, September 2023

\section{Details and Additional Results\label{sec:appx}}

\subsection{Stochastic Regressors, Exogeneity, and Efficiency\label{subsec:Stochastic-Regressors}}

In the main article, we assume that event times, which periods are
observed for each unit, and the covariates are non-stochastic. In
this section, we show how the generalized fixed-effect model in \ref{sec:Setting}
can be obtained from a model where unit-level information is stochastic.
This connection allows us to (1) express assumptions about the exogeneity
of covariates in a more natural way and (2) expand our efficiency
statements beyond the fixed-regressor case. This discussion complements
the fixed-regressor approach in the main article, which we choose
because it allows for a direct connection to the Gauss\textendash Markov
theorem and its efficiency implications.

We now denote by $(Y_{i},Y_{i}(0),A_{i},A_{i}(0),X_{i},X_{i}(0),E_{i},\Omega_{i},\lambda_{i})_{i=1}^{I}$
the relevant variables for a set of $I$ units that are independently
distributed across $i$, and where each component of $Y_{i},Y_{i}(0),A_{i},A_{i}(0),\allowbreak X_{i},X_{i}(0)$
collects information for unit $i$ across all time periods, $X_{it}(0)$
and $A_{i}(0)$ are the potential values of covariates in the absence
of treatment, and $\Omega_{i}$ collects all periods $t$ for which
the outcomes for unit $i$ is observed. We write $\mathcal{I}_{i}=\left(A_{i}(0),X_{i}(0),E_{i},\Omega_{i},\lambda_{i}\right)$
for baseline information about unit $i$, and consider treatment effects
$\tau_{it}(\mathcal{I}_{i})=\expec{Y_{it}-Y_{it}(0)\mid\mathcal{I}_{i}}$
for all $t\geq E_{i}$. We now express assumptions on the distribution
of the units that allow us to perform causal inference on treatment
effects $\tau_{it}(\mathcal{I}_{i})$ based on observations $((Y_{it},A_{it},X_{it})_{t\in\Omega_{i}},E_{i})_{i=1}^{I}$,
which expresses that only realized covariates and outcomes for times
$t\in\Omega_{i}$ are observed. These assumptions are:
\begin{enumerate}
\item The generalized fixed-effects model holds for potential outcomes with
strictly exogenous baseline information: $Y_{it}(0)=A_{it}'(0)\lambda_{i}+X_{it}'(0)\delta+\varepsilon_{it}$
and $\expec{\varepsilon_{it}\mid\mathcal{I}_{i}}=0$ for all $i$
and $t\in\Omega_{i}$;
\item No anticipation: $Y_{it}=Y_{it}(0)$ for all $i$ and $t\in\Omega_{i}$
with $t<E_{i}$;
\item Covariates are causally unaffected by treatment: $A_{it}{=}A_{it}(0)$
and $X_{it}{=}X_{it}(0)$ for all $i$, $t\in\Omega_{i}$.
\end{enumerate}
Here, the strict exogeneity assumption in 1. implies that the relationship
of $Y_{it}(0)$ to the control covariates $X_{is}$ for $s\neq t$
and the treatment time $E_{i}$ is fully captured by concurrent covariates
$X_{it}$ and the random effect $\lambda_{i}$. The fixed-regressor
model is then obtained by conditioning on $\mathcal{I}=(\mathcal{I}_{i})_{i=1}^{I}$,
and fulfills \ref{assu:A1prime,assu:A2,assu:clustered} (provided
that the outcomes in the population distribution have bounded second
moments).

This formulation also allows us to extend the efficiency results
from \ref{sec:Imputation-Solution} to the case of stochastic regressors.
If we consider estimands $\tau_{w}(\mathcal{I})=\sum_{it\in\Omega}w_{it}(\mathcal{I})\:\tau_{it}(\mathcal{I}_{i})$
with weights that are $\mathcal{I}$-measurable and estimators $\hat{\tau}$
that are, conditional on baseline information $\mathcal{I}$, linear
in outcomes $Y_{it}$ and unbiased for $\tau_{w}(\mathcal{I})$, then
the efficiency results in our fixed-regressor setting still apply
since $\var{\hat{\tau}}=\var{\expec{\hat{\tau}\mid\mathcal{I}}}+\expec{\var{\hat{\tau}\mid\mathcal{I}}}$
by the law of total variance, where the first part is pinned down
by unbiasedness, and the second part is minimized by an estimator
$\hat{\tau}$ that is efficient conditional on baseline information.
The conditionally unbiased efficient estimator is therefore the imputation
estimator we describe in \ref{sec:Imputation-Solution}. To achieve
lower variance than the imputation estimator one would have to forgo
unbiasedness conditional on the baseline information $\mathcal{I}$.

\subsection{Sampling-Based Parallel Trend Assumptions\label{subsec:appx-Sampling}}

In this section, we show that the simple TWFE model in \ref{assu:A1}
can be obtained from parallel-trend assumptions formulated in terms
of group-wise averages in a population model with randomly sampled
units. Our model in \ref{sec:Setting} is therefore more general;
it captures common sampling and parallel-trend assumptions, while
it allows for more flexible modeling of heterogeneity and may be preferable
when the panel is incomplete.

To relate our approach to random sampling with parallel trends, assume
now that for a given number $T$ of periods we observe a complete
panel of $I$ units with outcomes $Y_{i}=(Y_{i1},\ldots,Y_{iT})$
and treatment time $E_{i}$. We write $Y_{i}(0)=(Y_{i1}(0),\ldots,Y_{iT}(0))$
for the corresponding vector of potential outcomes if the unit was
never to be treated. We assume that $(Y_{i,},Y_{i}(0),E_{i})$ are
\emph{iid }across units. We can then impose assumptions on the distribution
of $(Y_{i,},Y_{i}(0),E_{i})$:
\begin{enumerate}
\item Parallel trends: $\expec{Y_{i,t+1}(0)-Y_{it}(0)\mid E_{i}}$ does
not vary with $E_{i}$;
\item No anticipation: $Y_{it}=Y_{it}(0)$ for $t<E_{i}$, a.s.
\end{enumerate}
This formulation is similar to those in e.g. \textcite{DeChaisemartin2018,Abraham2018,Callaway2018},
although details vary about the cohorts and periods for which parallel
trends are assumed. In this model, we can define unit fixed effects
$\alpha_{i}=\expec{Y_{i1}(0)\mid E_{i}}$, period fixed effects $\beta_{t}=\expec{Y_{it}(0)-Y_{i1}(0)}$,
and treatment effects $\tau_{it}=\expec{Y_{it}-Y_{it}(0)\mid E_{i}}$,
where unit FEs and treatment effects do not vary within cohorts. We
then obtain the fixed-regressor model in \ref{sec:Setting} by conditioning
on event timing $\left\{ E_{i}\right\} _{i}$. Thus, \emph{iid} sampling,
a parallel-trend assumption defined on the cohort level, and a complete
panel jointly imply our assumptions in \ref{sec:Setting}.

Relative to a sampling-based approach, we see three main advantages
of our more general conditional fixed-effects model. First, considering
units individually allows us to estimate their treatment effects separately,
which permits the estimation of weighted treatment effect sums that
put non-constant weights on units with the same relative event time
(e.g. when estimating the gap between average treatment effects for
men and women two periods since the treatment onset). Second, we can
more naturally capture settings in which the convenience assumption
of sampling from a population is unrealistic, such as when we observe
all US states. Third, our unit fixed-effects model can handle missing
observations even when the composition of units changes over time,
while cohort-based parallel trend assumptions (e.g. that $\expec{Y_{i,t+1}(0)-Y_{it}(0)\mid E_{i}}$
does not vary with $E_{i}$) may be unattractive; similarly, estimation
with cohort, rather than individual, fixed effects can lead to biases
with incomplete panels.

\subsection{The Challenges Persist when Trimming around Event Times\label{subsec:Trimming}}

In this section, we show that the challenges described in \ref{sec:Conventional-Practice}
apply even if the sample is ``trimmed'' to a fixed window around
the event time. By trimming, we mean the relatively common practice
(sometimes called ``balancing'') of dropping observations more than
$a$ periods before or $b$ periods after the event, for some $a>0$
and $b\ge0$ (e.g. \textcite{Miller2017}, \textcite{Bartik2019a}).

One may think that the estimands from static and dynamic conventional
specifications may be close to their desirable targets in trimmed
samples because by construction the composition of units is unchanged
across horizons. However, we show that, with staggered adoption, the
weights implied by TWFE regressions can remain complex and highly
skewed. Intuitively, this follows because trimmed panels are necessarily
unbalanced in terms of units and periods; for instance, the composition
of units varies across periods by construction.\footnote{The issues are also not resolved by dropping unit FEs. Since the composition
of units is mechanically identical across horizons in a trimmed sample,
one may think that excluding unit FEs is innocuous and would make
the estimands of conventional specifications closer to its desirable
target. In fact, in the presence of period FEs (and residualizing
on them), unit indicators are no longer orthogonal to the lag and
lead indicators in trimmed samples. Thus, regressions without unit
FEs do not estimate weighted sums of treatment effects under \ref{assu:A1,assu:A2}.}

To illustrate how the limitations of conventional specifications persist
with trimming, and can be made worse, we present a numerical example.
In the example, negative weighting in the static TWFE regression is
even more severe after trimming. Similarly, weights implied by the
dynamic specification with trimming show an even larger skew than
without trimming. Finally, trimming can exacerbate the issue of spurious
identification of long-run effects. Since trimming also generally
reduces the sample size and thus estimation efficiency, it appears
difficult to justify.\footnote{One scenario in which trimming is justified within our framework is
when \ref{assu:A1,assu:A2} are imposed only on observations within
a certain window of the event.} Our numerical example considers five equal-sized cohorts treated
in periods $E_{i}=5,\dots,9$ and observed in periods $t=1,\dots,12$,
with \ref{assu:A1,assu:A2} satisfied in the complete panel. We suppose
the researcher decides to trim the sample, keeping four untreated
($K_{it}=-4,\dots,-1$) and four treated ($K_{it}=0,\dots,3$) periods
for each unit.

We show that the short-term bias and negative weighting of the static
regression persist with trimming. Using the Frisch\textendash Waugh\textendash Lowell
theorem, we compute that cumulative weights that the static specification
puts on horizons $0,\dots,3$ are 0.875, 0.425, 0.025, and $\lyxmathsym{\textendash}0.325$,
respectively (with the penultimate weight combining positive and negative
weights on different cohorts). There is more\emph{ }negative weighting
with trimming than in the complete panel: the total of all negative
weights is $-0.367$ with trimming, compared to $-0.316$ without.

The challenges pertaining to dynamic specifications persist, too.
We consider the semi-dynamic specification which includes all lags
and no leads of the event, with or without trimming.\footnote{Underidentification of fully-dynamic TWFE specifications in the absence
of never-treated units also applies directly with trimmed samples,
and is only more relevant as trimming may involve dropping never-treated
units.} In our example, we find that the estimands $\tau_{0},\dots,\tau_{3}$
are less homogeneous in trimmed samples in terms of the composition
of cohorts underlying them. \ref{fig:Trimming} reports the weights
that $\tau_{h}$ places on the observations $h$ periods after treatment
across various cohorts, for $h=0,\dots,3$. Panel A, which corresponds
to the trimmed sample, shows that all estimands place higher weights
on earlier-treated cohorts, but much more so for larger $h$. This
is in contrast to the complete panel (Panel B), where the differences
both across cohorts and across $h$ are much smaller.

Spurious identification of long-run effects can also be reinforced
by trimming, as observations for late-treated units in early periods
are dropped. In our example, there are no admissible DiD comparisons
for any unit observed three periods after treatment in the trimmed
sample, making $\tau_{3}$ identified through extrapolation of treatment
effects only. In contrast, admissible comparisons are available in
the complete panel: the cohort treated at $E_{i}=5$ and observed
at $t=8$ can be compared to that treated at $E_{i}=9$ and to any
period $t=1,\dots,4$ when both cohorts are not yet treated.\footnote{A less extreme version of this problem is that, like in complete panels,
$\tau_{0}$, $\tau_{1}$ and $\tau_{2}$ are confounded by the heterogeneity
of treatment effects at other horizons \parencite{Abraham2018}. The
argument here focuses on the horizons which are present in the trimmed
sample; naturally, trimming eliminates some horizons, such as $h=4,\dots,7$
in our example, for which causal effects are not identified in the
complete panel.}

\subsection{Sufficient Conditions for Identification}

In this section, we provide sufficient conditions for the identification
of $\tau_{w}$. For brevity of notation, we write $A_{it}^{\prime}\lambda_{i}+X_{it}^{\prime}\delta\equiv Z_{it}^{\prime}\pi$,
where all parameters $\lambda_{i}$ and $\delta$ are collected into
a single column vector $\pi$, with the corresponding covariates collected
in $Z_{it}$. We further let $Z$ be the matrix with $N$ rows $Z_{it}^{\prime}$,
and $Z_{1}$ and $Z_{0}$ its restrictions to observations $\Omega_{1}$
and $\Omega_{0}$, respectively. Under \ref{assu:A1prime,assu:A2},
we can therefore write $Y_{it}=Z_{it}^{\prime}\pi+D_{it}\tau_{it}+\varepsilon_{it}$.
\begin{prop}[Sufficient high- and low-level conditions for identification]
\label{prop:identification}Under \ref{assu:A1prime,assu:A2} and
null \ref{assu:A3prime}:
\end{prop}
\begin{enumerate}
\item In the general case $Y_{it}=Z_{it}'\pi_{it}+D_{it}\tau_{it}+\varepsilon_{it}$,
write $Z_{1}^{*}=(Z_{it}')_{it\in\Omega_{1},w_{it}\ne0}$ for the
matrix with rows $Z_{it}'$ corresponding to treated observations
with non-zero weights. Then $\tau_{w}$ is identified if $\text{rank}(Z_{0})=\text{rank}\left(\begin{psmallmatrix}Z_{1}^{*} \\ Z_{0}\end{psmallmatrix}\right)$.
\item In the specific case $Y_{it}=A_{t}^{\prime}\lambda_{i}+\beta_{t}+D_{it}\tau_{it}+\varepsilon_{it}$
with simple period FEs and unit-specific exposures $\lambda_{i}$
to factors $A_{t}$ that only vary by period (such as unit FEs or
unit-specific trends), the following two assumptions are sufficient
to ensure identification of $\tau_{w}$:
\begin{enumerate}
\item For every unit\textit{ $i$} with $\sum_{t;it\in\Omega_{1}}\left|w_{it}\right|>0$,
we have that $(A_{t}')_{t;it\in\Omega_{0}}$ is of full column rank.
(In the TWFE case $A_{t}=1$ this means that for every such unit there
is at least one untreated observation. Adding unit-specific trends,
two are required.)
\item There is at least one unit for which untreated outcomes are observed
for all time periods up to $T(w)=\max_{t}\left\{ \exists i:w_{it}\ne0\right\} $.
\end{enumerate}
\end{enumerate}
Here, by identification we mean that the parameter is recoverable
from the \emph{distributions} of all outcomes $Y_{it}$.\footnote{Recall that the set of observations and the realizations of treatments
and included covariates are non-stochastic.} Under these assumptions, two stronger implications are true: first,
knowledge of the expectations $\expec{Y_{it}}$ (rather than the full
distributions of $Y_{it}$) is sufficient to recover $\tau_{w}$;
and second, an unbiased estimator of $\tau_{w}$ exists. Additional
conditions, like those we impose in \ref{sec:Imputation-Solution},
are required for such an unbiased estimators to also be consistent
(cf. \cite{Goldsmith-Pinkham2013-eq} for a discussion of identification
and consistency in those separate steps). While \ref{prop:identification}
applies in the case of unrestricted treatment-effect heterogeneity,
it extends directly to non-trivial \ref{assu:A3prime} since adding
structure on the treatment effects only makes identification easier.

\subsection{Efficient Estimation beyond Spherical Errors\label{subsec:appx-GLS}}

In this section, we consider an extension of our efficient estimator
to the case of non-spherical errors, connecting to the idea of feasible
GLS estimation from \textcite{Wooldridge2021}. The estimator we propose
is robust to unrestricted treatment effect heterogeneity, and is efficient
if variance\textendash covariance matrices of error terms can be estimated
consistently. We illustrate these points in our general model $Y_{it}=Z_{it}'\pi+D_{it}\tau_{it}+\varepsilon_{it}$
with unrestricted treatment-effect heterogeneity. We assume observations
are independent across units $i$, but unit-level error vectors $\varepsilon_{i}=(\varepsilon_{it})_{t;it\in\Omega}$
can have non-trivial (and unit-specific) variance\textendash covariance
matrices $\Sigma_{i}=\textnormal{Var}(\varepsilon_{i})$, which we
assume are non-singular. We consider the class of linear estimators
$\hat{\tau}_{w}=\sum_{it\in\Omega}v_{it}Y_{it}$ that are unbiased
for $\tau_{w}=\sum_{it\in\Omega_{1}}w_{it}\tau_{it}$ under unrestricted
treatment-effect heterogeneity, which implies that (and indeed is
equivalent to) $v_{it}=w_{it}$ for all $it\in\Omega_{1}$ and $\sum_{it\in\Omega}Z_{it}v_{it}=0$.

We start by expressing the problem of finding the ``oracle'' efficient
unbiased estimator under non-spherical errors. The variance of the
linear estimator $\hat{\tau}_{w}=\sum_{it\in\Omega}v_{it}Y_{it}$
is $\var{\hat{\tau}_{w}}=\sum_{i}v_{i}'\Sigma_{i}v_{i}$ (where we
write $v_{i}=(v_{it})_{t;it\in\Omega}$). The efficient oracle estimator
is the GLS estimator that minimizes this variance subject to the unbiasedness
constraint $\expec{\hat{\tau}_{w}}=\tau_{w}$. Writing $\Sigma_{i}^{1}=\var{\varepsilon_{i}^{0}}$
and $\Sigma_{i}^{01}=\cov{\varepsilon_{i}^{0},\varepsilon_{i}^{1}}$
for the components of $\Sigma_{i}$ corresponding to the variance
of $\varepsilon_{i}^{0}=(\varepsilon_{it})_{t;it\in\Omega_{0}}$ and
its covariance with $\varepsilon_{i}^{1}=(\varepsilon_{it})_{t;it\in\Omega_{1}}$,
this optimal unbiased estimator has weights $v=(v^{1},v^{0})$ with
$v^{1}=w_{1}$ and $v^{0}$ that solve
\begin{align}
\min_{v^{0}} & \sum_{i}v_{i}^{0\prime}\Sigma_{i}^{0}v_{i}^{0}+2v_{i}^{0\prime}\Sigma_{i}^{01}w_{i} & \text{s.t. } & w_{1}'Z_{1}+v^{0\prime}Z_{0}=0\label{eq:optimization}
\end{align}
If we know the variances $\Sigma_{i}$, then we can find the efficient
estimator by solving this straightforward quadratic minimization problem
subject to linear constraints. We note that the imputation step, i.e.
the solution for $v^{0}$, now depends on the estimand via $w_{it}$,
in contrast to the case of spherical errors, which implies $\Sigma_{i}^{01}=\Null$.

If we do \textit{not} have ex-ante knowledge about the variances,
then the efficient GLS estimator is not generally feasible since $\Sigma_{i}$
may vary across $i$ with non-\emph{iid} data. At the same time, we
may still be able to obtain a feasible estimator that is asymptotically
efficient (in the sense that its asymptotic variance converges to
the minimal asymptotic variance among unbiased linear estimators)
given a simpler auxiliary model of the variances. Assume that $\Sigma_{i}=\bar{\Sigma}_{i}(\mu)$
with some low-dimensional parameter $\mu$. If required for estimating
the variance model consistently, we may further assume that treatment
effects themselves follow a model $\tau_{it}=\bar{\tau}_{it}(\mu)$
(where the sets of parameters may or may not overlap). We then consider
the following procedure: First, estimate $\mu$ by some $\hat{\mu}$.
Then, plug in $\hat{\mu}$ to obtain estimates of the variance matrices
$\Sigma_{i}$, the resulting plug-in estimate of $v$ from \ref{eq:optimization},
and the plug-in imputation estimator of $\tau_{w}$.

Consistency and unbiasedness of this estimator do not rely on correctness
of the auxiliary model of either treatment effects or variances. Under
regularity conditions similar to those in \ref{prop:asymptoticnormality,prop:normality-sufficient},
and when variance estimates are consistent, we obtain an efficient
and robust estimator in the following sense. First, if the model is
correct, then the estimator is asymptotically efficient. Second, even
if the model is not correct, the estimator is still consistent and
asymptotically unbiased for $\tau_{w}$ under unrestricted treatment-effect
heterogeneity. We label it the imputation GLS estimator.

We give two examples. First, one can show that our construction yields
as a special case the feasible GLS estimator proposed in \textcite{Wooldridge2021},
which we obtain in the case with two-way FEs, no other covariates,
and a complete panel, with an auxiliary model $\Sigma_{i}=\bar{\Sigma}(E_{i})$,
$\tau_{it}=\bar{\tau}_{t}(E_{i})$ that allows for arbitrary variation
across cohorts, but requires homogeneity within.

Second, one may want to leverage a more restricted model of variances
to improve small-sample performance. A model which allows the variances
to vary arbitrarily across cohort\textendash periods may have too
many parameters to allow for an effective estimation of the variance
as soon as $T$ is moderately large, in which case the theoretical
efficiency gain may not materialize in smaller samples (similar to
the findings of \textcite{Marcus2020} on the practical problems with
estimating TWFE models using efficient GMM). In such cases, the approach
outlined above allows one to instead assume a variance model with
fewer parameters. In particular, one could model the errors as an
autoregressive process with homoskedastic innovations (e.g., AR(1)),
estimate the autoregression parameters, and achieve an efficiency
improvement.\footnote{The estimation of the autoregression parameters can be done in two
ways, either using untreated observations only without an auxiliary
model of treatment effect homogeneity, or in the entire sample with
such a model.} Using an AR(1) process is enough to bridge the gap between two benchmarks:
spherical errors, as in the main text of our paper, and errors following
a random walk, as in \textcite{harmon_DiD}. In the latter case, \textcite{harmon_DiD}
has recently derived the Stepwise Difference-in-Differences estimator
as the efficient one and shown that it coincides with the \textcite{DeChaisemartin2018}
estimator for horizon $h=0$.

\subsection{Imputation and Weight Representations for Efficient Unbiased Estimators\label{subsec:appx-weights}}

In this section, we provide several representations of unbiased estimators,
both efficient and not, including for the case when a non-trivial
treatment-effect model is imposed.

We first show that even when a non-trivial model $\tau=\Gamma\theta$
is imposed, the imputation result for unbiased estimators from \ref{prop:generalimputation}
applies with respect to an adjusted estimand.
\begin{prop}[Imputation representation of unbiased estimators with a non-trivial
treatment effect model]
\label{prop:model-imputation}Under \ref{assu:A1prime,assu:A2},
any linear estimator $\hat{\tau}_{w}$ of $\tau_{w}$ that is unbiased
when the model $\tau=\Gamma\theta$ is imposed can be written as a
linear estimator of some alternatively weighted estimand $\tau_{v}=\sum_{it\in\Omega_{1}}v_{it}\tau_{it}$
that is unbiased without restrictions on the treatment effects. In
particular, the imputation representation in \ref{prop:generalimputation}
still applies with $\hat{\tau}_{w}=\sum_{it\in\Omega_{1}}v_{it}\hat{\tau}_{it}$
in the third step. The weights $v_{1}=\left(v_{it}\right)_{it\in\Omega_{1}}$
satisfy $\Gamma'w_{1}=\Gamma'v_{1}$, such that $\tau_{v}=\tau_{w}$
when the model $\tau=\Gamma\theta$ is correct.
\end{prop}
We now provide explicit expressions for the weights implied by the
the efficient estimator $\hat{\tau}_{w}^{*}$, both with and without
\ref{assu:A3prime}.

\begin{prop}[Weight representation of efficient estimator]
\label{prop:olsweights}The efficient estimator from \ref{thm:OLS-BLUE}
can be represented as $\hat{\tau}_{w}^{*}=v^{*\prime}Y$ with the
weight vector $v^{\ast}=\left(v_{1}^{\ast\prime},v_{0}^{\ast\prime}\right)^{\prime}$
that satisfies \[v^{*}=\begin{psmallmatrix}\I-Z_{1}(Z'Z)^{-1}Z_{1}^{\prime}\\-Z_{0}(Z'Z)^{-1}Z_{1}^{\prime}\end{psmallmatrix}\Gamma(\Gamma'(\I-Z_{1}(Z'Z)^{-1}Z_{1}^{\prime})\Gamma)^{-1}\Gamma'w_{1}\]
and that does not depend on the realization of the $Y_{it}$. In the
special case of $\Gamma=\I_{N_{1}}$, $v_{1}^{*}=w_{1}$ and $v_{0}^{*}=-Z_{0}(Z_{0}'Z_{0})^{-1}Z_{1}'w_{1}$.
\end{prop}
With a non-trivial \ref{assu:A3prime}, we can characterize these
weights by a combination of variance minimization for the treated
observations and imputation for the untreated observations.
\begin{prop}[Characterization of weights in terms of imputation and variance minimization]
\label{prop:imputation-model}With a non-trivial model $\tau=\Gamma\theta$
for the treatment effects, the efficient estimator from \ref{thm:OLS-BLUE}
can be written as the efficient imputation estimator from \ref{thm:imputation}
under unrestricted heterogeneity with alternative weights $v_{1}^{*}$
on the treatment effects, which solve the variance-minimization problem
\begin{equation}
\min_{v_{1}}\:v_{1}^{\prime}\Phi^{-1}v_{1}\qquad\text{subject to }\Gamma'v_{1}=\Gamma'w_{1},\label{eq:weight-optimization}
\end{equation}
where $\Phi=\I_{N_{1}}-Z_{1}(Z'Z)^{-1}Z'_{1}$ is the variance of
the OLS estimator of $\tau$ in the case of spherical errors with
unit variance and unrestricted treatment-effect heterogeneity.
\end{prop}

\subsection{Low-Level Sufficient Asymptotic Conditions \label{subsec:appx-sufficient}}

In this section we develop low-level sufficient conditions for consistency,
asymptotic normality, and valid inference that directly restrict the
weights $w_{1}$ on treated observations and cohort sizes. We focus
on the case of a (possibly incomplete) panel with $I$ units and $T$
time periods with respective FEs and no other covariates. We first
state sufficient conditions for consistency in a panel where the number
of periods $T$ is allowed to grow slowly.
\begin{assumption}[Low-level sufficient conditions for consistency]
\label{assu:lowlevel-consistent}Assume that in the first period
every unit is observed and not treated, and that
\begin{enumerate}
\item $\sum_{i=1}^{I}\left(\sum_{t;D_{it}=1}|w_{it}|\right)^{2}\rightarrow0$,
i.e., the weights on treatment effects fulfill a (clustered) Herfindahl
condition;
\item $T\sum_{i=1}^{I}\left(\sum_{t;D_{it}=1}w_{it}\right)^{2}\rightarrow0$,
i.e. the concentration of unit net weights decays fast enough;
\item $T^{2}\sum_{t=2}^{T}\frac{(\sum_{i;D_{it}=1}w_{it})^{2}}{\sum_{i;D_{it}=0}1}\rightarrow0$,
i.e., the sum of squared total weight on observations treated at $t$
relative to the number of untreated observations in $t$ vanishes
sufficiently quickly.
\end{enumerate}
\end{assumption}
The first two conditions express that the weights do not concentrate
on too few units. They are similar, but not redundant unless $T$
is fixed; when some weights within a unit are negative and some positive,
the weights may cancel out within units, yielding the second condition
even when the first one is not fulfilled. The three conditions address
different sources of variation of the efficient estimator $\hat{\tau}_{w}^{*}$:
the first condition bounds the variation from the treated observations
themselves; the second, from estimating unit FEs from the untreated
observations; and the third, from estimating period FEs from the untreated
observations. Together with the conditions in the main text, \ref{assu:lowlevel-consistent}
yields consistency.
\begin{prop}[Consistency under low-level sufficient conditions]
\label{prop:consistency-sufficient} Suppose \ref{assu:A1,assu:A2,assu:clustered,assu:lowlevel-consistent}
hold and treatment effects are allowed to vary arbitrarily (trivial
\ref{assu:A3}). Then $\hat{\tau}_{w}^{*}$ from \ref{thm:OLS-BLUE}
is consistent for $\tau_{w}$.
\end{prop}
Next, we develop sufficient conditions that will imply asymptotic
normality and valid inference in the special case of a complete panel
with a fixed $T$.
\begin{assumption}[Low-level sufficient conditions for asymptotic normality and inference]
\label{assu:lowlevel-inference} The panel is complete, $T$ is fixed,
$\left|\tau_{it}\right|,$ $\left|\bar{\tau}_{it}\right|$ and $\expec{\varepsilon_{it}^{4}}$
are uniformly bounded, and
\begin{enumerate}
\item There is some constant $C$ such that for all $t$ and $i,j$ with
$E_{i}=E_{j}$, $D_{it}=1=D_{jt}$ and $w_{it}\neq0$, we have $w_{jt}\ne0$
and $\frac{|w_{it}|}{|w_{jt}|}\leq C$; that is, weights do not vary
too much within cohort\textendash periods;
\item $\sum_{i;E_{i}=e}1\rightarrow\infty$ for $e=2,\dots,T,\infty$; that
is, the size of all cohorts grows.\footnote{This assumption can be relaxed in the context of \ref{prop:normality-sufficient}
for some of the cohorts on which the estimand does not put any weight,
as long as we retain enough data to estimate period FEs consistently.}
\end{enumerate}
\end{assumption}
If we add the latter conditions, we also achieve asymptotic normality.
\begin{prop}[Asymptotic normality in short panels]
\label{prop:normality-sufficient} Suppose \ref{assu:A1,assu:A2,assu:clustered,assu:lowlevel-consistent,assu:lowlevel-inference}
hold, treatment effects are allowed to vary arbitrarily (trivial \ref{assu:A3}),
and the variance of $\hat{\tau}_{w}^{*}$ does not vanish too quickly,
$\liminf n_{H}\sigma_{w}^{2}>0$. Then $\hat{\tau}_{w}^{*}$ from
\ref{thm:OLS-BLUE} fulfills the conditions of \ref{prop:asymptoticnormality},
and therefore $\sigma_{w}^{-1}(\hat{\tau}_{w}-\tau_{w})\stackrel{d}{\rightarrow}\N(0,1).$
\end{prop}
Finally, these conditions are also sufficient for obtaining consistent
variance estimates.
\begin{prop}[Consistent variance estimation in short panels]
\label{prop:se-short-sufficient} Consider the estimator $\hat{\sigma}_{w}^{2}$
from \ref{thm:se}, and suppose \ref{assu:A1,assu:A2,assu:clustered,assu:lowlevel-consistent,assu:lowlevel-inference}
hold and treatment effects are allowed to vary arbitrarily (trivial
\ref{assu:A3}). Then $\hat{\tau}_{w}^{*}$ from \ref{thm:OLS-BLUE}
fulfills the conditions of \ref{thm:se} when $\tilde{\tau}_{it}$
are equal to an overall average or to cohort\textendash period averages
calculated as in equation \ref{eq:taubar-G}.
\end{prop}
We note that \ref{assu:lowlevel-consistent,assu:lowlevel-inference}
are fulfilled in particular in the common case of fixed $T$, growing
cohort sizes, and weights that are bounded in the sense that $\sum_{it\in\Omega_{1}}\left|w_{it}\right|<C$
and are constant within cohort\textendash period cells, $w_{it}=w_{jt}$
for $t\geq E_{i}=E_{j}$.

The above conditions apply to the case of unit and time fixed effects.
The weight conditions for consistency (\ref{prop:consistency}) and
asymptotic normality (\ref{prop:asymptoticnormality}) in the main
text can also be verified for more complicated models, such as when
there are unit-specific linear trends in addition to fixed effects.
As a concrete example, we consider the case of a complete panel of
$T=T_{0}+T_{1}$ periods with a single treated cohort of size $I_{1}$,
first treated at time $T_{0}+1$, along with a cohort of size $I_{0}$
of never-treated units. For now, suppose we are interested in the
average treatment effect in a specific period $t>T_{0}$ for the treated
cohort. With some unreported algebra we can calculate weights underlying
the imputation estimator with unit and period fixed effects and unit-specific
linear trends:
\[
v_{is}^{t}=\left(\one\left[s{=}t\right]-\frac{1{+}12\left(s{-}\frac{T_{0}{+}1}{2}\right)\left(t{-}\frac{T_{0}{+}1}{2}\right)/\left(T_{0}^{2}{-}1\right)}{T_{0}}\one\left[s{\le}T_{0}\right]\right)\cdot\left(\frac{1}{I_{1}}\one\left[E_{i}{=}T_{0}{+}1\right]-\frac{1}{I_{0}}\one\left[E_{i}{=}\infty\right]\right)
\]

One can check that these weights satisfy the bound $\sum_{i=1}^{I}\left(\sum_{s=1}^{T}|v_{is}^{t}|\right)^{2}\leq\left(\frac{6t}{T_{0}}\right)^{2}\left(\frac{1}{I_{0}}+\frac{1}{I_{1}}\right)$.
The weights therefore fulfill the Herfindahl condition of \ref{assu:Herfindahl}
(for which \ref{prop:consistency} establishes consistency) provided
that $I_{0},I_{1}\rightarrow\infty$ and $\frac{T}{T_{0}\sqrt{\min\{I_{0},I_{1}\}}}\rightarrow0$.
The same condition extends to \ref{prop:asymptoticnormality} (asymptotic
normality) and the consistency of the imputation estimator of convex
averages of the cohort treatment effects.

\subsection{Optimal Choices for Treatment Averages in Variance Estimation\label{subsec:appx-SE-Weights}}

The variance estimator in \ref{thm:se} is asymptotically conservative,
since it includes the variation $\sigma_{\tau}^{2}=\sum_{i}\left(\sum_{t;D_{it}=1}v_{it}(\tau_{it}-\bar{\tau}_{it})\right)^{2}\geq0$
of the treatment effects around their averages $\bar{\tau}_{it}$.
Here, we consider reasonable choices for the $\bar{\tau}_{it}$. As
discussed in \ref{subsec:Conservative-Inference}, a natural conservative
choice is to estimate a single average $\bar{\tau}_{it}\equiv\bar{\tau}$.
The$\bar{\tau}$ that minimizes $\sigma_{\tau}^{2}(\bar{\tau})=\sum_{i}\left(\sum_{t;D_{it}=1}v_{it}(\tau_{it}-\bar{\tau})\right)^{2}$
is

\begin{align}
\bar{\tau} & =\frac{\sum_{i}\left(\sum_{t;D_{it}=1}v_{it}\right)\left(\sum_{t;D_{it}=1}v_{it}\tau_{it}\right)}{\sum_{i}\left(\sum_{t;D_{it}=1}v_{it}\right)^{2}}.\label{eq:taubar-overall}
\end{align}
Indeed, $\sigma_{\tau}^{2}(\bar{\tau})$ is convex in $\bar{\tau}$,
and the first-order condition $0=\frac{\partial}{\partial\bar{\tau}}\sigma_{\tau}^{2}(\bar{\tau})$
locates the above solution. A natural estimator is its sample analog.\footnote{Also note that the denominator of \ref{eq:taubar-overall} is zero
if and only if the estimand makes only within-unit comparisons of
treatment effects over time; in that case the choice of $\bar{\tau}$
is inconsequential, as it cancels out in \ref{eq:SE}.}

When multiple group-wise averages are estimated, \ref{eq:taubar-G}
should be viewed as a heuristic extension of \ref{eq:taubar-overall}.
The optimal solution for $\tilde{\tau}_{g}$ is generally more complex,
as it may depend on treatment effect estimates and weights outside
the group. For concreteness, consider averages $\bar{\tau}_{et}$
that vary by cohort $E_{i}=e$ and period $t$, yielding excess variance
$\sigma_{\tau}^{2}=\sum_{i}\left(\sum_{t;D_{it}=1}v_{it}(\tau_{it}-\bar{\tau}_{E_{i}t})\right)^{2}$.
Then the first-order conditions $0=\frac{\partial}{\partial\bar{\tau}_{es}}\sum_{i}\left(\sum_{t;D_{it}=1}v_{it}(\tau_{it}-\bar{\tau}_{E_{i}t})\right)^{2}=-2\sum_{i;E_{i}=e}v_{is}\left(\sum_{t;D_{it}=1}v_{it}(\tau_{it}-\bar{\tau}_{et})\right)$
have to be solved for each cohort simultaneously across $t$, provided
that the estimator puts non-zero weight on multiple periods within
the same cohort. The exception is when only one period for every cohort
receives non-zero weight, as when estimating the ATT for a given number
of periods since treatment. In that situation the optimal solution
$\bar{\tau}_{et}=\frac{\sum_{i;E_{i}=e}v_{it}^{2}\tau_{it}}{\sum_{i;E_{i}=e}v_{it}^{2}}$
coincides with equation \ref{eq:taubar-overall}.

\subsection{Leave-Out Conservative Variance Estimation\label{subsec:appx-Leave-Out}}

Here we formalize the leave-out conservative variance estimator for
$\tau_{w}$, contrast it to leave-out variance estimators from prior
work, and provide a computationally efficient way of obtaining them.

As in equation \ref{eq:taubar-G}, suppose $\Omega_{1}$ is partitioned
into groups of treated observations given by $G_{g}$. Let $v_{ig}=\sum_{t;it\in G_{g}}v_{it}$
and $\hat{T}_{ig}=\left(\sum_{t;it\in G_{g}}v_{it}\hat{\tau}_{it}^{*}\right)/v_{ig}$
(with an arbitrary value if $v_{ig}=0$). Then our \emph{non}-leave-out
variance estimator is based on $\tilde{\tau}_{it}\equiv\tilde{\tau}_{g}=\frac{\sum_{j}v_{jg}^{2}\hat{T}_{jg}}{\sum_{j}v_{jg}^{2}}$for
$it\in G_{g}$. The leave-out version is defined as $\tilde{\tau}_{it}^{LO}=\frac{\sum_{j\ne i}v_{jg}^{2}\hat{T}_{jg}}{\sum_{j\ne i}v_{jg}^{2}}$.\footnote{It is well-defined whenever there are no groups in which only one
unit receives a non-zero total weight.}

Our leave-out strategy differs from leave-out variance estimation
procedures of \textcite{MacKinnonWhite1985} and \textcite{Kline2020}.
Those papers assume that the OLS parameter vector is still identified
when dropping individual units. In fact, Lemma 1 in \textcite{Kline2020}
shows that unbiased variance estimation is impossible outside that
case. Our results, in contrast, provide conservative inference in
models where that condition is violated. Indeed, the imputation estimator
of \ref{thm:imputation} is a special case of \ref{thm:OLS-BLUE}
for the model in which each treated observation gets its own treatment
effect parameter $\tau_{it}$. Naturally, $\tau_{it}$ cannot be unbiasedly
estimated without unit $i$ in the data, and thus the results from
\textcite{Kline2020} do not apply.

While computing $\tilde{\tau}_{it}^{LO}$ directly may be computationally
intensive, a more efficient procedure is available based on a simple
rescaling of residuals $\tilde{\varepsilon}_{it}$ in equation \ref{eq:SE}.
For $it\in G_{g}$ consider $\tilde{\varepsilon}_{it}^{LO}=\tilde{\varepsilon}_{it}\cdot\frac{1}{1-\left(v_{ig}^{2}/\sum_{j}v_{jg}^{2}\right)}$.
Then replacing $\tilde{\varepsilon}_{it}$ with $\tilde{\varepsilon}_{it}^{LO}$
in equation \ref{eq:SE} implements the leave-out adjustment (the
proof is based on straightforward algebra and available by request).
The leave-out variance estimator based on $\tilde{\varepsilon}_{it}^{LO}$
is available as an option in our \texttt{did\_imputation} command.

When all residuals use only out-of-cluster observations for estimating
$\delta$, too, the resulting variance estimator is exactly unbiased
for an upper bound on the true variance:
\begin{prop}
\label{prop:leaveout}Assume that $\hat{\tau}_{w}$ is unbiased for
$\tau_{w}$ and that the variance of $\hat{\tau}_{w}$ is estimated
by $\hat{\sigma}_{w}^{2}=\sum_{i}\left(\sum_{t;it\in\Omega}v_{it}\tilde{\varepsilon}_{it}^{LO}\right)^{2}$
with $\tilde{\varepsilon}_{it}^{LO}=Y_{it}-A_{i}^{\prime}\hat{\lambda}_{i}-X_{it}^{\prime}\hat{\delta}^{-i}-D_{it}\tilde{\tau}_{it}^{-i}$,
where $\hat{\delta}^{-i}$ and $\tilde{\tau}_{it}^{-i}$ are estimated
based on outcomes $Y_{js}$ with $j\neq i$ only, and $\hat{\delta}^{-i}$
is an unbiased estimator of $\delta$. Then $\expec{\hat{\sigma}_{w}^{2}}\geq\sigma_{w}^{2}$.
\end{prop}

\subsection{Computationally Efficient Calculation of $v_{it}^{\ast}$ Weights\label{subsec:appx-Weights-Imputation}}

In this section we provide a computationally efficient algorithm for
computing the weights $v_{it}^{\ast}$ corresponding to the \ref{thm:OLS-BLUE}
estimator. We first establish a general result about the weights underlying
any linear combination of estimates in any regression: that the weights
can themselves be represented as a linear combination of the regressors,
with certain coefficients. We then characterize a system of equations
for those coefficients and modify the iterative procedure for computing
OLS estimators with high-dimensional FEs to obtain them.
\begin{prop}
\textup{\label{lem:linear-est-weights}}Consider some scalar estimator
$\hat{\psi}_{w}=w^{\prime}\hat{\psi}$ obtained from an arbitrary
point-identified OLS regression $y_{j}=\psi^{\prime}z_{j}+\varepsilon_{j}$.
Like every linear estimator, it can be uniquely represented as $\hat{\psi}_{w}=v^{\prime}y$,
with $y$ collecting $y_{j}$ and with implied weights $v=\left(v_{j}\right)_{j}$
that do not depend on the outcome realizations. Then weights $v_{j}$
can be represented as a linear combination of $z_{j}$ in the sample,
i.e. $v_{j}=z_{j}^{\prime}\check{\psi}$ for some vector $\check{\psi}$,
the same for all $j$.
\end{prop}
We now apply this proposition to \ref{thm:OLS-BLUE}, with the general
model of $Y_{it}(0)=Z_{it}^{\prime}\pi+\varepsilon$. Then $\hat{\tau}_{w}^{\ast}=v^{\ast\prime}Y$,
where weights $v^{\ast}$ can be represented as $v_{it}^{\ast}=Z_{it}^{\prime}\check{\pi}+D_{it}\Gamma_{it}^{\prime}\check{\theta}.$
It remains to find the unknown $\check{\pi}$ and $\check{\theta}$
to obtain the $v^{\ast}$ weights. To do so, we use the properties
of $\hat{\tau}_{w}^{\ast}$. First, it equals to zero if $Y_{it}$
is linear in $Z_{it}$. Second, letting $\mu=\Gamma^{\prime}w_{1}$,
$\hat{\tau}_{w}^{\ast}=\mu_{j}$ if $Y_{it}=\Gamma_{it,j}D_{it}$
for all $it\in\Omega$, as in that case $\hat{\theta}_{j}=1$ and
$\hat{\theta}_{-j}=0$. Thus we have a system of equations which determine
$\check{\pi}$ and $\check{\theta}$:
\begin{align}
\sum_{it\in\Omega}Z_{it}\left(Z_{it}^{\prime}\check{\pi}+D_{it}\Gamma_{it}^{\prime}\check{\theta}\right) & =0; & \sum_{it\in\Omega_{1}}\Gamma_{it}\left(Z_{it}^{\prime}\check{\pi}+D_{it}\Gamma_{it}^{\prime}\check{\theta}\right) & =\Gamma^{\prime}w_{1}.\label{eq:iterative}
\end{align}

When $Z_{it}$ has a block structure in which some covariates are
FEs, solving this system iteratively is most convenient and computationally
efficient. For instance, suppose $Z_{it}^{\prime}\pi\equiv\alpha_{i}+X_{it}^{\prime}\delta$
and $\Gamma=\I_{N_{1}}$ (i.e. \ref{assu:A3prime} is trivial). Then
\ref{lem:linear-est-weights} implies $v_{it}^{\ast}=\check{\alpha}_{i}+X_{it}^{\prime}\check{\delta}+D_{it}^{\prime}\check{\theta}_{it}$
for all $it\in\Omega$, and the second part of \ref{eq:iterative}
simplifies to $v_{it}^{\ast}=w_{it}$ for all $it\in\Omega_{1}$.
Using this and the structure of $Z_{it}$, we rewrite the first part
of \ref{eq:iterative} as a system
\begin{align}
\sum_{t,\ it\in\Omega_{0}}\left(\check{\alpha}_{i}+X_{it}^{\prime}\check{\delta}\right) & =-\sum_{t,\ it\in\Omega_{1}}w_{it},\quad\text{for all }i;\label{eq:it-simple1}\\
\sum_{it\in\Omega_{0}}X_{it}\left(\check{\alpha}_{i}+X_{it}^{\prime}\check{\delta}\right) & =-\sum_{it\in\Omega_{1}}X_{it}w_{it}.\label{eq:it-simple2}
\end{align}
This system suggests an algorithm similar to iterative OLS \parencite[e.g.][]{Guimaraes2010}:
\begin{enumerate}
\item Given a guess of $\check{\delta}$, set $\check{\alpha}_{i}$ for
each unit to satisfy \ref{eq:it-simple1};
\item Given $\check{\alpha}_{i}$, set $\check{\delta}$ to satisfy \ref{eq:it-simple2};
\item Repeat until convergence.
\end{enumerate}

\subsection{Monte-Carlo Simulation\label{subsec:Monte-Carlo-BP-Appx-NEW}}

We now quantify the efficiency properties of the imputation estimator
in a simulation based on the BP application, both under homoskedastic,
serially uncorrelated error terms and without those assumptions. We
compare the imputation estimator to the alternative robust estimators
of \textcite{DeChaisemartin2020} and \textcite{Abraham2018}, as
well as to conventional dynamic specifications in terms of efficiency
and bias. We also verify correct coverage of our inference procedure
and check sensitivity of different estimators to anticipation effects.

\paragraph{Setting.}

Our baseline simulation uses the sample and treatment timing from
the application of \ref{sec:Application}; in line with \ref{sec:Setting},
we view them as non-stochastic. Our target estimands are the ATTs
$\tau_{h}$ for each horizon $h=0,\dots,11$ (in dollars); while $\tau_{0}$
is an average of the short-run effects on 21,545 units, $\tau_{11}$
corresponds to 1,498 units only. 

In order to simulate data, we obtain estimates of treatment effects
and residuals from the data. We first apply the imputation estimator
to the observed data (without disbursement method FEs) to obtain $\hat{\tau}_{it}$
and, for the untreated observations, residuals $\tilde{\varepsilon}_{it}$,
as in \ref{eq:SE}. We then compute $\tilde{\tau}_{it}$ as the average
of $\hat{\tau}_{it}$ within each cohort-period cell (using projection
weights) for horizons $h=0,\dots,3$. For $h>3$ we set $\tilde{\tau}_{it}=0$,
based on the main draft's findings that the MPXs decay quickly. We
set $\tilde{\varepsilon}_{it}=\hat{\tau}_{it}-\tilde{\tau}_{it}$
for treated observations, attributing the within-cell variation in
treatment effects to the residual.\footnote{This procedure almost coincides with that of \ref{subsec:Conservative-Inference},
except that we average $\hat{\tau}_{it}$ with projection weights,
rather than $v_{it}^{2}$ weights as in \ref{eq:taubar-G}. This is
to ensure that the true ATT in the simulation equals the imputation
estimate in real data. To analyze conventional estimators we also
need to choose $\tilde{\tau}_{it}$ and $\tilde{\varepsilon}_{it}$
for periods $t\ge30$ when all households are already treated and
$\hat{\beta}_{t}$ cannot be estimated, although $\hat{\alpha}_{i}$
can. There we set $\tilde{\tau}_{it}$ equal to the horizon-specific
imputation estimate for horizons $h=0,\dots,3$ and $\tilde{\tau}_{it}=0$
for $h>3$. Then we choose $\hat{\beta}_{t}$ such that $\tilde{\varepsilon}_{it}=Y_{it}-\hat{\alpha}_{i}-\hat{\beta}_{t}-\tilde{\tau}_{it}$
average to zero for each $t\ge30$.} In the simulations we use $\tilde{\tau}_{it}$ as the true treatment
effects and consider several DGPs for the error terms.\footnote{Note that the choice of $\tilde{\tau}_{it}$ does not affect the comparison
among estimators robust to treatment effect heterogeneity since they
are unbiased for the same estimand. The homogeneity of treatment effects
within cohort\textendash period cells also makes the standard errors
for the imputation estimator comparable to those of alternative robust
estimators (see also \ref{fn:comparable estimators}).} In the baseline version, the error terms are spherical with the variance
equal to $\sigma_{\tilde{\varepsilon}}^{2}$, the sample variance
of $\tilde{\varepsilon}_{it}$.

We compare the efficiency of the imputation estimator to those of
the \textcite{DeChaisemartin2020} and \textcite{Abraham2018} estimators
(denoted dCDH and SA, respectively). Two versions of the \textcite{Callaway2018}
estimator are equivalent to dCDH and SA, respectively, in this setting
with no additional covariates. Importantly, the estimands are exactly
the same for all three robust estimators we consider. We further consider
two versions of OLS estimators: semi-dynamic, which includes all lags
but no leads of treatment, and fully-dynamic, which further includes
all leads except $h=-1$ and $h=-5$.\footnote{We also considered dropping $h{=}{-}1$ and $h{=}{-}2$, as in \ref{fig:OLS underidentification},
but that was \emph{very} inefficient and is not reported.} For each estimator, we compute the underlying weights $v_{it}$ and
use them to calculate the exact properties of the estimators, such
as their finite-sample variance. For inference on the imputation estimator,
we use the results from \ref{subsec:Conservative-Inference}. In the
absence of treatment effect heterogeneity within cohort-period cells,
inference is asymptotically exact rather than conservative.

\paragraph{Results.}

We first compare the imputation estimator to its robust counterparts.
Panel 1a of \ref{fig:Monte-Carlo-Simulation} reports the exact standard
deviation of each estimator, for each horizon-specific estimand, and
Panel 2a shows it relative to the imputation estimator. In line with
\ref{thm:OLS-BLUE}, the imputation estimator is most efficient among
the robust estimators for all horizons, but the simulation is useful
in quantifying the magnitude of the efficiency gain. Under homoskedasticity,
the standard deviations (SDs) of the dCDH estimator are 35\textendash 39\%
higher than the SD of the imputation estimator, implying that 82\textendash 92\%
more units would be needed to obtain confidence intervals of a similar
length if these estimators are used. SDs of the SA estimator, which
uses a smaller reference group, are much higher yet: 39\textendash 290\%
higher than with imputation.\footnote{The earliest treatment happens in period 14, and dCDH and SA estimators
ignore all data prior to period 13, while the imputation estimator
uses all pre-periods. Yet, efficiency gains persist even if periods
1\textendash 12 are dropped. In unreported simulations, the SDs of
dCDH (SA) are still 19\textendash 35\% (32\textendash 245\%) higher
than for imputation.}

In Panels 1b\textendash 1d and 2b\textendash 2d of \ref{fig:Monte-Carlo-Simulation},
we report estimator SDs under deviations from spherical errors, such
that the relative efficiency of the imputation estimator is no longer
guaranteed by \ref{thm:OLS-BLUE}. In Panels 1b and 2b we make the
error terms heteroskedastic in a way implied by the data (while still
mutually independent): $\varepsilon_{it}\sim\mathcal{N}(0,\sigma_{it}^{2})$
for $\sigma_{it}^{2}=\tilde{\varepsilon}_{it}$. In Panels 1c and
2c we instead suppose that $\varepsilon_{it}$ follow a stationary
AR(1) process with $\var{\varepsilon_{it}}=\sigma_{\tilde{\varepsilon}}^{2}$
and $\cov{\varepsilon_{it},\varepsilon_{it'}}=0.5^{\lvert t-t'\rvert}$,
with $\varepsilon_{it}$ still normally distributed and independent
across units. In Panels 1d and 2d we finally learn the patterns of
both heteroskedasticity and serial correlation from the data by using
wild clustered bootstrap, i.e. setting $\varepsilon_{it}=z_{i}^{\ast}\tilde{\varepsilon}_{it}$
for $z_{i}^{\ast}$ drawn independently across households and taking
values $\pm1$ with equal probabilities. The imputation estimator
remains the most efficient of the three, with SDs of dCDH (SA) higher
by 32\textendash 56\% (48\textendash 264\%) in Panel 2b, 13\textendash 38\%
(43\textendash 343\%) in Panel 2c, and 34\textendash 66\% (42\textendash 249\%)
in Panel 2d. The only exception is for $h=0$ with AR(1) errors, where
the dCDH estimator has a 2\% lower SD.

Panels 1\textendash 2 of \ref{fig:Monte-Carlo-Simulation} also compare
the imputation estimator to conventional OLS estimators, showing that
the robustness to treatment effect heterogeneity associated with the
imputation estimator comes at a nearly zero efficiency cost for the
shorter horizons, $h=0,1,2$. Across the four DGPs, the SD of the
semi-dynamic OLS estimator is at most 13\% lower than that of the
imputation estimator. Moreover, the SD of fully-dynamic OLS is in
most cases \emph{higher }than that of the imputation estimator, as
the former does not fully impose \ref{assu:A2}. For longer horizons,
the efficiency advantage of the semi-dynamic estimator is larger.
These efficiency gains come at a cost of a small bias, reported in
Panel 3; the bias is guaranteed to be zero for the imputation estimator.
Because of the bias, the advantage of the semi-dynamic estimator shrinks
slightly, when measured by the root mean squared error in Panel 4.

In Panel 5 we consider the sensitivity of different estimators to
anticipation effects. We add an anticipation effect of \$10 to the
outcomes of each household in the period right before treatment, $t=E_{i}-1$,
and report the exact bias of each estimator due to this. We find the
imputation estimator to be uniformly \emph{less} sensitive than its
robust alternatives, although the semi-dynamic OLS specification has
an even smaller bias.\footnote{\label{fn:no_ranking}We do not, however, view the lower sensitivity
of the imputation estimator as its general feature. In unreported
results we find the imputation estimator to be more sensitive for
most horizons than dCDH and SA to anticipation effects at $t=E_{i}-2$.
\ref{lem:Equal-sensitivity} proves that there cannot be a clear ranking
between robust estimators in terms of sensitivity to various types
of anticipation effects.}

Finally, we verify that the inference procedure for the imputation
estimator proposed in \ref{subsec:Conservative-Inference} performs
well. Panel 6 reports its simulated coverage: the fraction of the
1,000 simulations in which a $t$-test does not reject the null of
$\tau_{h}$ taking its ``true'' value at the 5\% significance level
(as implemented via our Stata command \texttt{did\_imputation}), across
the four DGPs for the errors. The rejection rate is close to 5\% across
all horizons and DGPs.

Taken together, these results suggest that the imputation estimator
has sizable efficiency advantages over alternative robust estimators,
extending to heteroskedasticity and serial correlation of error terms.
The analysis further highlights that the efficiency gains do not come
at a cost of systematically higher sensitivity to parallel-trend violations.
Moreover, our analytical inference tools perform well in finite samples,
and the efficiency costs relative to dynamic OLS estimators are low,
except for long-run effects. Naturally, these results may be specific
to the data-generating processes we considered, and we recommend that
researchers perform similar simulations based on their data.

\subsection{Equal Sensitivity of Robust Estimators to Linear Pre-Trends\label{subsec:appx-Equal-Sensitivity}}
\begin{prop}
\label{lem:Equal-sensitivity}Suppose there are no never-treated units.
Then all linear estimators $\hat{\tau}_{w}$ of $\tau_{w}$ that are
unbiased under \ref{assu:A1,assu:A2} with a trivial \ref{assu:A3},
and thus robust to arbitrary treatment effect heterogeneity, have
the same sensitivity to linear anticipation trends. Specifically,
if $Y_{it}=\left(\kappa_{0}+\kappa_{1}K_{it}\right)\cdot\one\left[D_{it}=0\right]$
for some $\kappa_{0},\kappa_{1}\in\mathbb{R}$, then $\expec{\hat{\tau}_{w}}=-\sum_{it\in\Omega_{1}}w_{it}\left(\kappa_{0}+\kappa_{1}K_{it}\right)$.
\end{prop}
This result implies that there cannot be a general ranking in the
sensitivity of different estimators to anticipation effects. We formalize
this intuition by:
\begin{cor}
\label{lem:no_sensitivity_ranking}For an estimator $\hat{\tau}_{w}$
of $\tau_{w}$ that is unbiased under unrestricted treatment effect
heterogeneity and $y_{0}\in\mathbb{R}^{\left|\Omega_{0}\right|}$,
let $B_{\hat{\tau}_{w}}(y_{0})$ be its bias when $\expec{Y_{0}}=y_{0}$,
potentially violating \ref{assu:A1,assu:A2}. Consider two such linear
estimators, $\hat{\tau}_{w}^{A}$ and $\hat{\tau}_{w}^{B}$, and suppose
$\hat{\tau}_{w}^{A}$ is more biased for some $y_{0}$, $\lvert B_{\hat{\tau}_{w}^{A}}(y_{0})\rvert>\lvert B_{\hat{\tau}_{w}^{B}}(y_{0})\rvert$.
Then there exists $\tilde{y}_{0}\in\mathbb{R}^{\left|\Omega_{0}\right|}$
such that the comparison is reversed, $\lvert B_{\hat{\tau}_{w}^{B}}(\tilde{y}_{0})\rvert>\lvert B_{\hat{\tau}_{w}^{A}}(\tilde{y}_{0})\rvert$.
\end{cor}

\section{Proofs\label{sec:appx-proofs}}

In this appendix, we collect proofs for the results in the main text
and in the appendix. We first restate our general matrix notation
for convenience. Specifically, we stack the vectors $\lambda_{i}$
into a single vector $\lambda=(\lambda_{i})_{i}$. We set $Z_{it}=\begin{psmallmatrix}(\one\left[i=j\right]A_{jt})_{j}\\X_{it}\end{psmallmatrix},\pi=\begin{psmallmatrix}\lambda\\\delta\end{psmallmatrix}$
to summarize the nuisance component of the model. In matrix-vector
notation, we write $Y$ for the vector of outcomes, $Z=(A,X)$ for
the covariate matrix, $D$ for the matrix of indicators for treated
units, $\varepsilon$ for the vector of error terms, and $\Sigma=\var{\varepsilon}$
for their variance. We write $Y_{1},Z_{1}=(A_{1},X_{1}),D_{1},\varepsilon_{1}$
for the rows corresponding to treated observations ($it\in\Omega_{1}$);
in particular, $D_{1}=\I$. Analogously, we write $Y_{0},Z_{0}=(A_{0},X_{0}),D_{0},\varepsilon_{0}$
for the rows corresponding to untreated observations ($it\in\Omega_{0}$);
in particular, $D_{0}=\mathbb{O}$. We write $\tau=(\tau_{it})_{it\in\Omega_{1}}$
for the vector of treatment effects of the treated units, $\theta=(\theta_{m})_{m=1}^{N_{1}-M}$
for the vector of underlying parameters, $\Gamma=(\Gamma_{it,j})_{it\in\Omega_{1},j\in\{1,\ldots,N_{1}-M\}}$
for the matrix linking the two, and $w_{1}=(w_{it})_{it\in\Omega_{1}}$
for the weight vector. Then we can write model and estimand as
\[
Y=Z\pi+D\tau+\varepsilon,\qquad\tau=\Gamma\theta,\qquad\tau_{w}=w_{1}'\tau,
\]
where $\expec{\varepsilon}=\0$, $\var{\varepsilon}=\Sigma$, and
$\Sigma$ has block structure according to units $i$. For unit $i$,
we write $A_{i}=(A_{it})_{t},X_{i}=(X_{it})_{t},Y_{i}=(Y_{it})_{t},\varepsilon_{i}=(\varepsilon_{it})_{t},v_{i}=(v_{it})_{t}$
and denote by $\Sigma_{i}=\var{\varepsilon_{i}}$ the within-unit
variance\textendash covariance matrix of error terms.

\subsection{Proofs of Results from Main Text}
\begin{proof}[Proof of \ref{prop:Underid}]
In the absence of never-treated units and defining $\tau_{-1}=0$,
we can write $\sum_{h\ne-1}\tau_{h}\one\left[K_{it}=h\right]=\tau_{K_{it}}.$
Now consider some collection of $\tau_{h}$ (with $\tau_{-1}=0$)
and FEs $\tilde{\alpha}_{i}$ and $\tilde{\beta}_{t}$. For any $\kappa\in\R$,
let $\tau_{h}^{\star}=\tau_{h}+\kappa\left(h+1\right)$, $\tilde{\alpha}_{i}^{\star}=\tilde{\alpha}_{i}+\kappa\left(E_{i}-1\right)$,
and $\tilde{\beta}_{t}^{\star}=\tilde{\beta}_{t}-\kappa t$. Then
for any observation $it$, $\tilde{\alpha}_{i}^{\star}+\tilde{\beta}_{t}^{\star}+\tau_{K_{it}}^{\star}=\tilde{\alpha}_{i}+\tilde{\beta}_{t}+\tau_{K_{it}}+\kappa\left(E_{i}-1\right)-\kappa t+\kappa\left(t-E_{i}+1\right)=\tilde{\alpha}_{i}+\tilde{\beta}_{t}+\tau_{K_{it}}$,
and equation \ref{eq:dynamicOLS} has exactly the same fit under the
original and modified FEs and $\tau_{h}$ coefficients, indicating
perfect collinearity.
\end{proof}
\begin{proof}[Proof of \ref{prop:OLS-linear}]
By the Frisch\textendash Waugh\textendash Lovell theorem, $\tau^{\text{static}}$
can be obtained by a regression of $\expec{Y_{it}}=\alpha_{i}+\beta_{t}+\tau_{it}D_{it}$
on $\tilde{D}_{it}$ (without a constant), where $\tilde{D}_{it}=D_{it}-\check{\alpha}_{i}-\check{\beta}_{t}$
are the residuals from the auxiliary regression of $D_{it}$ on the
unit and period FEs. Thus, $\tau^{\text{static}}=\frac{\sum_{it\in\Omega}\tilde{D}_{it}\left(\alpha_{i}+\beta_{t}+\tau_{it}D_{it}\right)}{\sum_{it\in\Omega}\tilde{D}_{it}^{2}}.$
We have $\sum_{it\in\Omega}\tilde{D}_{it}\alpha_{i}=\sum_{i}\alpha_{i}\sum_{t;it\in\Omega}\tilde{D}_{it}=0$
because the residuals in the auxiliary regression are orthogonal to
all unit indicators. Analogously, $\sum_{it\in\Omega}\tilde{D}_{it}\beta_{t}=0$.
Defining 
\begin{equation}
w_{it}^{\text{static}}=\frac{\tilde{D}_{it}}{\sum_{it\in\Omega}\tilde{D}_{it}^{2}},\label{eq:weights-FWL}
\end{equation}
we have that $\tau^{\text{static}}=\sum_{it\in\Omega}w_{it}^{\text{static}}\tau_{it}D_{it}=\sum_{it\in\Omega_{1}}w_{it}^{\text{static}}\tau_{it}$,
as required. Clearly, $w_{it}^{\text{static}}$ do not depend on the
outcome realizations. Moreover, these weights add up to one because
we have that $\sum_{it\in\Omega_{1}}\tilde{D}_{it}=\sum_{it\in\Omega}\tilde{D}_{it}D_{it}=\sum_{it\in\Omega}\tilde{D}_{it}\left(\tilde{D}_{it}+\check{\alpha}_{i}+\check{\beta}_{t}\right)=\sum_{it\in\Omega}\tilde{D}_{it}^{2}$,
where the last equality holds because $\tilde{D}_{it}$ are orthogonal
to the unit and period FEs.
\end{proof}
\begin{proof}[Proof of \ref{prop:example-static}]
We use the characterization of the static TWFE specification weights
in equation \ref{eq:weights-FWL}. Given the complete panel, the regression
of $D_{it}$ on TWFE produces residuals $\tilde{D}_{it}=D_{it}-\bar{D}_{i\cdot}-\bar{D}_{\cdot t}+\bar{D}_{\cdot\cdot}$,
where $\bar{D}_{i\cdot}=\frac{1}{3}\sum_{t=1}^{3}D_{it}$, $\bar{D}_{\cdot t}=\frac{1}{2}\sum_{i=A,B}D_{it}$,
and $\bar{D}_{\cdot\cdot}=\frac{1}{6}\sum_{i=A,B}\sum_{t=1}^{3}D_{it}$
\parencite{DeChaisemartin2018}. Plugging in $\bar{D}_{A\cdot}=2/3$,
$\bar{D}_{B\cdot}=1/3$, $\bar{D}_{\cdot1}=0$, $\bar{D}_{\cdot2}=1/2$,
$\bar{D}_{\cdot3}=1$, and $\bar{D}_{\cdot\cdot}=1/2$, and computing
$\sum_{it\in\Omega_{1}}\tilde{D}_{it}=1/3$, we have
\begin{align*}
\hat{\tau}^{\text{static}} & =\frac{\sum_{it\in\Omega}\tilde{D}_{it}Y_{it}}{\sum_{it\in\Omega_{1}}\tilde{D}_{it}}=\left(Y_{A2}-Y_{B2}\right)-\frac{1}{2}\left(Y_{A1}-Y_{B1}\right)-\frac{1}{2}\left(Y_{A3}-Y_{B3}\right).
\end{align*}
Similarly, the static estimand equals $\tau^{\text{static}}=\frac{\sum_{it\in\Omega_{1}}\tilde{D}_{it}\tau_{it}}{\sum_{it\in\Omega_{1}}\tilde{D}_{it}}=\tau_{A2}-\frac{1}{2}\left(\tau_{A3}-\tau_{B3}\right).$
\end{proof}
\begin{proof}[Proof of \ref{prop:anynegweights}]
Let $I$ be the total number of units and $I_{\text{ever}}$ be that
of ever-treated units (those treated by $t=T$). As in the proofs
of \ref{prop:OLS-linear,prop:example-static}, $w_{it}^{\text{static}}=\tilde{D}_{it}/\sum_{it\in\Omega}\tilde{D}_{it}^{2}$
with 
\begin{equation}
\tilde{D}_{it}=D_{it}-\bar{D}_{i\cdot}-\bar{D}_{\cdot t}+\bar{D},\label{eq:completepanel_weight}
\end{equation}
for $\bar{D}_{i\cdot}=\frac{1}{T}\sum_{t=1}^{T}D_{it}=\frac{T-\left(E_{i}-1\right)}{T}$,
$\bar{D}_{\cdot t}=\frac{1}{I}\sum_{i=1}^{I}D_{it}=\frac{\sum_{i}\one\left[E_{i}\le t\right]}{I}$,
and $\bar{D}_{\cdot\cdot}=\frac{1}{IT}\sum_{i,t}D_{it}=\frac{I_{\text{ever}}\left(T-\left(E_{\text{first}}-1\right)\right)-N_{0}^{\ast}}{IT}$,
where the last expression holds because treated observations are those
belonging to ever-treated units in periods since $E_{\text{first}}$,
excluding the $N_{0}^{\ast}$ pre-treatment observations.

Since $\bar{D}_{i\cdot}$ monotonically declines in $E_{i}$ and $\bar{D}_{\cdot t}$
increases in $t$, equation \ref{eq:completepanel_weight} implies
that the lowest weight on any treated observation corresponds to $E_{i}=E_{\text{first}}$
and $t=T$. Thus, there is no negative weighting if and only if $\tilde{D}_{it}\ge0$
for such observations. Considering equation \ref{eq:completepanel_weight}
for one of those observations and using $\bar{D}_{\cdot T}=\frac{I_{\text{ever}}}{I}$,
there is no negative weighting if and only if 
\begin{align*}
0 & \le1-\frac{T-\left(E_{\text{first}}-1\right)}{T}-\frac{I_{\text{ever}}}{I}+\frac{I_{\text{ever}}\left(T-\left(E_{\text{first}}-1\right)\right)-N_{0}^{\ast}}{IT}=\frac{N_{1}^{\ast}-N_{0}^{\ast}}{IT},
\end{align*}
where the last equality uses the definition of $N_{1}^{\ast}$. This
is equivalent to $N_{1}^{\ast}\ge N_{0}^{\ast}$, as required.
\end{proof}
\begin{proof}[Proof of \ref{prop:No-Longrun}]
For any observation $it$, $K_{it}=t-E_{i}\ge\bar{H}$ implies $t\ge E_{i}+\bar{H}\ge\min_{i}E_{i}+\bar{H}\ge\max_{i}E_{i}$.
Thus, all observations considered by the estimand correspond to the
periods in which all units are already treated. Consider one such
period $t^{\star}$ for which the total weights are non-zero, $\sum_{i\colon K_{it^{\star}}\ge\bar{H}}w_{it}\ne0$.
(It exists because all weights $w_{it}$ are assumed non-negative
and are not identically zero.) Then consider a data-generating process
in which $\beta_{t^{\star}}$ is replaced with $\beta_{t^{\star}}-\kappa$
for some $\kappa\ne0$ and $\tau_{it^{\star}}$ is replaced with $\tau_{it^{\star}}+\kappa$
for all $i$. This DGP is observationally equivalent in terms of the
observed $Y_{it}$ and continues to satisfy \ref{assu:A1,assu:A2}.
Yet, the estimand differs by any arbitrary $\kappa\sum_{i\colon K_{it^{\star}}\ge\bar{H}}w_{it}\ne0$,
and is therefore not identified.
\end{proof}
\begin{proof}[Proof of \ref{thm:OLS-BLUE}]
The result is a consequence of the Gauss\textendash Markov theorem.
By itself, the Gauss\textendash Markov theorem establishes the efficiency
of the OLS estimator for $\theta$. Here, we extend efficiency of
the estimation of $\theta$ to efficiency for the estimation of the
weighted sum of treatment effects $\tau_{w}=w_{1}'\Gamma\theta$.
By construction, $\hat{\tau}_{w}^{*}=w_{1}^{\prime}\Gamma\hat{\theta}^{*}$
for the OLS estimator $\hat{\theta}^{*}$ of $\theta$. For every
linear estimator $\hat{\tau}_{w}$ that is unbiased for $\tau_{w}$
for all $\theta$ there is a linear unbiased estimator $\hat{\theta}$
of $\theta$ (with variance $\Sigma_{\hat{\theta}}$) for which $\tilde{\tau}_{w}=w_{1}'\Gamma\hat{\theta}$,
for example the estimator $\hat{\theta}=\hat{\theta}^{*}+\Gamma'w_{1}(w_{1}'\Gamma\Gamma'w_{1})^{-1}(\hat{\tau}_{w}-w_{1}'\Gamma\hat{\theta}^{*})$
(where we assume that $\Gamma'w_{1}\neq\0$, since the estimand and
efficient estimator are otherwise zero). Indeed, 
\begin{align*}
w_{1}'\Gamma\hat{\theta} & =w_{1}'\Gamma\hat{\theta}^{*}+w_{1}'\Gamma\Gamma'w_{1}(w_{1}'\Gamma\Gamma'w_{1})^{-1}(\hat{\tau}_{w}-w_{1}'\Gamma\hat{\theta}^{*})=\hat{\tau}_{w},\\
\expec{\hat{\theta}} & =\expec{\hat{\theta}^{*}}+\Gamma'w_{1}(w_{1}'\Gamma\Gamma'w_{1})^{-1}\left(\expec{\hat{\tau}_{w}}{-}\expec{w_{1}'\Gamma\hat{\theta}^{*}}\right)=\theta+\Gamma'w_{1}(w_{1}'\Gamma\Gamma'w_{1})^{-1}(w_{1}'\Gamma\theta{-}w_{1}'\Gamma\theta)=\theta.
\end{align*}
Under spherical errors, the OLS estimator $\hat{\theta}^{*}$ is the
best linear unbiased estimator for $\theta$ in the regression $Y=Z\pi+D\Gamma\theta+\varepsilon$,
with variance $\Sigma_{\hat{\theta}^{*}}$ that is minimal (in the
partial ordering implied by positive semi-definiteness) among the
variances of linear unbiased estimators of $\theta$ by Gauss\textendash Markov.
Hence, $\Var(w_{1}^{\prime}\Gamma\hat{\theta}^{*})-\Var(w_{1}^{\prime}\Gamma\hat{\theta})=w_{1}'\Gamma(\Sigma_{\hat{\theta}^{*}}-\Sigma_{\hat{\theta}})\Gamma'w_{1}\leq0$,
establishing efficiency. The efficient linear estimator of $\tau_{w}$
is also unique; indeed, if there was some unbiased linear estimator
$\hat{\tau}_{w}$ with $\var{\tilde{\tau}_{w}}=\var{\hat{\tau}_{w}^{*}}$
but $\expec{(\hat{\tau}_{w}-\hat{\tau}_{w}^{\ast})^{2}}>0$ (and thus
$\cov{\hat{\tau}_{w},\hat{\tau}_{w}^{\ast}}<\var{\hat{\tau}_{w}^{\ast}}$),
then $\frac{\hat{\tau}_{w}^{*}+\hat{\tau}_{w}}{2}$ would be an unbiased
linear estimator with lower variance.

It remains to argue unbiasedness in the heteroskedastic case. Since
the OLS estimator $\hat{\theta}^{*}$ remains unbiased for $\theta$
even without spherical errors, so does $\hat{\tau}_{w}^{*}=w_{1}'\Gamma\hat{\theta}^{*}$
for $\tau_{w}=w_{1}'\Gamma\theta$.
\end{proof}
\begin{proof}[Proof of \ref{thm:imputation}]
The efficient estimator from \ref{thm:OLS-BLUE} is obtained from
the OLS estimator $\hat{\tau}^{*}$ of $\tau$ in $Y=Z\pi+D\tau+\varepsilon$
by setting $\hat{\tau}_{w}^{*}=w_{1}'\hat{\tau}^{*}$. We now show
that the OLS estimator $\hat{\tau}_{w}^{\ast}$ has the desired imputation
form. By Frisch\textendash Waugh\textendash Lovell applied to residualization
of $Y$ and $Z$ with respect to $D$, the OLS estimate $\hat{\pi}^{*}$
of $\pi$ in the linear regression $Y=Z\pi+D\tau+\varepsilon$ is
the same as the estimate of $\pi$ in the linear regression $Y_{0}=Z_{0}\pi+\varepsilon$
restricted to $\Omega_{0}$. Indeed, \[\I_{N}-D(D'D)^{-1}D=\I_{N}-\begin{psmallmatrix}\I_{N_{1}} & \Null\\\Null & \Null\end{psmallmatrix}=\begin{psmallmatrix}\Null & \Null\\\Null & \I_{N_{0}}\end{psmallmatrix}\]and
thus $\hat{\pi}^{*}=(Z'(\I_{N}-D(D'D)^{-1}D)Z)^{-1}Z'(\I_{N}-D(D'D)^{-1}D)Y=(Z_{0}'Z_{0})^{-1}Z_{0}'Y_{0}$.

The OLS estimator $\hat{\tau}^{*}$ of $\tau$ in $Y=Z\pi+D\tau+\varepsilon$
is the same as the OLS estimator in the regression of $Y-Z\hat{\pi}^{*}$
on $D$. This is because $(\hat{\pi}^{*},\hat{\tau}^{*})$ minimize
the sum of squares $\|Y-Z\hat{\pi}-D\hat{\tau}\|^{2}$ over choices
of $(\hat{\pi},\hat{\tau})$, and $\hat{\tau}^{*}$ therefore minimizes
$\|Y-Z\hat{\pi}^{*}-D\hat{\tau}\|^{2}$ over $\hat{\tau}$ given $\hat{\pi}^{*}$.
We can therefore write $\hat{\tau}^{*}=(D'D)^{-1}D'(Y-Z\hat{\pi}^{*})=Y_{1}-Z_{1}\hat{\pi}^{*}$,
which has the desired imputation form.
\end{proof}
\begin{proof}[Proof of \ref{prop:generalimputation}]
 As in the proof of \ref{thm:OLS-BLUE}, there exists an unbiased
linear estimator $\hat{\tau}$ of $\tau$ such that $\hat{\tau}_{w}=w_{1}^{\prime}\hat{\tau}$.
We now construct a linear estimator $\hat{C}$ that is unbiased for
$Z_{1}\pi$, does not depend on $Y_{1}$, and yields $\hat{\tau}_{w}=w_{1}'(Y_{1}-\hat{C})$.
To this end, let $\hat{C}=Y_{1}-\hat{\tau}$. Then $\hat{C}$ is a
linear estimator with $\E[\hat{C}]=\E[Y_{1}]-\E[\hat{\tau}]=Z_{1}\pi$.
Since $\hat{C}$ is linear, we can write $\hat{C}=UY_{1}+VY_{0}$
for matrices $U,V$. Since $Z_{1}\pi=\expec{\hat{C}}=U\expec{Y_{1}}+V\expec{Y_{0}}=U\tau+(UZ_{1}+VZ_{0})\pi$
for all $\tau,\pi$, we must have $U=\Null$. Therefore, $\hat{C}$
satisfies the requirement of the proposition.
\end{proof}
\begin{proof}[Proof of \ref{prop:consistency}]
Writing $v_{i}=(v_{it})_{t}$, consistency follows from $\E\left[\hat{\tau}_{w}\right]=\tau_{w}$,
$\var{\hat{\tau}_{w}}=\sigma_{w}^{2}=\sum_{i=1}^{I}v_{i}'\Sigma_{i}v_{i}\leq\min\left\{ \sum_{i}\left(\sum_{t;it\in\Omega}|v_{it}|\right)^{2},R\left(\sum_{it\in\Omega}v_{it}^{2}\right)\right\} \bar{\sigma}^{2}\rightarrow0$.
Here the first case covers the condition $\sum_{i}\left(\sum_{t;it\in\Omega}|v_{it}|\right)^{2}\rightarrow0$
from \ref{assu:Herfindahl} and the second covers the alternative
condition $R\left(\sum_{it\in\Omega}v_{it}^{2}\right)\rightarrow0$
from \ref{fn:R-criterion}, where we use $v_{i}'\Sigma_{i}v_{i}\leq\|v\|^{2}\cdot(\text{max eigenvalue of \ensuremath{\Sigma_{i}}})\leq\|v\|^{2}\cdot R\cdot\bar{\sigma}^{2}$.
\end{proof}
\begin{proof}[Proof of \ref{prop:asymptoticnormality}]
Write $\hat{\tau}_{w}-\tau_{w}=\sum_{it\in\Omega}v_{it}\varepsilon_{it}=\sum_{i}\zeta_{i}$
with $\zeta_{i}=v_{i}'\varepsilon_{i},\E[\zeta_{i}]=0,\Var(\zeta_{i})=v_{i}'\Sigma_{i}v_{i}.$
Write $p=2+\kappa$ and let $q$ be the solution to $\frac{1}{p}+\frac{1}{q}=1$
(so in particular $1<q<2<p$). Using H\"{o}lder's inequality to establish
$\sum_{t;it\in\Omega}|v_{it}|^{\frac{1}{q}}\:\left(|v_{it}|^{\frac{1}{p}}|\varepsilon_{it}|\right)\leq\left(\sum_{t;it\in\Omega}|v_{it}|^{\frac{q}{q}}\right)^{\frac{1}{q}}\left(\sum_{t;it\in\Omega}|v_{it}|^{\frac{p}{p}}|\varepsilon_{it}|^{p}\right)^{\frac{1}{p}}$
for any $i$, and using $\E\left[|\varepsilon_{it}|^{p}\right]\leq C$
and $\frac{p}{q}+1=p$, we have that
\begin{align*}
 & \E[|\zeta_{i}|^{2+\kappa}]=\expec{\lvert\sum_{t;it\in\Omega}v_{it}\varepsilon_{it}\rvert^{p}}\leq\expec{\Big(\sum_{t;it\in\Omega}|v_{it}\varepsilon_{it}|\Big)^{p}}=\expec{\Big(\sum_{t;it\in\Omega}|v_{it}|^{\frac{1}{q}}\:|v_{it}|^{\frac{1}{p}}|\varepsilon_{it}|\Big)^{p}}\\
 & \leq\Big(\sum_{t;it\in\Omega}|v_{it}|\Big)^{\frac{p}{q}}\sum_{t;it\in\Omega}|v_{it}|\:\expec{|\varepsilon_{it}|^{p}}\leq\Big(\sum_{t;it\in\Omega}|v_{it}|\Big)^{\frac{p}{q}+1}C=\Big(\sum_{t;it\in\Omega}|v_{it}|\Big)^{p}C.
\end{align*}
Hence,
\begin{align*}
\frac{\sum_{i}\E[|\zeta_{i}|^{2+\kappa}]}{\Big(\sum_{i}\Var(\zeta_{i})\Big)^{\frac{2+\kappa}{2}}} & =\frac{\sum_{i}\E[|\zeta_{i}|^{2+\kappa}]}{\sigma_{w}^{2+\kappa}}\leq\frac{\sum_{i}\Big(\sum_{t;it\in\Omega}|v_{it}|\Big)^{2+\kappa}C}{\sigma_{w}^{2+\kappa}}=\frac{\wnorm^{2+\kappa}}{\sigma_{w}^{2+\kappa}}\sum_{i}\Big(\frac{\sum_{t;it\in\Omega}|v_{it}|}{\wnorm}\Big)^{2+\kappa}C\rightarrow0,
\end{align*}
where we have used that $\limsup\wnorm^{2}/\sigma_{w}^{2}<\infty$
and that $\sum_{i}\left(\frac{\sum_{t;it\in\Omega}|v_{it}|}{\wnorm}\right)^{2+\kappa}\rightarrow0$,
and so by the Lyapunov central limit theorem we have that $\sigma_{w}^{-1}(\hat{\tau}_{w}-\tau_{w})\stackrel{d}{\rightarrow}\N(0,1).$
\end{proof}
\begin{proof}[Proof of \ref{thm:se}]
We have that:
\begin{align*}
\text{for \ensuremath{D_{it}}} & =0, & \tilde{\varepsilon}_{it}=\hat{\varepsilon}_{it} & =Y_{it}-A_{it}'\hat{\lambda}_{i}^{*}-X_{it}'\hat{\delta}^{*}=\varepsilon_{it}-A_{it}'(\hat{\lambda}_{i}^{*}-\lambda_{i})-X_{it}'(\hat{\delta}^{*}-\delta),\\
\text{for \ensuremath{D_{it}}} & =1, & \tilde{\varepsilon}_{it}=\hat{\tau}_{it} & -\tilde{\tau}_{it}=Y_{it}-A_{it}'\hat{\lambda}_{i}^{*}-X_{it}'\hat{\delta}^{*}-\tilde{\tau}_{it}=\varepsilon_{it}+\tau_{it}-\tilde{\tau}_{it}-A_{it}'(\hat{\lambda}_{i}^{*}-\lambda_{i})-X_{it}'(\hat{\delta}^{*}-\delta),
\end{align*}
so, for each $i$, $\sum_{t;it\in\Omega}v_{it}\tilde{\varepsilon}_{it}=\sum_{t;D_{it}=0}v_{it}\hat{\varepsilon}_{it}+\sum_{t;D_{it}=1}v_{it}(\hat{\tau}_{it}-\tilde{\tau}_{it})=v_{i}'\varepsilon_{i}-v_{i}'A_{i}(\hat{\lambda}_{i}^{*}-\lambda_{i})-v_{i}'X_{i}(\hat{\delta}^{*}-\delta)+\sum_{t;D_{it}=1}v_{it}(\tau_{it}-\tilde{\tau}_{it})$.
Since the estimator $\hat{\tau}_{w}$ is invariant with respect to
a change in $\lambda_{i},$ and $\lambda_{i}$ only appears within
unit $i$ with covariates $A_{i}$, we must have that $0=\frac{\partial}{\partial\lambda_{i}}\hat{\tau}_{w}=\frac{\partial}{\partial\lambda_{i}}v_{i}'Y_{i}=v_{i}'\left(\frac{\partial}{\partial\lambda_{i}}Y_{i}\right)=v_{i}'A_{i}$.
Hence, $\hat{\sigma}_{w}^{2}=\sum_{i=1}^{I}\left(v_{i}'\varepsilon_{i}+\sum_{t,D_{it}=1}v_{it}(\tau_{it}-\tilde{\tau}_{it})-v_{i}'X_{i}(\hat{\delta}^{*}-\delta)\right)^{2}$.
We show convergence of $\hat{\sigma}_{w}^{2}$ by establishing consistency
of $\bar{\sigma}_{w}^{2}=\sum_{i=1}^{I}\left(v_{i}'\varepsilon_{i}+\sum_{t,D_{it}=1}v_{it}\allowbreak(\tau_{it}-\bar{\tau}_{it})\right)^{2}$
with respect to $\sigma_{w}^{2}+\sigma_{\tau}^{2}$ and showing that
$\hat{\sigma}_{w}^{2}$ is close to $\bar{\sigma}_{w}^{2}$. We proceed
in three steps.

First, we consider consistency of $\bar{\sigma}_{w}^{2}$. Note that
for any $i$, $\expec{\left(v_{i}'\varepsilon_{i}+\sum_{t,D_{it}=1}v_{it}(\tau_{it}-\bar{\tau}_{it})\right)^{2}}=v_{i}'\Sigma_{i}v_{i}+\left(\sum_{t,D_{it}=1}v_{it}(\tau_{it}-\bar{\tau}_{it})\right)^{2}$,
and, for $\expec{\varepsilon_{it}^{4}},|\tau_{it}|^{4},|\bar{\tau}_{it}|^{4}\leq C$,
\begin{align*}
 & \var{\left(v_{i}'\varepsilon_{i}+\sum_{t,D_{it}=1}v_{it}(\tau_{it}-\bar{\tau}_{it})\right)^{2}}\leq\E\left[\left(v_{i}'\varepsilon_{i}+\sum_{t,D_{it}=1}v_{it}(\tau_{it}-\bar{\tau}_{it})\right)^{4}\right]\\
 & \leq16\:\E\left[\left(v_{i}'\varepsilon_{i}\right)^{4}+\left(\sum_{t,D_{it}=1}v_{it}(\tau_{it}-\bar{\tau}_{it})\right)^{4}\right]\leq16\:(1+2^{4})\:\left(\sum_{t;it\in\Omega}|v_{it}|\right)^{4}C.
\end{align*}
Since $\sum_{i}\left(\frac{\sum_{t;it\in\Omega}|v_{it}|}{\wnorm}\right)^{4}\rightarrow0$,
we therefore have that
\begin{align*}
\var{\wnorm^{-2}\sum_{i}\left(v_{i}'\varepsilon_{i}+\sum_{t,D_{it}=1}v_{it}(\tau_{it}-\bar{\tau}_{it})\right)^{2}} & \leq16\:(1+2^{4})\:\sum_{i}\left(\frac{\sum_{t;it\in\Omega}|v_{it}|}{\wnorm}\right)^{4}C\rightarrow0
\end{align*}
and thus $\wnorm^{-2}\left(\bar{\sigma}_{w}^{2}-\sigma_{w}^{2}-\sigma_{\tau}^{2}\right)\stackrel{p}{\rightarrow}0$.

Second, we show that $\wnorm^{-2}(\sigma_{w}^{2}+\sigma_{\tau}^{2})$
is bounded from above, which we will later use to bound the difference
between $\hat{\sigma}_{w}^{2}$ and $\bar{\sigma}_{w}^{2}$. To establish
a bound, note that
\begin{align*}
 & \sigma_{w}^{2}{+}\sigma_{\tau}^{2}=\sum_{i}\var{v_{i}'\varepsilon_{i}}{+}\sum_{i}\left(\sum_{t,D_{it}=1}v_{it}(\tau_{it}{-}\bar{\tau}_{it})\right)^{2}\leq\sum_{i}\left(\sum_{t;it\in\Omega}|v_{it}|\right)^{2}\left(\max_{t;it\in\Omega}\var{\varepsilon_{it}}{+}\max_{t;it\in\Omega}(\tau_{it}{-}\bar{\tau}_{it})^{2}\right)\\
 & \leq\left(\max_{it\in\Omega}\sqrt{\expec{\varepsilon_{it}^{4}}}+2\max_{it\in\Omega}\tau_{it}^{2}+2\max_{it\in\Omega}\bar{\tau}_{it}^{2}\right)\sum_{i}\left(\sum_{t;it\in\Omega}|v_{it}|\right)^{2}\leq5\sqrt{C}\wnorm^{2},
\end{align*}
where the first line used that 
\begin{align*}
 & \var{v_{i}'\varepsilon_{i}}=\sum_{t,s;it,is\in\Omega}v_{it}v_{is}\cov{\varepsilon_{it},\varepsilon_{is}}\leq\sum_{t,s;it,is\in\Omega}|v_{it}|\:|v_{is}|\:\sqrt{\var{\varepsilon_{it}}}\sqrt{\var{\varepsilon_{is}}}\\
 & \leq\left(\max_{t;it\in\Omega}\var{\varepsilon_{it}}\right)\cdot\sum_{t,s;it,is\in\Omega}|v_{it}|\:|v_{is}|=\left(\max_{t;it\in\Omega}\var{\varepsilon_{it}}\right)\cdot\left(\sum_{t;it\in\Omega}|v_{it}|\right)^{2}.
\end{align*}
Hence, $\wnorm^{-2}(\sigma_{w}^{2}+\sigma_{\tau}^{2})$ is bounded.
Together with the previous point, it follows that $\wnorm^{-2}\allowbreak\sum_{i}\left(v_{i}'\varepsilon_{i}\allowbreak+\sum_{t,D_{it}=1}v_{it}(\tau_{it}-\bar{\tau}_{it})\right)^{2}=O_{p}(1).$

Third, we bound the difference between $\hat{\sigma}_{w}^{2}$ and
$\bar{\sigma}_{w}^{2}$ by
\begin{align*}
 & \lvert\hat{\sigma}_{w}^{2}-\bar{\sigma}_{w}^{2}\rvert=\left|\sum_{i}\left(v_{i}'\varepsilon_{i}+\sum_{t,D_{it}=1}v_{it}(\tau_{it}-\tilde{\tau}_{it})-v_{i}'X_{i}(\hat{\delta}^{*}-\delta)\right)^{2}-\sum_{i}\left(v_{i}'\varepsilon_{i}+\sum_{t,D_{it}=1}v_{it}(\tau_{it}-\bar{\tau}_{it})\right)^{2}\right|\\
 & \leq\sum_{i}\left(\sum_{t,D_{it}=1}v_{it}(\tilde{\tau}_{it}-\bar{\tau}_{it})+v_{i}'X_{i}(\hat{\delta}^{*}-\delta)\right)^{2}\\
 & \phantom{\le}+2\underbrace{\left|\sum_{i}\left(\sum_{t,D_{it}=1}v_{it}(\tilde{\tau}_{it}-\bar{\tau}_{it})+v_{i}'X_{i}(\hat{\delta}^{*}-\delta)\right)\left(v_{i}'\varepsilon_{i}+\sum_{t,D_{it}=1}v_{it}(\tau_{it}-\bar{\tau}_{it})\right)\right|}_{\leq\sqrt{\sum_{i}\left(\sum_{t,D_{it}=1}v_{it}(\tilde{\tau}_{it}-\bar{\tau}_{it})+v_{i}'X_{i}(\hat{\delta}^{*}-\delta)\right)^{2}}\sqrt{\sum_{i}\left(v_{i}'\varepsilon_{i}+\sum_{t,D_{it}=1}v_{it}(\tau_{it}-\bar{\tau}_{it})\right)^{2}}}.
\end{align*}
To bound the $\sum_{i}\left(\sum_{t,D_{it}=1}v_{it}(\tilde{\tau}_{it}-\bar{\tau}_{it})+v_{i}'X_{i}(\hat{\delta}^{*}-\delta)\right)^{2}$
terms, note that 
\begin{align*}
 & \sum_{i}\left(\sum_{t,D_{it}=1}v_{it}(\tilde{\tau}_{it}{-}\bar{\tau}_{it})+v_{i}'X_{i}(\hat{\delta}^{*}{-}\delta)\right)^{2}\leq2\sum_{i}\left(\sum_{t,D_{it}=1}v_{it}(\tilde{\tau}_{it}{-}\bar{\tau}_{it})\right)^{2}+2\sum_{i}\left(\sum_{t;it\in\Omega}v_{i}'X_{i}(\hat{\delta}^{*}{-}\delta)\right)^{2}.
\end{align*}
Since $\wnorm^{-2}\sum_{i}\left(\sum_{t,D_{it}=1}v_{it}(\tilde{\tau}_{it}-\bar{\tau}_{it})\right)^{2}\stackrel{p}{\rightarrow}0$
and $\wnorm^{-2}\sum_{i}\left(\sum_{t;it\in\Omega}v_{it}X_{it}'(\hat{\delta}^{\ast}-\delta)\right)^{2}\stackrel{p}{\rightarrow}0$
holds, we obtain $\wnorm^{-2}\lvert\hat{\sigma}_{w}^{2}-\bar{\sigma}_{w}^{2}\rvert\stackrel{p}{\rightarrow}0$.
Combining with the first step, we conclude $\wnorm^{-2}(\hat{\sigma}_{w}^{2}-\sigma_{w}^{2}-\sigma_{\tau}^{2})\stackrel{p}{\rightarrow}0$.
\end{proof}
\begin{proof}[Proof of \ref{prop:pretesting}]
To show that $\hat{\tau}_{w}^{*}$ and $\hat{\gamma}$ are uncorrelated,
we invoke the logic of the Hausman test. Under the null hypothesis
$\gamma=0$, $\hat{\tau}_{w}^{*}$ is efficient for $\tau_{\tilde{w}}=\sum_{it\in\Omega_{1}}\tilde{w}_{it}\tau_{it}$
for the weights $\tilde{w}_{it}$ that $\hat{\tau}_{w}^{*}$ places
on treated observations $Y_{it}$, $it\in\Omega_{1}$. (If $\Gamma=\I$
then $\tilde{w}=w$, but otherwise the weights can differ.) Moreover,
$\hat{\gamma}$ is an unbiased estimator of $\gamma=0$, and thus
has to be uncorrelated with $\hat{\tau}_{w}^{\ast}$. Indeed, otherwise
$\var{\hat{\tau}_{w}^{\ast}+\zeta^{\prime}\hat{\gamma}}=\var{\hat{\tau}_{w}^{\ast}}+2\zeta'\cov{\hat{\gamma},\hat{\tau}_{w}^{\ast}}+\zeta'\var{\hat{\gamma}}\zeta<\var{\hat{\tau}_{w}^{\ast}}$
for $\zeta=-\cov{\hat{\gamma},\hat{\tau}_{w}^{\ast}}\varepsilon$
with small $\varepsilon>0$, leading to the more efficient unbiased
estimator $\hat{\tau}_{w}^{\ast}+\zeta^{\prime}\hat{\gamma}$ of $\tau_{\tilde{w}}$.
Since the variances of $Y_{it}$ do not depend on $\gamma$, the covariance
between $\hat{\tau}_{w}^{*}$ and $\hat{\gamma}$ is still zero under
alternatives $\gamma\neq0$. If the error terms are also normal, then
$\hat{\tau}_{w}^{*}$ and $\hat{\gamma}$ (as linear estimators) are
jointly normal and therefore also independent. As a consequence, conditioning
on $\hat{\gamma}$ being outside some test rejection region does not
affect the distribution of $\hat{\tau}_{w}^{*}$.

A similar statement holds asymptotically without the error normality
assumption. Consider for simplicity the case where the null hypothesis
$\gamma=0$ holds. If the sequence of $\Xi=\left(\sigma_{w}^{-1}(\hat{\tau}_{w}^{*}-\tau_{\tilde{w}}),\right.\allowbreak\left.\var{\hat{\gamma}}^{-1/2}\hat{\gamma}\right)$
is asymptotically normal and $\|\Xi\|^{2}$ asymptotically uniformly
integrable, then $\Xi$ converges to a bivariate standard normal distribution,
and conditioning on $\var{\hat{\gamma}}^{-1/2}\hat{\gamma}\not\in R_{\gamma}$
for any non-stochastic rejection region $R_{\gamma}$ (that is a Borel
set of Lebesgue measure below one with a boundary of measure zero)
does not affect the asymptotic distribution of $\sigma_{w}^{-1}(\hat{\tau}_{w}^{*}-\tau_{\tilde{w}})$.\footnote{This claim extends to alternatives $\gamma$ that are local to zero,
i.e. $\var{\hat{\gamma}}^{1/2}\gamma\rightarrow a_{\gamma}$ for some
$a_{\gamma}$, such that the probability of rejection is below one
in the limit. In this case, we may have $\tau_{\tilde{w}}\neq\expec{\hat{\tau}_{w}^{*}}$,
and the claim applies to $\Xi=\left(\sigma_{w}^{-1}(\hat{\tau}_{w}^{*}-\expec{\hat{\tau}_{w}^{*}}),\var{\hat{\gamma}}^{-1/2}(\hat{\gamma}-\gamma)\right)$.}

We now prove this claim. If $\|\Xi\|^{2}$ is asymptotically uniformly
integrable along the asymptotic sequence, then the components of $\Xi$
and $\Xi\Xi'$ are also asymptotically uniformly integrable. Since
$\Xi$ is also asymptotically normal, $\Xi\stackrel{d}{\longrightarrow}\N\left(0,V_{\Xi}\right)$
with $V_{\Xi}=\lim\var{\Xi}=\begin{psmallmatrix}1 & 0' \\ 0 & \I \end{psmallmatrix}$
by Theorem 2.20 in \textcite{van2000asymptotic}. In other words,
$\Xi\stackrel{d}{\longrightarrow}(N_{1},N_{2})$ with independent
$N_{1}\sim\N(0,1),N_{2}\sim\N(0,\I)$. Then, by the Portmanteau lemma
(Lemma 2.2 in \textcite{van2000asymptotic}) and independence in the
limit, for $B$ denoting another Borel set in $\mathbb{R}$ with a
boundary of Lebesgue measure zero,
\begin{align*}
\mathbb{P}\left(\sigma_{w}^{-1}(\hat{\tau}_{w}^{*}-\tau_{\tilde{w}})\in B\middle|\var{\hat{\gamma}}^{-1/2}\hat{\gamma}\notin R_{\gamma}\right) & \rightarrow\text{\ensuremath{\frac{\mathbb{P}\left(N_{1}\in B,N_{2}\notin R_{\gamma}\right)}{\mathbb{P}\left(N_{2}\notin R_{\gamma}\right)}}}=\mathbb{P}(N_{1}\in B),
\end{align*}
so the asymptotic distribution of $\sigma_{w}^{-1}(\hat{\tau}_{w}^{*}-\tau_{\tilde{w}})$
is not affected by conditioning on $\var{\hat{\gamma}}^{-1/2}\hat{\gamma}\notin R_{\gamma}$.
\end{proof}

\subsection{Proofs of Appendix Results}
\begin{proof}[Proof of \ref{prop:identification}]
In the first case, writing $\tau^{*}=(\tau_{it})_{it\in\Omega_{1},w_{it}\ne0}$,
and adopting notation for vectors of $Y_{it}$ and $\varepsilon_{it}$
similar to matrices for $Z_{it}$, we have that $Y_{1}^{*}=Z_{1}^{*}\pi+\tau^{*}+\varepsilon_{1}^{*},Y_{0}=Z_{0}\pi+\varepsilon_{0}.$
There exists a matrix $M$ such that $Z_{1}^{*}=MZ_{0}$, since otherwise
$\text{rank}\left(\begin{psmallmatrix}Z_{1}^{*} \\ Z_{0}\end{psmallmatrix}\right) > \text{rank}(Z_{0})$.
Then we have that $\tau^{*}=\expec{Y_{1}^{*}-Z_{1}^{*}\pi}=\expec{Y_{1}^{*}-MZ_{0}\pi}=\expec{Y_{1}^{*}-MY_{0}}$,
so $\tau^{*}$ and thus $\tau_{w}$ are identified.

In the second case, assume without loss that unit $i=1$ is observed
for all time periods up to $T(w)$, and for all $it\in\Omega$ with
$i>1$ and $t\leq T(w)$ let $\tilde{Y}_{it}=Y_{it}-Y_{1t}$, which
fulfills $\expec{\tilde{Y}_{it}}=A_{t}^{\prime}(\lambda_{i}-\lambda_{1})+D_{it}\tau_{it}$.
For every $i>1$ with $\sum_{t;it\in\Omega_{1}}\lvert w_{it}\rvert>0$,
write $\tilde{Y}_{i}=(\tilde{Y}_{it})_{it\in\Omega_{0},t\leq T(w)}$,
$\tilde{A}_{i}=(A_{it}')_{it\in\Omega_{0},t\leq T(w)}$, and $\tilde{\lambda}_{i}=\lambda_{i}-\lambda_{1}$.
We have that $\expec{\tilde{Y}_{i}}=\tilde{A}_{i}\tilde{\lambda}_{i}$,
so $\tilde{\lambda}_{i}=(\tilde{A}_{i}'\tilde{A}_{i})^{-1}\tilde{A}_{i}'\expec{\tilde{Y}_{i}}$
(where the inverse exists because of full column rank of $\tilde{A}_{i}$).
For every $it\in\Omega_{1}$ with $\lvert w_{it}\rvert>0$, we must
have that $i>1$ and $\sum_{t;it\in\Omega_{1}}\lvert w_{it}\rvert>0$,
so we can identify $\tau_{it}$ by $\tau_{it}=\expec{\tilde{Y}_{it}}-A_{t}'\tilde{\lambda}_{i}=\expec{\tilde{Y}_{it}}-A_{t}'(\tilde{A}_{i}'\tilde{A}_{i})^{-1}\tilde{A}_{i}'\expec{\tilde{Y}_{i}}$.
Hence, $\tau_{w}$ is identified.
\end{proof}
\begin{proof}[Proof of \ref{prop:model-imputation}]
Since $\hat{\tau}_{w}$ is linear by assumption we can write $\hat{\tau}_{w}=v_{1}'Y_{1}+v_{0}'Y_{0}$.
Unbiasedness implies that $\expec{\hat{\tau}_{w}}=v_{1}'\Gamma\theta+(v_{1}'Z_{1}+v_{0}'Z_{0})\pi=w_{1}'\Gamma\theta$
for any $\theta$ and $\pi$. Hence, $\Gamma'v_{1}=\Gamma'w_{1}$
and $v_{1}'Z_{1}+v_{0}'Z_{0}=\0'$. It follows that $\hat{\tau}_{w}$
is an unbiased estimator of $\tau_{v}=v_{1}'\tau$ for all $\tau$.
The remaining part of the proposition is a direct consequence of \ref{prop:generalimputation}.
\end{proof}
\begin{proof}[Proof of \ref{prop:olsweights}]
Write $\Phi_{Z}=\I-Z(Z'Z)^{-1}Z'$ for the annihilator matrix with
respect to the control variables, then
\begin{align*}
\Phi & =\Phi_{Z}|_{\Omega_{1}\times\Omega_{1}}=\I-Z_{1}(Z'Z)^{-1}Z_{1}^{\prime}=\I-Z_{1}(Z_{1}'Z_{1}+Z_{0}'Z_{0})^{-1}Z_{1}^{\prime},\\
\Phi_{0} & =\Phi_{Z}|_{\Omega_{0}\times\Omega_{1}}=-Z_{0}(Z'Z)^{-1}Z_{1}^{\prime}=-Z_{0}(Z_{1}'Z_{1}+Z_{0}'Z_{0})^{-1}Z_{1}^{\prime}.
\end{align*}
By the Frisch\textendash Waugh\textendash Lovell theorem for $\hat{\theta}$
in \ref{thm:OLS-BLUE}, $\hat{\theta}=\left((D\Gamma)'\Phi_{Z}(D\Gamma)\right)^{-1}(D\Gamma)'\Phi_{Z}Y$,
and, since $v^{\ast\prime}Y=\hat{\tau}_{w}^{\ast}=w_{1}^{\prime}\Gamma\hat{\theta}$,
using $D=\begin{psmallmatrix}\I\\\mathbb{O}\end{psmallmatrix}$, and
plugging in the expressions for $\Phi,\Phi_{0}$, \begin{align*}v^{*}&=\Phi_{Z}D\Gamma\left((D\Gamma)'\Phi_{Z}(D\Gamma)\right)^{-1}\Gamma'w_{1}=\begin{psmallmatrix}\Phi\\\Phi_{0}\end{psmallmatrix}\Gamma(\Gamma'\Phi\Gamma)^{-1}\Gamma'w_{1}\\&=\left(\begin{psmallmatrix}\I\\\Null\end{psmallmatrix}-Z(Z'Z)^{-1}Z_{1}'\right)\Gamma(\Gamma'(\I-Z_{1}(Z'Z)^{-1}Z_{1}^{\prime})\Gamma)^{-1}\Gamma'w_{1},\end{align*}as
required by the proposition.

In the special case of $\Gamma=\I_{N_{1}}$, $\Gamma(\Gamma'\Phi\Gamma)^{-1}\Gamma'=\Phi^{-1}$
(where we note that $\Phi$ is invertible by the assumption that $\tau$
is identified, which requires that $\Phi_{Z}D$ is not collinear or,
equivalently, that $D'\Phi_{Z}D=\Phi$ is not singular), simplifying
the expression to $v^{*}=\begin{psmallmatrix}\Phi\\\Phi_{0}\end{psmallmatrix}\Gamma(\Gamma'\Phi\Gamma)^{-1}\Gamma'w_{1}=\begin{psmallmatrix}w_{1}\\\Phi_{0}\Phi^{-1}w_{1}\end{psmallmatrix}$.
To simplify $\Phi_{0}\Phi^{-1}$, we note that $(Z'Z)^{-1}Z_{1}^{\prime}=(Z_{0}'Z_{0})^{-1}Z_{0}'Z_{0}(Z'Z)^{-1}Z_{1}^{\prime}=(Z_{0}'Z_{0})^{-1}(Z'Z(Z'Z)^{-1}-Z_{1}'Z_{1}(Z'Z)^{-1})Z_{1}^{\prime}=(Z_{0}'Z_{0})^{-1}Z_{1}'(\I-Z_{1}(Z'Z)^{-1}Z_{1}^{\prime})$,
and thus, plugging in for $(Z'Z)^{-1}Z_{1}^{\prime}$,
\begin{align*}
 & \Phi_{0}\Phi^{-1}=-Z_{0}(Z'Z)^{-1}Z_{1}^{\prime}(\I-Z_{1}(Z'Z)^{-1}Z_{1}^{\prime})^{-1}\\
 & =-Z_{0}(Z_{0}'Z_{0})^{-1}Z_{1}'(\I-Z_{1}(Z'Z)^{-1}Z_{1}^{\prime})(\I-Z_{1}(Z'Z)^{-1}Z_{1}^{\prime})^{-1}=-Z_{0}(Z_{0}'Z_{0})^{-1}Z_{1}',
\end{align*}
which shows the expression for $\Gamma=\I_{N_{1}}$ and thus concludes
the proof.
\end{proof}
\begin{proof}[Proof of \ref{prop:imputation-model}]
We show that the weights from \ref{prop:olsweights} are the same
as those stated in the proposition. We can solve the optimization
problem for $v_{1}^{*}$ from the Lagrangian relaxation $\min_{v}v'\Phi^{-1}v-2\lambda'\Gamma'(v-w_{1})$
with the first order condition $\Gamma\lambda=\Phi{}^{-1}v$, which
is solved for $v_{1}^{*}=\Phi\Gamma(\Gamma'\Phi\Gamma)^{-1}\Gamma'w_{1}$
(as claimed) with $\lambda=(\Gamma'\Phi\Gamma)^{-1}\Gamma'w_{1}$.
Here, $\Phi^{-1}$ can be written as the variance $\Var^{\#}\left[\hat{\tau}^{*}\right]$
of the OLS estimator $\hat{\tau}^{*}$ of $\tau$ with unrestricted
heterogeneity and spherical errors with unit variance, which is 
\begin{align*}
 & \Var^{\#}\left[\hat{\tau}^{*}\right]=\Var^{\#}\left[(D'(\I-Z(Z'Z)^{-1}Z')D)^{-1}D'(\I-Z(Z'Z)^{-1}Z')Y\right]\\
 & =(D'(\I\!-\!Z(Z'Z)^{-1}Z')D)^{-1}D'(\I\!-\!Z(Z'Z)^{-1}Z')\Var^{\#}\left[Y\right](\I\!-\!Z(Z'Z)^{-1}Z')D(D'(\I\!-\!Z(Z'Z)^{-1}Z')D)^{-1}\\
 & =(D'(\I-Z(Z'Z)^{-1}Z')D)^{-1}=\Phi^{-1}.
\end{align*}
The weight on $Y_{0}$ is then implied by imputation. Indeed, the
imputation estimator from \ref{thm:imputation} with weights $v_{1}^{*}$
is $v_{1}^{*\prime}(Y_{1}+\Phi^{-1}\Phi_{0}Y_{0})$ by \ref{prop:olsweights}
with $\Gamma=\I_{N_{1}}$. Given $v_{1}^{*}=\Phi\Gamma(\Gamma'\Phi\Gamma)^{-1}\Gamma'w_{1}$,
this yields the same weights $\Phi_{0}\Gamma(\Gamma'\Phi\Gamma)^{-1}\Gamma'w_{1}$
on $Y_{0}$ as $v_{0}^{*}$ in the proof of \ref{prop:olsweights}.
\end{proof}

\begin{proof}[Proof of \ref{prop:consistency-sufficient}]
 To prove consistency, we show that the Herfindahl condition from
\ref{assu:Herfindahl} is fulfilled for $\hat{\tau}_{w}^{*}$, that
is, $\|v^{*}\|_{H}^{2}\rightarrow0$, which allows us to invoke \ref{prop:consistency}.
As a preliminary result, note that for any $b\geq1$, $S\in\mathbb{Z}_{+}$,
and any $a_{1},\dots,a_{S}$, Jensen's inequality implies
\begin{equation}
\left|\sum_{s=1}^{S}a_{s}\right|^{b}\leq S^{b-1}\sum_{s=1}^{S}|a_{s}|^{b}.\label{eq:Jensen}
\end{equation}
Assume without loss of generality that $\beta_{1}=0$. To bound the
weights $v_{it}^{*}$, we consider the alternative unbiased linear
estimator $\hat{\alpha}_{i}^{\#}=Y_{i1},\hat{\beta}_{t}^{\#}=\frac{\sum_{i;it\in\Omega_{0}}(Y_{it}-Y_{i1})}{\sum_{i;it\in\Omega_{0}}1}$
of the unit and period FEs. Since $\sum_{it\in\Omega_{0}}v_{it}^{*}Y_{it}$
is the best linear unbiased estimator for $(-\sum_{it\in\Omega_{1}}w_{it}(\alpha_{i}+\beta_{t}))$
under spherical errors, we can bound the sum of squares weights $\sum_{it\in\Omega_{0}}(v_{it}^{*})^{2}$
since they are the same as the variance $\Var^{\#}\left[\sum_{it\in\Omega_{0}}v_{it}^{*}Y_{it}\right]$
of $\sum_{it\in\Omega_{0}}v_{it}^{*}Y_{it}$ for spherical errors
$\varepsilon_{it}$ with unit variance. Specifically,
\begin{align*}
 & \sum_{it\in\Omega_{0}}(v_{it}^{*})^{2}=\Var^{\#}\left[\sum_{it\in\Omega_{0}}v_{it}^{*}Y_{it}\right]\leq\Var^{\#}\left[-\sum_{it\in\Omega_{1}}w_{it}(\hat{\alpha}_{i}^{\#}+\hat{\beta}_{t}^{\#})\right]\\
 & =\Var^{\#}\left[\sum_{i=1}^{I}\left(\sum_{t;it\in\Omega_{1}}w_{it}\right)\hat{\alpha}_{i}^{\#}+\sum_{t=2}^{T}\left(\sum_{i;it\in\Omega_{1}}w_{it}\right)\hat{\beta}_{t}^{\#}\right]\\
 & \leq2\Var^{\#}\left[\sum_{i=1}^{I}\left(\sum_{t;it\in\Omega_{1}}w_{it}\right)Y_{i1}\right]+2\Var^{\#}\left[\sum_{t=2}^{T}\left(\sum_{i;it\in\Omega_{1}}w_{it}\right)\left(\frac{\sum_{i;it\in\Omega_{0}}(Y_{it}-Y_{i1})}{\sum_{i;it\in\Omega_{0}}1}\right)\right]\\
 & \leq2\sum_{i=1}^{I}\Var^{\#}\left[\left(\sum_{t;it\in\Omega_{1}}w_{it}\right)Y_{i1}\right]+2T\sum_{t=2}^{T}\Var^{\#}\left[\left(\sum_{i;it\in\Omega_{1}}w_{it}\right)\left(\frac{\sum_{i;it\in\Omega_{0}}(Y_{it}-Y_{i1})}{\sum_{i;it\in\Omega_{0}}1}\right)\right]\\
 & \leq2\sum_{i=1}^{I}\left(\sum_{t;it\in\Omega_{1}}w_{it}\right)^{2}+8T\sum_{t=2}^{T}\frac{\left(\sum_{i;it\in\Omega_{1}}w_{it}\right)^{2}}{\sum_{i;it\in\Omega_{0}}1},
\end{align*}
where we repeatedly use \ref{eq:Jensen} for $b=2$. Hence, we have
that
\begin{align*}
 & \sum_{i}\left(\sum_{t;it\in\Omega}|v_{it}^{*}|\right)^{2}\!\leq\!2\sum_{i=1}^{I}\left(\sum_{t;it\in\Omega_{1}}|v_{it}^{*}|\right)^{2}\!+\!2\sum_{i=1}^{I}\left(\sum_{t;it\in\Omega_{0}}|v_{it}^{*}|\right)^{2}\!\leq\!2\sum_{i=1}^{I}\left(\sum_{t;it\in\Omega_{1}}|v_{it}^{*}|\right)^{2}\!+\!2T\sum_{it\in\Omega_{0}}(v_{it}^{*})^{2}\\
 & \leq2\sum_{i=1}^{I}\left(\sum_{t;it\in\Omega_{1}}|w_{it}|\right)^{2}+4T\sum_{i=1}^{I}\left(\sum_{t;it\in\Omega_{1}}w_{it}\right)^{2}+16T^{2}\sum_{t=2}^{T}\frac{\left(\sum_{i;it\in\Omega_{1}}w_{it}\right)^{2}}{\sum_{i;it\in\Omega_{0}}1}\rightarrow0
\end{align*}
under \ref{assu:lowlevel-consistent}, which implies \ref{assu:Herfindahl}
and thus consistency by \ref{prop:consistency}.
\end{proof}
\begin{proof}[Proof of \ref{prop:normality-sufficient}]
To establish asymptotic normality, we want to establish that $\sum_{i}\left(\frac{\sum_{t;it\in\Omega}|v_{it}|}{\wnorm}\right)^{2+\kappa}\rightarrow0$
for $\kappa=2$, which allows us to invoke \ref{prop:asymptoticnormality}.
By construction and since $\hat{\alpha}_{i}$ is the least-squares
solution to $\min\sum_{it;D_{it}=0}(Y_{it}-\hat{\alpha}_{i}-\hat{\beta}_{t})^{2}$
, we have $\hat{\tau}_{w}^{*}=\sum_{it;D_{it}=1}w_{it}(Y_{it}-\hat{\alpha}_{i}-\hat{\beta}_{t})$,
$\hat{\alpha}_{i}=\frac{\sum_{t;D_{it}=0}(Y_{it}-\hat{\beta}_{t})}{\sum_{t;D_{it}=0}1}$,
so we can express

\begin{align*}
 & \hat{\tau}_{w}^{\ast}=\sum_{it;D_{it}=1}w_{it}\left(Y_{it}-\frac{\sum_{s;D_{is}=0}(Y_{is}-\hat{\beta}_{s})}{\sum_{s;D_{is}=0}1}-\hat{\beta}_{t}\right),\\
 & =\sum_{it;D_{it}=1}w_{it}Y_{it}+\sum_{it;D_{it}=0}\frac{-\sum_{s;D_{is}=1}w_{is}}{\sum_{s;D_{is}=0}1}Y_{it}+\sum_{t=2}^{T}\left(\sum_{i;D_{it}=0}\frac{\sum_{s;D_{is}=1}w_{it}}{\sum_{s;D_{is}=0}1}-\sum_{i;D_{it}=1}w_{it}\right)\hat{\beta}_{t}\\
 & =\sum_{\mathclap{it;D_{it}=1}}w_{it}Y_{it}+\sum_{\mathclap{it;D_{it}=0}}\frac{-\sum_{s;D_{is}=1}w_{is}}{\sum_{s;D_{is}=0}1}Y_{it}+\sum_{\mathclap{it;D_{it}=0}}u_{it}Y_{it}=\sum_{i=1}^{I}\left(\sum_{t=1}^{E_{i}-1}\left(u_{it}-\frac{\sum_{s;D_{is}=1}w_{is}}{\sum_{s;D_{is}=0}1}\right)Y_{it}+\sum_{t=E_{i}}^{T}w_{it}Y_{it}\right)
\end{align*}
where the $u_{it}$ are the weights on the $Y_{it},it\in\Omega_{0}$
in $\sum_{t=2}^{T}\left(\sum_{i;D_{it}=0}\frac{\sum_{s;D_{is}=1}w_{it}}{\sum_{s;D_{is}=0}1}-\sum_{i;D_{it}=1}w_{it}\right)\hat{\beta}_{t}$,
and we write $E_{i}=T+1$ for the never-treated cohort. Hence, we
can write
\begin{equation}
v_{it}^{*}=\begin{cases}
w_{it}, & it\in\Omega_{1},\\
u_{it}-\frac{\sum_{s;D_{is}=1}w_{is}}{\sum_{s;D_{is}=0}1}, & it\in\Omega_{0}.
\end{cases}\label{eq:vstarandu}
\end{equation}

It follows with \ref{eq:Jensen} in the previous proof for $b=2+\kappa=4$
that
\begin{align*}
 & \sum_{i=1}^{I}\left(\sum_{t=1}^{T}|v_{it}^{*}|\right)^{2+\kappa}\leq T^{1+\kappa}\sum_{it\in\Omega}|v_{it}^{*}|^{2+\kappa}=T^{1+\kappa}\left(\sum_{it\in\Omega_{1}}|w_{it}|^{2+\kappa}+\sum_{it\in\Omega_{0}}\lvert u_{it}-\frac{\sum_{s;D_{is}=1}w_{is}}{\sum_{s;D_{is}=0}1}\rvert^{2+\kappa}\right)\\
 & \leq T{}^{1+\kappa}\Bigg(\sum_{it\in\Omega_{1}}|w_{it}|^{2+\kappa}+2^{1+\kappa}\sum_{it\in\Omega_{0}}|u_{it}|^{2+\kappa}+2^{1+\kappa}\underbrace{\sum_{it\in\Omega_{0}}\left|\frac{\sum_{s;D_{is}=1}w_{is}}{\sum_{s;D_{is}=0}1}\right|^{2+\kappa}}_{=\sum_{i=1}^{I}\frac{\lvert\sum_{t;D_{it}=1}w_{is}\rvert^{2+\kappa}}{\lvert\sum_{t;D_{it}=0}1\rvert^{1+\kappa}}}\Bigg)\\
 & \leq(2T){}^{1+\kappa}\left(\sum_{it\in\Omega_{1}}|w_{it}|^{2+\kappa}+\sum_{it\in\Omega_{0}}|u_{it}|^{2+\kappa}+T^{1+\kappa}\sum_{it\in\Omega_{1}}|w_{it}|^{2+\kappa}\right)\\
 & \leq(2T)^{2(1+\kappa)}\left(\sum_{it\in\Omega_{1}}|w_{it}|^{2+\kappa}+\sum_{it\in\Omega_{0}}|u_{it}|^{2+\kappa}\right).
\end{align*}
We now consider the two parts of this sum. First, note that there
is only a finite number of cohort\textendash period cells $et$ with
$E_{i}=e\geq t$, and within every such cell
\begin{align*}
\frac{\sum_{i;E_{i}=e}|w_{it}|^{2+\kappa}}{(\sum_{i;E_{i}=e}w_{it}^{2})^{\frac{2+\kappa}{2}}} & \leq\frac{\max_{i;E_{i}=e}|w_{it}|^{2+\kappa}\sum_{i;E_{i=e}}1}{\min_{i;E_{i}=e}|w_{it}|^{2+\kappa}(\sum_{i;E_{i=e}}1)^{\frac{2+\kappa}{2}}}\leq C^{2+\kappa}\frac{1}{(\sum_{i;E_{i=e}}1)^{\frac{\kappa}{2}}}\rightarrow0,
\end{align*}
For the second part of the sum, $u_{it}$ only depend on cohort and
period (since period FEs are invariant to exchanging unit identities
within cohorts), $u_{it}=U_{E_{i},t}$, so similarly, for $E_{i}=e<t$,
\[
\frac{\sum_{i;E_{i}=e}|u_{it}|^{2+\kappa}}{(\sum_{i;E_{i}=e}|u_{it}|^{2})^{\frac{2+\kappa}{2}}}=\frac{\sum_{i;E_{i}=e}|U_{et}|^{2+\kappa}}{(\sum_{i;E_{i}=e}|U_{et}|^{2})^{\frac{2+\kappa}{2}}}=\frac{1}{(\sum_{i;E_{i}=e}1)^{\frac{\kappa}{2}}}\rightarrow0.
\]
From the fact that $\frac{\sum_{i;E_{i}=e}|w_{it}|^{2+\kappa}}{(\sum_{i;E_{i}=e}w_{it}^{2})^{\frac{2+\kappa}{2}}}$
and $\frac{\sum_{i;E_{i}=e}|u_{it}|^{2+\kappa}}{(\sum_{i;E_{i}=e}u_{it}^{2})^{\frac{2+\kappa}{2}}}$
vanish for all $e$ and $t$ we now derive that $\frac{\sum_{it\in\Omega}|v_{it}^{*}|^{2+\kappa}}{(\sum_{it\in\Omega}v_{it}^{*}{}^{2})^{\frac{2+\kappa}{2}}}$
also vanishes. To this end, we note that 
\begin{align*}
\sum_{it\in\Omega_{0}}\left|\frac{\sum_{s;D_{is}=1}w_{is}}{\sum_{s;D_{is}=0}1}\right|^{2} & =\sum_{i=1}^{I}\left(\sum_{t;D_{it}=0}1\right)\left|\frac{\sum_{s;D_{is}=1}w_{is}}{\sum_{s;D_{is}=0}1}\right|^{2}\leq\sum_{i=1}^{I}\left|\sum_{t;it\in\Omega_{1}}w_{it}\right|^{2}\leq T\sum_{it\in\Omega_{1}}w_{it}^{2}.
\end{align*}
Using \ref{eq:vstarandu}, we obtain that $\frac{1}{2}\sum_{it\in\Omega_{0}}u_{it}^{2}\leq\sum_{it\in\Omega_{0}}\left|\frac{\sum_{s;D_{is}=1}w_{is}}{\sum_{s;D_{is}=0}1}\right|^{2}+\sum_{it\in\Omega_{0}}v_{it}^{*}{}^{2}\leq T\sum_{it\in\Omega_{1}}w_{it}^{2}+\sum_{it\in\Omega_{0}}v_{it}^{*}{}^{2}\leq T\sum_{it\in\Omega}v_{it}^{*}{}^{2}.$
At the same time, $\sum_{it\in\Omega}v_{it}^{*}{}^{2}\geq\sum_{it\in\Omega_{1}}w_{it}^{2}$.
Putting everything together,
\begin{align*}
 & \frac{\sum_{it\in\Omega}|v_{it}^{*}|^{2+\kappa}}{(\sum_{it\in\Omega}v_{it}^{*}{}^{2})^{\frac{2+\kappa}{2}}}\leq(2T)^{2(1+\kappa)}\frac{\sum_{it\in\Omega_{1}}|w_{it}|^{2+\kappa}+\sum_{it\in\Omega_{0}}|u_{it}|^{2+\kappa}}{(\sum_{it\in\Omega}v_{it}^{*}{}^{2})^{\frac{2+\kappa}{2}}}\\
 & \leq(2T)^{2(1+\kappa)}\frac{\sum_{it\in\Omega_{1}}|w_{it}|^{2+\kappa}}{\left(\sum_{it\in\Omega_{1}}w_{it}^{2}\right)^{\frac{2+\kappa}{2}}}+(2T)^{\frac{2+\kappa}{2}}(2T)^{2(1+\kappa)}\frac{\sum_{it\in\Omega_{0}}|u_{it}|^{2+\kappa}}{\left(\sum_{it\in\Omega_{0}}u_{it}^{2}\right)^{\frac{2+\kappa}{2}}}\\
 & \leq(2T)^{2(1+\kappa)}\sum_{et;e\leq t}\frac{\sum_{i;E_{i}=e}|w_{it}|^{2+\kappa}}{(\sum_{i;E_{i}=e}w_{it}^{2})^{\frac{2+\kappa}{2}}}+(2T)^{2(1+\kappa)+\frac{2+\kappa}{2}}\sum_{et;e>t}\frac{\sum_{i;E_{i}=e}|u_{it}|^{2+\kappa}}{(\sum_{i;E_{i}=e}u_{it}^{2})^{\frac{2+\kappa}{2}}}\rightarrow0.
\end{align*}
Since also $\wnorm^{2}=\sum_{i}\left(\sum_{t;it\in\Omega}|v_{it}^{*}|\right)^{2}\geq\sum_{it\in\Omega}v_{it}^{*}{}^{2}$,
we conclude that $\sum_{i}\left(\frac{\sum_{t;it\in\Omega}|v_{it}^{\ast}|}{\wnorm}\right)^{2+\kappa}\leq\frac{\sum_{it\in\Omega}|v_{it}^{*}|^{2+\kappa}}{(\sum_{it\in\Omega}v_{it}^{*}{}^{2})^{\frac{2+\kappa}{2}}}\rightarrow0,$
allowing us to invoke \ref{prop:asymptoticnormality} to obtain asymptotic
normality.
\end{proof}
\begin{proof}[Proof of \ref{prop:se-short-sufficient}]
We show that the assumptions of this proposition imply the assumptions
of \ref{thm:se}, which guarantee that standard errors are asymptotically
conservative. The two preceding proofs establish that \ref{assu:Herfindahl}
holds and $\sum_{i}\left(\frac{\sum_{t;it\in\Omega}|v_{it}|}{\wnorm}\right)^{4}\rightarrow0$.
It remains to show that 
\begin{align}
\wnorm^{-2}\sum_{i}\left(\sum_{t;it\in\Omega_{1}}v_{it}^{*}(\tilde{\tau}_{it}-\bar{\tau}_{it})\right)^{2} & \stackrel{p}{\rightarrow}0, & \wnorm^{-2}\sum_{i}\left(\sum_{t;it\in\Omega}v_{it}^{*}(\hat{\beta}_{t}-\beta_{t})\right)^{2} & \stackrel{p}{\rightarrow}0.\label{eq:nuisanceconsistency}
\end{align}
We assume throughout that the first time-fixed effect is dropped,
$\beta_{1}=0$, which is without loss since it does not affect the
estimator or the standard errors. We first consider the left expression
for $\tilde{\tau}_{it}=\tilde{\tau}_{it}^{A}\equiv\tilde{\tau}=\frac{\sum_{i}(\sum_{t;it\in\Omega_{1}}v_{it}^{*})(\sum_{t;it\in\Omega_{1}}v_{it}^{*}\hat{\tau}_{it})}{\sum_{i}(\sum_{t;it\in\Omega_{1}}v_{it}^{*})^{2}}$,
with the corresponding $\bar{\tau}_{it}\equiv\bar{\tau}$. In this
case,
\begin{align*}
\sum_{i}\left(\sum_{t;it\in\Omega_{1}}v_{it}^{*}(\tilde{\tau}_{it}{-}\bar{\tau}_{it})\right)^{2} & \!=\!(\tilde{\tau}{-}\bar{\tau})^{2}\sum_{i}\left(\sum_{t;it\in\Omega_{1}}v_{it}^{*}\right)^{2}\!=\!\frac{\left(\sum_{i}\left(\sum_{t;it\in\Omega_{1}}v_{it}^{*}\right)\left(\sum_{t;it\in\Omega_{1}}v_{it}^{*}(\hat{\tau}_{it}{-}\tau_{it})\right)\right)^{2}}{\sum_{i}(\sum_{t;it\in\Omega_{1}}v_{it}^{*})^{2}}.
\end{align*}
When $\tilde{\tau}_{it}$ are cohort\textendash horizon cell averages,
specifically $\tilde{\tau}_{it}=\tilde{\tau}_{it}^{B}=\tilde{\tau}_{E_{i}t}$
for $\tilde{\tau}_{et}=\frac{\sum_{i;E_{i}=e}v_{it}^{*}{}^{2}\hat{\tau}_{it}}{\sum_{i;E_{i}=e}v_{it}^{*}{}^{2}}$,
and $\bar{\tau}_{it}=\bar{\tau}_{E_{i}t}$ are defined correspondingly,
\begin{align*}
 & \sum_{i}\left(\sum_{t;it\in\Omega_{1}}v_{it}^{*}(\tilde{\tau}_{it}-\bar{\tau}_{it})\right)^{2}=\sum_{e}\sum_{i;E_{i}=e}\left(\sum_{t\geq e}v_{it}^{*}(\tilde{\tau}_{et}-\bar{\tau}_{et})\right)^{2}\leq T\sum_{e}\sum_{i;E_{i}=e}\sum_{t\geq e}\left(v_{it}^{*}(\tilde{\tau}_{et}-\bar{\tau}_{et})\right)^{2}\\
 & =T\sum_{e}\sum_{t\geq e}(\tilde{\tau}_{et}-\bar{\tau}_{et})^{2}\sum_{i;E_{i}=e}v_{it}^{*}{}^{2}=T\sum_{e}\sum_{t\geq e}\frac{\left(\sum_{i;E_{i}=e}v_{it}^{*}{}^{2}(\hat{\tau}_{it}-\tau_{it})\right)^{2}}{\sum_{i;E_{i}=e}v_{it}^{*}{}^{2}}.
\end{align*}
We finally consider the expression in \ref{eq:nuisanceconsistency}
involving the time fixed effects $\hat{\beta}_{t}$. Here,
\begin{align*}
 & \sum_{i}\left(\sum_{t;it\in\Omega}v_{it}^{*}(\hat{\beta}_{t}-\beta_{t})\right)^{2}\leq T\sum_{it\in\Omega}(\hat{\beta}_{t}-\beta_{t})^{2}v_{it}^{*}{}^{2}=T\sum_{t}(\hat{\beta}_{t}-\beta_{t})^{2}\sum_{i;it\in\Omega}v_{it}^{*}{}^{2}.
\end{align*}
From these three expressions, we conclude that
\begin{align*}
\expec{\sum_{i}\left(\sum_{t;it\in\Omega_{1}}v_{it}^{*}(\tilde{\tau}_{it}^{A}-\bar{\tau}_{it})\right)^{2}} & \leq\frac{\var{\sum_{i}\left(\sum_{t;it\in\Omega_{1}}v_{it}^{*}\right)\left(\sum_{t;it\in\Omega_{1}}v_{it}^{*}\hat{\tau}_{it}\right)}}{\sum_{i}(\sum_{t;it\in\Omega_{1}}v_{it}^{*})^{2}},\\
\expec{\sum_{i}\left(\sum_{t;it\in\Omega_{1}}v_{it}^{*}(\tilde{\tau}_{it}^{B}-\bar{\tau}_{it})\right)^{2}} & \leq T\sum_{e}\sum_{t\geq e}\frac{\var{\sum_{i;E_{i}=e}v_{it}^{*}{}^{2}\hat{\tau}_{it}}}{\sum_{i;E_{i}=e}v_{it}^{*}{}^{2}},\\
\expec{\sum_{i}\left(\sum_{t;it\in\Omega}v_{it}^{*}(\hat{\beta}_{t}-\beta_{t})\right)^{2}} & \leq T\sum_{t}\var{\hat{\beta}_{t}}\sum_{i;it\in\Omega}v_{it}^{*}{}^{2}.
\end{align*}
We now bound the three variances above. Each of them contains a linear
estimator. Such a linear estimator $\hat{\eta}=\sum_{it\in\Omega}a_{it}Y_{it}$
has variance $\var{\hat{\eta}}=\sum_{i}\E\left(\sum_{t;it\in\Omega}a_{it}\varepsilon_{it}\right)^{2}\leq T\sum_{it\in\Omega}a_{it}^{2}\;\E\varepsilon_{it}^{2}\leq T\bar{\sigma}^{2}\sum_{it\in\Omega}a_{it}^{2}$.
We can therefore bound it by $T\bar{\sigma}^{2}$ times the variance
$\Var^{\#}(\hat{\eta})=\sum_{it\in\Omega}a_{it}^{2}$ in data with
uncorrelated and spherical errors terms with unit variance. In such
data, the estimators $\hat{\tau}_{it}$ and $\hat{\beta}_{t}$ have
minimal variance (by Gauss\textendash Markov), which extends to their
weighted averages as in the proof of \ref{thm:OLS-BLUE}. Under $\Var^{\#}$,
we can therefore bound the respective variances by the variances using
any other linear unbiased estimators of the $\tau_{it}$ and $\beta_{t}$.
Recalling that we normalized $\beta_{1}=0$, we can use the estimators
$\beta_{t}^{\#}=\frac{\sum_{i;E_{i}=\infty}(Y_{it}-Y_{i1})}{\sum_{i;E_{i}=\infty}1}$
of $\beta_{t}$ and $Y_{it}-Y_{i1}-\beta_{t}^{\#}$ of $\tau_{it}$.
We find that
\begin{align*}
 & \Var^{\#}\left[\sum_{i}\left(\sum_{t;it\in\Omega_{1}}v_{it}^{*}\right)\left(\sum_{t;it\in\Omega_{1}}v_{it}^{*}\hat{\tau}_{it}\right)\right]\leq\Var^{\#}\left[\sum_{i}\left(\sum_{t;it\in\Omega_{1}}v_{it}^{*}\right)\left(\sum_{t;it\in\Omega_{1}}v_{it}^{*}(Y_{it}-Y_{i1}-\beta_{t}^{\#})\right)\right]\\
 & \leq3\Var^{\#}\left[\sum_{i}\left(\sum_{t;it\in\Omega_{1}}v_{it}^{*}\right)^{2}Y_{i1}\right]+3\Var^{\#}\left[\sum_{i}\left(\sum_{t;it\in\Omega_{1}}v_{it}^{*}\right)\left(\sum_{t;it\in\Omega_{1}}v_{it}^{*}Y_{it}\right)\right]\\
 & \phantom{=}+3\Var^{\#}\left[\sum_{i}\left(\sum_{t;it\in\Omega_{1}}v_{it}^{*}\right)\left(\sum_{t;it\in\Omega_{1}}v_{it}^{*}\beta_{t}^{\#}\right)\right]\\
 & =3\sum_{i}\left(\sum_{t;it\in\Omega_{1}}v_{it}^{*}\right)^{4}+3\sum_{i}\left(\sum_{t;it\in\Omega_{1}}v_{it}^{*}\right)^{2}\sum_{t;it\in\Omega_{1}}v_{it}^{*}{}^{2}+3\Var^{\#}\left[\sum_{t}\beta_{t}^{\#}\sum_{i;it\in\Omega_{1}}v_{it}^{*}\left(\sum_{s;is\in\Omega_{1}}v_{is}^{*}\right)\right]\\
 & \leq3(T+1)\sum_{i}\left(\sum_{t;it\in\Omega_{1}}v_{it}^{*}\right)^{2}\sum_{t;it\in\Omega_{1}}v_{it}^{*}{}^{2}+3T\sum_{t}\left(\sum_{i;it\in\Omega_{1}}v_{it}^{*}\left(\sum_{s;is\in\Omega_{1}}v_{is}^{*}\right)\right)^{2}\Var^{\#}\left[\beta_{t}^{\#}\right]\\
 & \leq3(T+1)\sum_{e}\sum_{i;E_{i}=e}\left(\sum_{t;it\in\Omega_{1}}v_{it}^{*}\right)^{2}\underbrace{\sum_{t\geq e}\max_{j;E_{j}=e}w_{jt}^{2}}_{\mathclap{\leq\frac{C^{2}\sum_{j;E_{j}=e}\sum_{t\geq e}w_{jt}^{2}}{\sum_{j;E_{j}=e}1}}}+3T\underbrace{\sum_{t}\Var^{\#}\left[\beta_{t}^{\#}\right]\sum_{i;it\in\Omega_{1}}v_{it}^{*}{}^{2}\cdot\sum_{i;it\in\Omega_{1}}\left(\sum_{s;is\in\Omega_{1}}v_{is}^{*}\right)^{2}}_{\leq\left(\sum_{i}\left(\sum_{t;it\in\Omega_{1}}v_{it}^{*}\right)^{2}\right)\cdot\left(\sum_{it\in\Omega_{1}}v_{it}^{*}{}^{2}\right)\cdot\left(\sum_{t}\Var^{\#}\left[\beta_{t}^{\#}\right]\right)}\\
 & \leq3(T+1)\left(\sum_{i}\left(\sum_{t;it\in\Omega_{1}}v_{it}^{*}\right)^{2}\right)\cdot\left(\sum_{it\in\Omega_{1}}v_{it}^{*}{}^{2}\right)\cdot\left(\frac{C^{2}}{\min_{e}\sum_{i;E_{i}=e}1}+\sum_{t}\Var^{\#}\left[\beta_{t}^{\#}\right]\right),
\end{align*}
where we have repeatedly used the Cauchy\textendash Schwartz inequality
and that weights do not vary too much within cohort\textendash period
cells. Similarly, for each $e$ and $t\geq e$,
\begin{align*}
 & \Var^{\#}\left[\sum_{i;E_{i}=e}v_{it}^{*}{}^{2}\hat{\tau}_{it}\right]\leq\Var^{\#}\left[\sum_{i;E_{i}=e}v_{it}^{*}{}^{2}(Y_{it}-Y_{i1}-\beta_{t}^{\#})\right]\\
 & \leq3\Var^{\#}\left[\sum_{i;E_{i}=e}v_{it}^{*}{}^{2}Y_{it}\right]+3\Var^{\#}\left[\sum_{i;E_{i}=e}v_{it}^{*}{}^{2}Y_{i1}\right]+3\Var^{\#}\left[\sum_{i;E_{i}=e}v_{it}^{*}{}^{2}\beta_{t}^{\#}\right]\\
 & =6\sum_{i;E_{i}=e}v_{it}^{*}{}^{4}+3\left(\sum_{i;E_{i}=e}v_{it}^{*}{}^{2}\right)^{2}\Var^{\#}\left[\beta_{t}^{\#}\right]\\
 & \leq6\left(\sum_{i;E_{i}=e}v_{it}^{*}{}^{2}\right)\cdot\underbrace{\left(\max_{i;E_{i}=e}v_{it}^{*}{}^{2}\right)}_{\leq C^{2}\frac{\sum_{i;E_{i}=e}v_{it}^{*}{}^{2}}{\sum_{i;E_{i}=e}1}}+3\left(\sum_{i;E_{i}=e}v_{it}^{*}{}^{2}\right)^{2}\Var^{\#}\left[\beta_{t}^{\#}\right]\\
 & \leq6\left(\sum_{i;E_{i}=e}v_{it}^{*}{}^{2}\right)\cdot\left(\sum_{it\in\Omega_{1}}v_{it}^{*}{}^{2}\right)\left(\frac{C^{2}}{\sum_{i;E_{i}=e}1}+\Var^{\#}\left[\beta_{t}^{\#}\right]\right).
\end{align*}
Finally, for each $t$, $\Var^{\#}\left[\hat{\beta}_{t}\right]\leq\Var^{\#}\left[\beta_{t}^{\#}\right]=\frac{\Var^{\#}\left[\sum_{i;E_{i}=\infty}(Y_{it}-Y_{i1})\right]}{\left(\sum_{i;E_{i}=\infty}1\right)^{2}}\leq\frac{4}{\sum_{i;E_{i}=\infty}1}\rightarrow0.$

Putting everything together,
\begin{align*}
 & \expec{\wnorm^{-2}\sum_{i}\left(\sum_{t;it\in\Omega_{1}}v_{it}^{*}(\tilde{\tau}_{it}^{A}-\bar{\tau}_{it})\right)^{2}}\leq T\bar{\sigma}^{2}\frac{\Var^{\#}\left[\sum_{i}\left(\sum_{t;it\in\Omega_{1}}v_{it}^{*}\right)\left(\sum_{t;it\in\Omega_{1}}v_{it}^{*}\hat{\tau}_{it}\right)\right]}{\wnorm^{2}\sum_{i}(\sum_{t;it\in\Omega_{1}}v_{it}^{*})^{2}}\\
 & \leq3(T+1)T\bar{\sigma}^{2}\frac{\sum_{it\in\Omega_{1}}v_{it}^{*}{}^{2}}{\wnorm^{2}}\left(\frac{C^{2}}{\min_{e}\sum_{j;E_{j}=e}1}+\sum_{t}\Var^{\#}\left[\tilde{\beta}_{t}\right]\right)\rightarrow0,\\
 & \expec{\wnorm^{-2}\sum_{i}\left(\sum_{t;it\in\Omega_{1}}v_{it}^{*}(\tilde{\tau}_{it}^{B}-\bar{\tau}_{it})\right)^{2}}\leq T\bar{\sigma}^{2}\frac{T\sum_{e}\sum_{t\geq e}\frac{\Var^{\#}\left[\sum_{i;E_{i}=e}v_{it}^{*}{}^{2}\hat{\tau}_{it}\right]}{\sum_{i;E_{i}=e}v_{it}^{*}{}^{2}}}{\wnorm^{2}}\\
 & \leq6T^{2}\bar{\sigma}^{2}\frac{\sum_{it\in\Omega_{1}}v_{it}^{*}{}^{2}}{\wnorm^{2}}\sum_{e}\sum_{t\geq e}\left(\frac{C^{2}}{\sum_{i;E_{i}=e}1}+\Var^{\#}\left[\tilde{\beta}_{t}\right]\right)\rightarrow0,\\
 & \expec{\wnorm^{-2}\sum_{i}\left(\sum_{t;it\in\Omega}v_{it}^{*}(\hat{\beta}_{t}-\beta_{t})\right)^{2}}\leq T\bar{\sigma}^{2}\sum_{t}\Var^{\#}\left[\hat{\beta}_{t}\right]\frac{\sum_{i;it\in\Omega}v_{it}^{*}{}^{2}}{\wnorm}\rightarrow0,
\end{align*}
where we have used that cohort sizes increase and that $\frac{\sum_{it\in\Omega_{1}}v_{it}^{*}{}^{2}}{\wnorm^{2}}\leq\frac{\sum_{it\in\Omega}v_{it}^{*}{}^{2}}{\wnorm^{2}}\leq1$.
This establishes \ref{eq:nuisanceconsistency}, and thus the conditions
of \ref{thm:se}.
\end{proof}
\begin{proof}[Proof of \ref{prop:leaveout}]
 As in the proof of \ref{thm:se}, we have that $\sum_{t;it\in\Omega}A_{i}^{\prime}\left(\hat{\lambda}_{i}-\lambda_{i}\right)=0$
by unbiasedness of $\hat{\tau}_{w}$. Therefore
\begin{align*}
 & \expec{\hat{\sigma}_{w}^{2}}=\expec{\sum_{i}\left(\sum_{t;it\in\Omega}v_{it}\tilde{\varepsilon}_{it}^{LO}\right)^{2}}=\expec{\sum_{i}\left(\sum_{t;it\in\Omega}v_{it}\left(\varepsilon_{it}-X_{it}^{\prime}(\hat{\delta}^{-i}-\delta)-D_{it}(\tilde{\tau}_{it}^{-i}-\tau_{it})\right)\right)^{2}}\\
 & =\expec{\sum_{i}\left(\sum_{t;it\in\Omega}v_{it}\varepsilon_{it}\right)^{2}}+\expec{\sum_{i}\left(\sum_{t;it\in\Omega}v_{it}\left(X_{it}^{\prime}(\hat{\delta}^{-i}-\delta)+D_{it}(\tilde{\tau}_{it}^{-i}-\tau_{it})\right)\right)^{2}}\geq\sigma_{w}^{2},
\end{align*}
where we have used that $\varepsilon_{it}$ is uncorrelated to $\hat{\delta}^{-i},\tilde{\tau}_{it}^{-i}$.
\end{proof}
\begin{proof}[Proof of \ref{lem:linear-est-weights}]
By standard OLS results, $\hat{\psi}_{w}=w^{\prime}\left(z^{\prime}z\right)^{-1}z^{\prime}y$.
Thus, $\hat{\psi}_{w}=v^{\prime}y$ for weights $v=z\left(z^{\prime}z\right)^{-1}w$
which do not depend on $y$. Hence, $v_{j}=z_{j}^{\prime}\check{\psi}$
for $\check{\psi}=\left(z^{\prime}z\right)^{-1}w$.
\end{proof}
\begin{proof}[Proof of \ref{lem:Equal-sensitivity}]
Any linear estimator unbiased for $\tau_{w}$ under \ref{assu:A1,assu:A2}
has to be numerically invariant to adding any combination of unit
and period FEs to the outcome. Thus, $\hat{\tau}_{w}$ would be the
same with the data $\tilde{Y}_{it}=Y_{it}-\kappa_{0}-\kappa_{1}E_{i}+\kappa_{1}t=-\left(\kappa_{0}+\kappa_{1}K_{it}\right)D_{it}$.
These data satisfy \ref{assu:A1,assu:A2} with the corresponding $\tau_{it}=-\left(\kappa_{0}+\kappa_{1}K_{it}\right)$
for $it\in\Omega_{1}$, and thus $\expec{\hat{\tau}_{w}}=\sum_{it\in\Omega_{1}}w_{it}\tau_{it}=-\sum_{it\in\Omega_{1}}w_{it}\left(\kappa_{0}+\kappa_{1}K_{it}\right).$
\end{proof}
\begin{proof}[Proof of \ref{lem:no_sensitivity_ranking}]
Without loss of generality, suppose $B_{\hat{\tau}_{w}^{A}}(y_{0})>0$,
and thus $B_{\hat{\tau}_{w}^{A}}(y_{0})>B_{\hat{\tau}_{w}^{B}}(y_{0})$.
\ref{lem:Equal-sensitivity} implies that there exists a set of violations
of \ref{assu:A1,assu:A2}, $y_{0}^{\star}\in\mathbb{R}^{\left|\Omega_{0}\right|}$,
such that $B_{\hat{\tau}_{w}^{A}}(y_{0}^{\star})=B_{\hat{\tau}_{w}^{B}}(y_{0}^{\star})>B_{\hat{\tau}_{w}^{A}}(y_{0})$.
Since $B$ is linear in $y_{0}$ for any linear estimator, this implies
$0<B_{\hat{\tau}_{w}^{A}}(y_{0}^{\star}-y_{0})<B_{\hat{\tau}_{w}^{B}}(y_{0}^{\star}-y_{0})$,
which concludes the proof with $\tilde{y}_{0}=y_{0}^{\star}-y_{0}$.
\end{proof}
\printbibliography

\section*{Appendix Tables and Figures}

\addcontentsline{toc}{section}{Appendix Tables and Figures}\setcounter{table}{0} \renewcommand{\thetable}{A\arabic{table}}
\setcounter{figure}{0} \renewcommand{\thefigure}{A\arabic{figure}}
\renewcommand\theHtable{Appendix.\thetable}\renewcommand\theHfigure{Appendix.\thefigure}

\begin{table}[H]
\noindent \begin{centering}
\caption{\ref{tab:replication_BP} Estimates Relative to Binned OLS\label{tab:Coef-differences}}
\medskip{}
{\small{}}%
\begin{tabular}{lccccc}
\toprule 
 & \multicolumn{2}{c}{{\small{}Without disbursement}} &  & \multicolumn{2}{c}{{\small{}With disbursement}}\tabularnewline
 & \multicolumn{2}{c}{{\small{}method FEs}} &  & \multicolumn{2}{c}{{\small{}method FEs}}\tabularnewline
\cmidrule{2-3} \cmidrule{3-3} \cmidrule{5-6} \cmidrule{6-6} 
 & {\small{}No binning} & {\small{}Imputation} &  & {\small{}No binning} & {\small{}Imputation}\tabularnewline
 & {\small{}(1)} & {\small{}(2)} &  & {\small{}(3)} & {\small{}(4)}\tabularnewline
\midrule
{\small{}Contemporaneous month} & {\small{}-7.56} & {\small{}-4.45} &  & {\small{}-19.69} & {\small{}-17.04}\tabularnewline
 & {\small{}(4.59)} & {\small{}(5.68)} &  & {\small{}(7.94)} & {\small{}(9.72)}\tabularnewline
 & {\small{}{[}0.100{]}} & {\small{}{[}0.433{]}} &  & {\small{}{[}0.013{]}} & {\small{}{[}0.080{]}}\tabularnewline
{\small{}First month after} & {\small{}-11.59} & {\small{}-11.78} &  & {\small{}-30.74} & {\small{}-18.83}\tabularnewline
 & {\small{}(5.91)} & {\small{}(8.17)} &  & {\small{}(10.36)} & {\small{}(16.26)}\tabularnewline
 & {\small{}{[}0.050{]}} & {\small{}{[}0.149{]}} &  & {\small{}{[}0.003{]}} & {\small{}{[}0.247{]}}\tabularnewline
{\small{}Second month after} & {\small{}-14.60} & {\small{}4.44} &  & {\small{}-34.34} & {\small{}-16.52}\tabularnewline
 & {\small{}(7.02)} & {\small{}(22.92)} &  & {\small{}(12.78)} & {\small{}(29.34)}\tabularnewline
 & {\small{}{[}0.038{]}} & {\small{}{[}0.846{]}} &  & {\small{}{[}0.007{]}} & {\small{}{[}0.574{]}}\tabularnewline
 &  &  &  &  & \tabularnewline
{\small{}Three-month total} & {\small{}-33.74} & {\small{}-11.78} &  & {\small{}-84.77} & {\small{}-52.39}\tabularnewline
 & {\small{}(17.45)} & {\small{}(33.28)} &  & {\small{}(30.66)} & {\small{}(48.00)}\tabularnewline
 & {\small{}{[}0.053{]}} & {\small{}{[}0.723{]}} &  & {\small{}{[}0.006{]}} & {\small{}{[}0.275{]}}\tabularnewline
\bottomrule
\end{tabular}\medskip{}
\par\end{centering}
\noindent \raggedright{}\emph{\footnotesize{}Notes:}{\footnotesize{}
The coefficients reported in this table are differences between the
estimates from OLS with no binning (columns 1 and 3) or imputation
(columns 2 and 4) and the binned OLS specification in Table 2 of the
draft. Standard errors clustered by household are reported in parentheses;
p-values for the null that the difference is equal to zero are shown
in brackets.}{\footnotesize\par}
\end{table}

\begin{figure}[H]
\caption{Weights Implied by Dynamic Specifications with and without Trimming\label{fig:Trimming}}

\vspace{4mm}
\begin{centering}
\begin{tabular}{cc}
{\small{}A: With Trimming} & {\small{}B: Without Trimming}\tabularnewline
\includegraphics[width=0.35\textwidth]{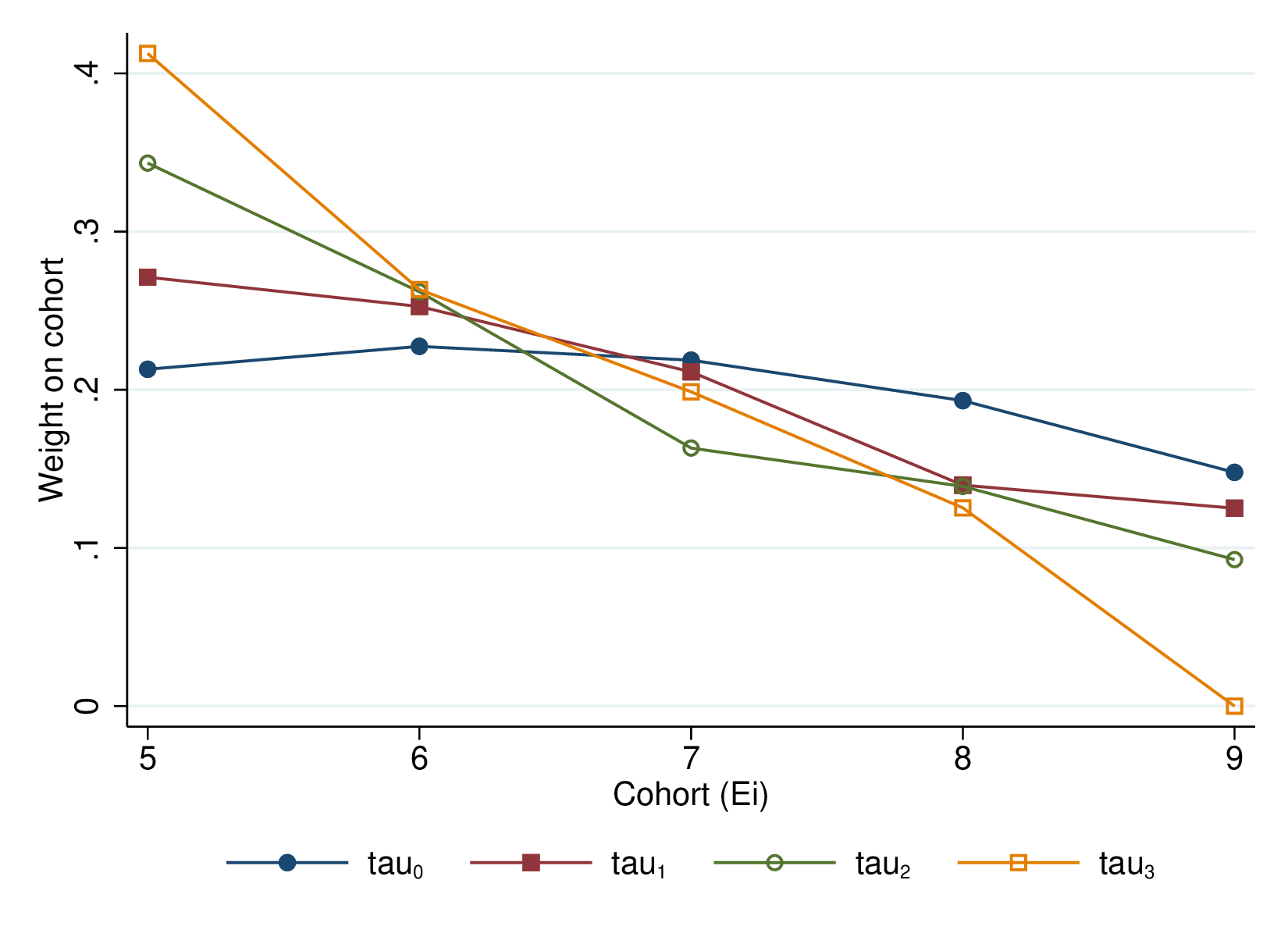} & \includegraphics[width=0.35\textwidth]{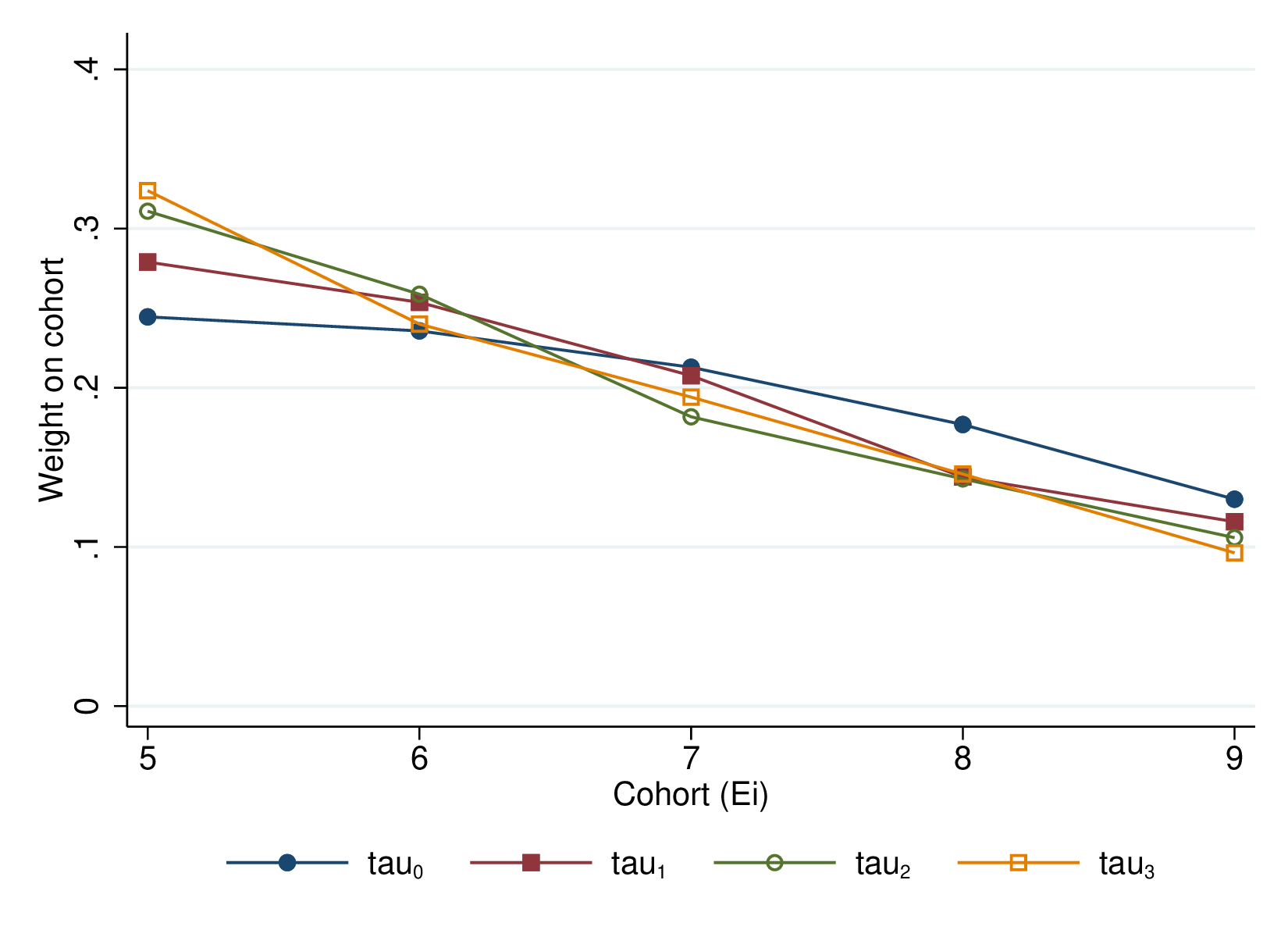}\tabularnewline
\end{tabular}
\par\end{centering}
\emph{\small{}Notes:}{\small{} For the numerical example described
in \ref{subsec:Trimming} this figure reports the total weight that
the $\hat{\tau}_{h}$ estimator from the semi-dynamic specification,
for each horizon $h=0,\dots,3$, places on the treated observations
from each cohort $e$ observed $h$ periods after treatment. The horizontal
axis and the lines correspond to the cohort $e$ and horizon $h$,
respectively. Panel A trims the sample to include observations with
$K_{it}\in[-4,3]$ only, while Panel B includes all data (but does
not report the weights for the coefficients $\tau_{4},\dots,\tau_{7}$).
The weights placed by $\hat{\tau}_{h}$ on observations at horizons
other than $h$ are not shown.}{\small\par}
\end{figure}

\begin{landscape}

\begin{figure}
\caption{Efficiency and Bias of Alternative Estimators\label{fig:Monte-Carlo-Simulation}}
\medskip{}

\begin{centering}
\begin{tabular}{cccc}
\multicolumn{4}{l}{\textbf{\small{}Standard deviations under different DGPs}}\tabularnewline
{\small{}1a: Spherical errors} & {\small{}1b: Heteroskedasticity} & {\small{}1c: AR(1)} & {\small{}1d: Wild clustered}\tabularnewline
\includegraphics[width=0.32\textwidth]{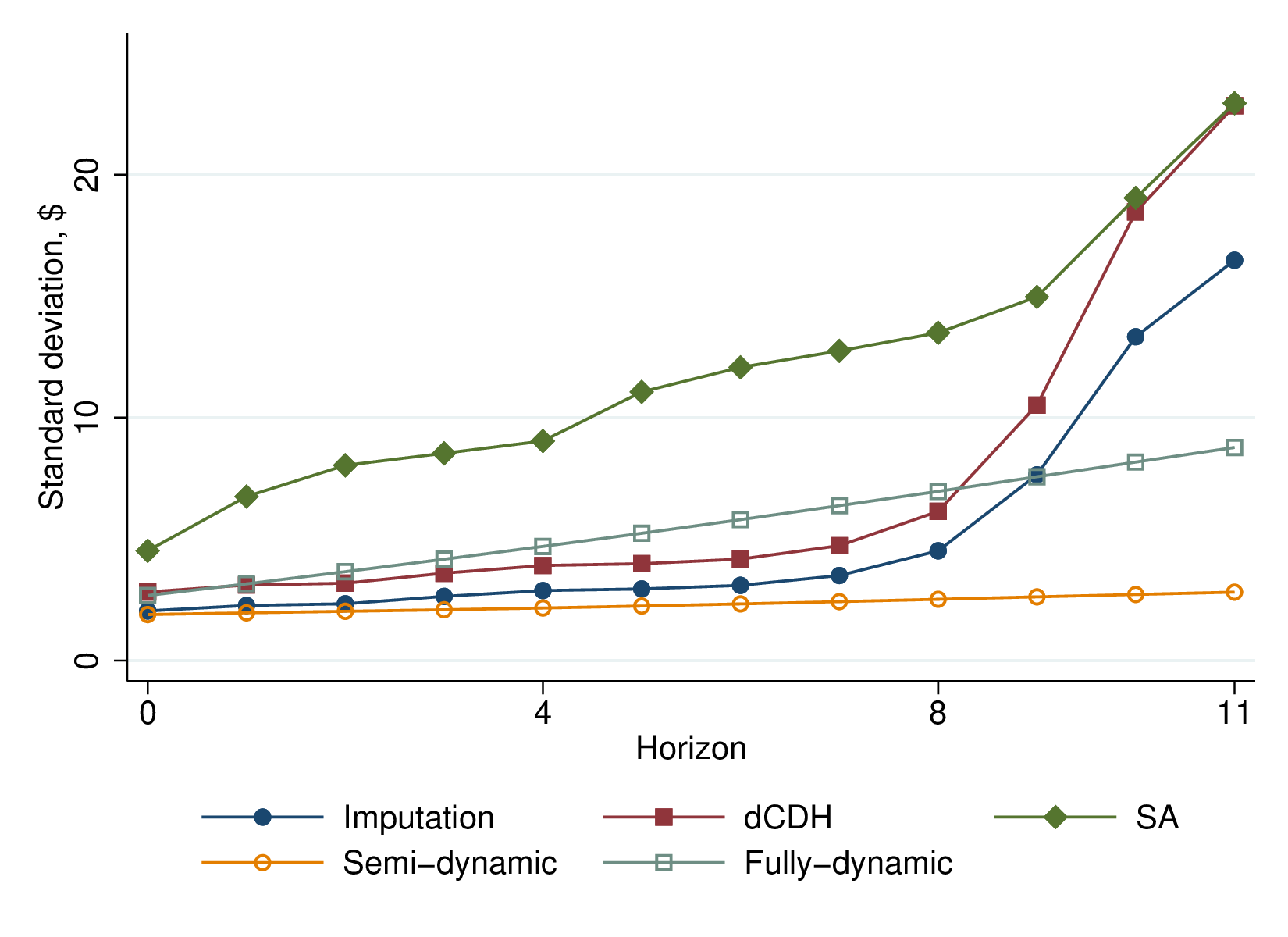} & \includegraphics[width=0.32\textwidth]{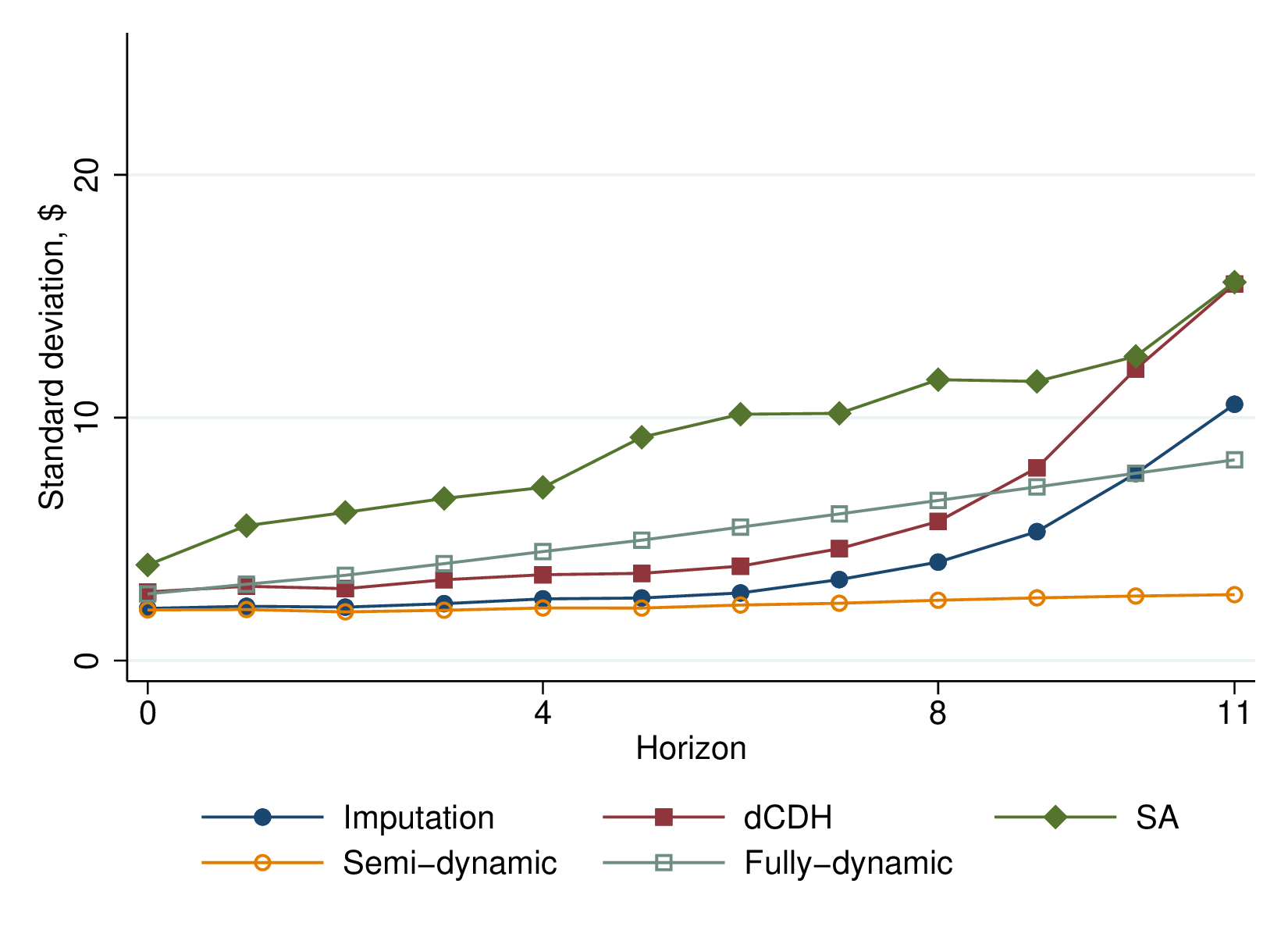} & \includegraphics[width=0.32\textwidth]{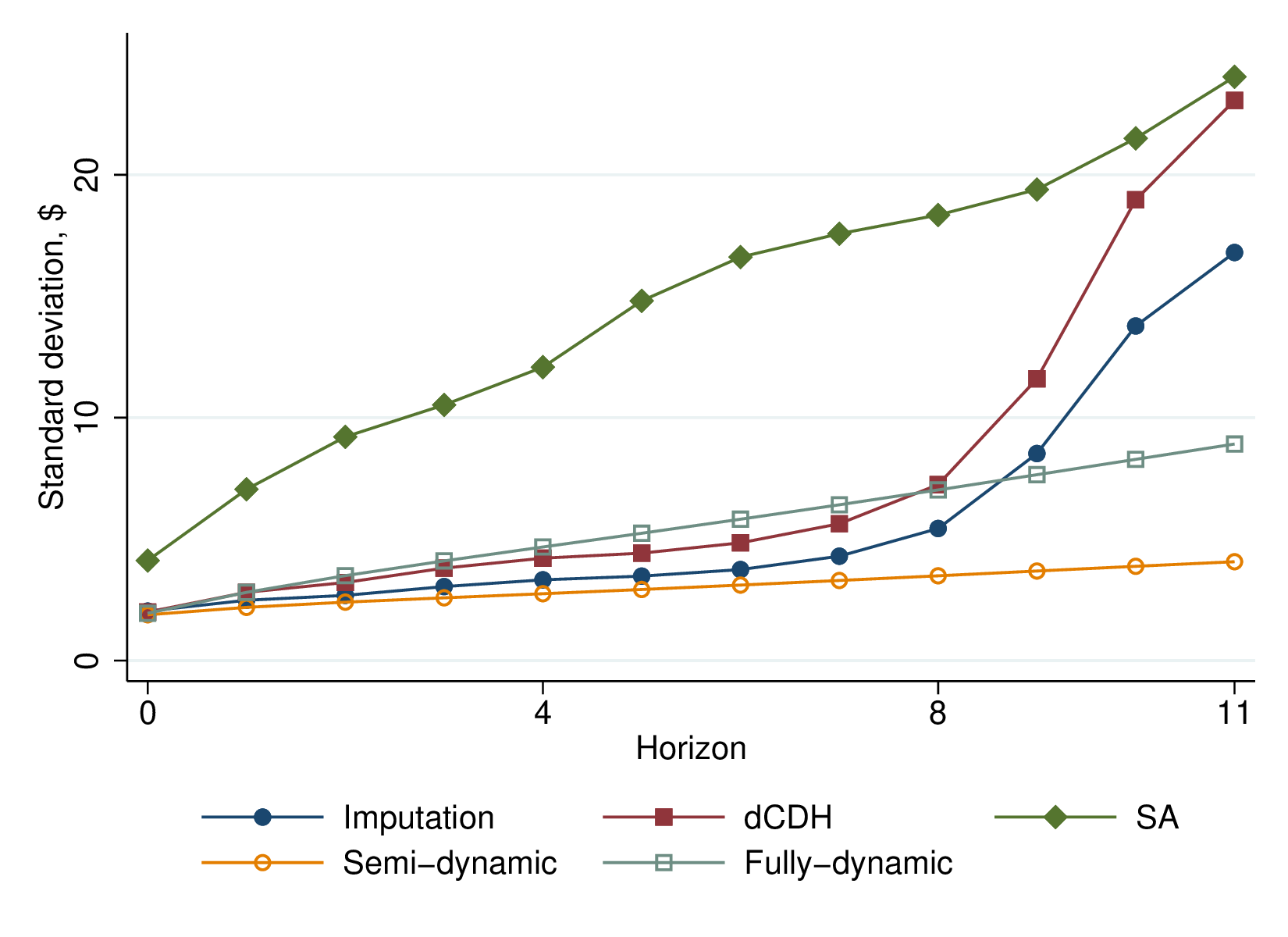} & \includegraphics[width=0.32\textwidth]{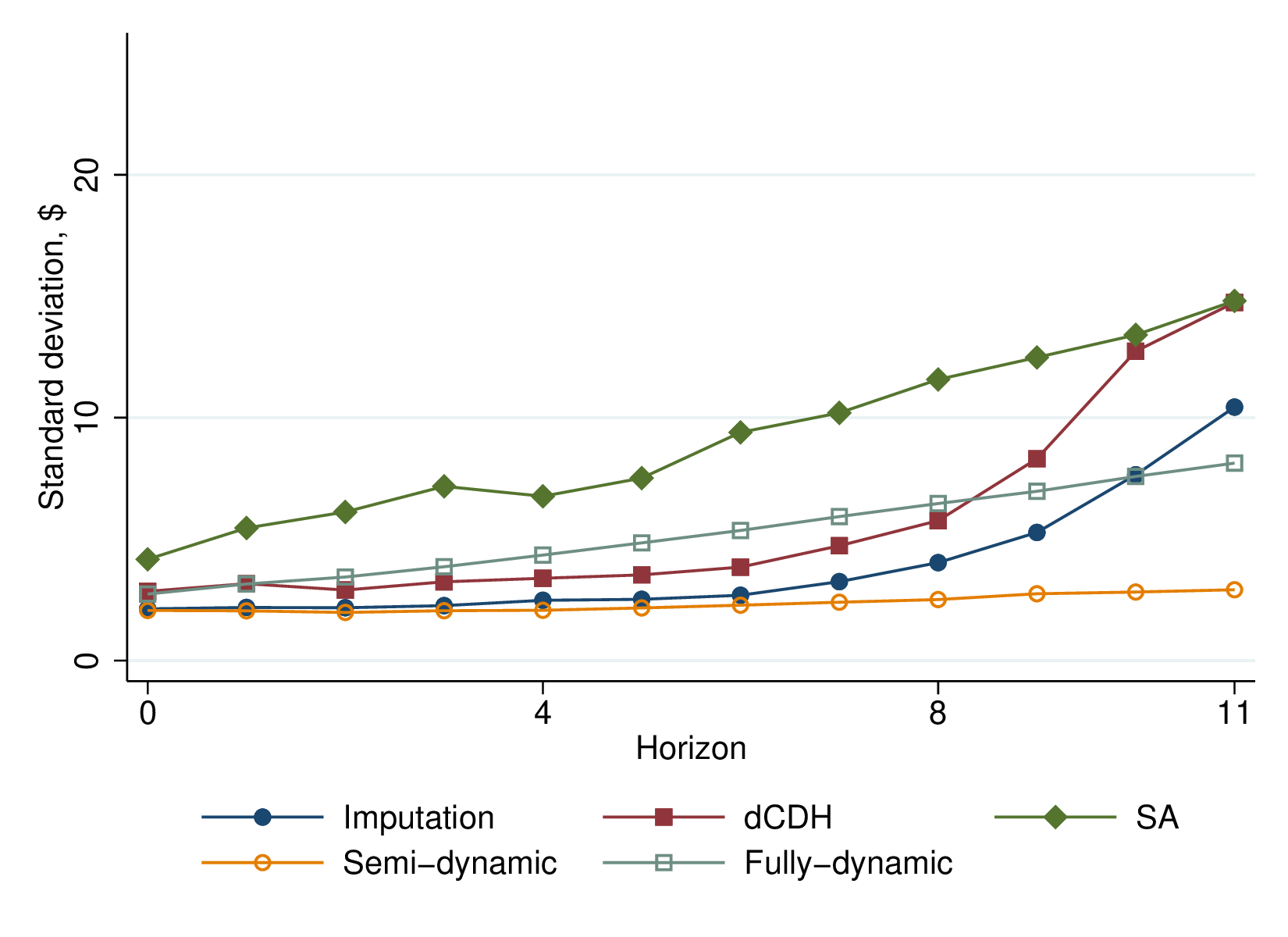}\tabularnewline
 &  &  & \tabularnewline
\multicolumn{4}{l}{\textbf{\footnotesize{}SD relative to the imputation estimator}}\tabularnewline
{\footnotesize{}2a: Spherical errors} & {\footnotesize{}2b: Heteroskedasticity} & {\footnotesize{}2c: AR(1)} & {\footnotesize{}2d: Wild clustered}\tabularnewline
\includegraphics[width=0.32\textwidth]{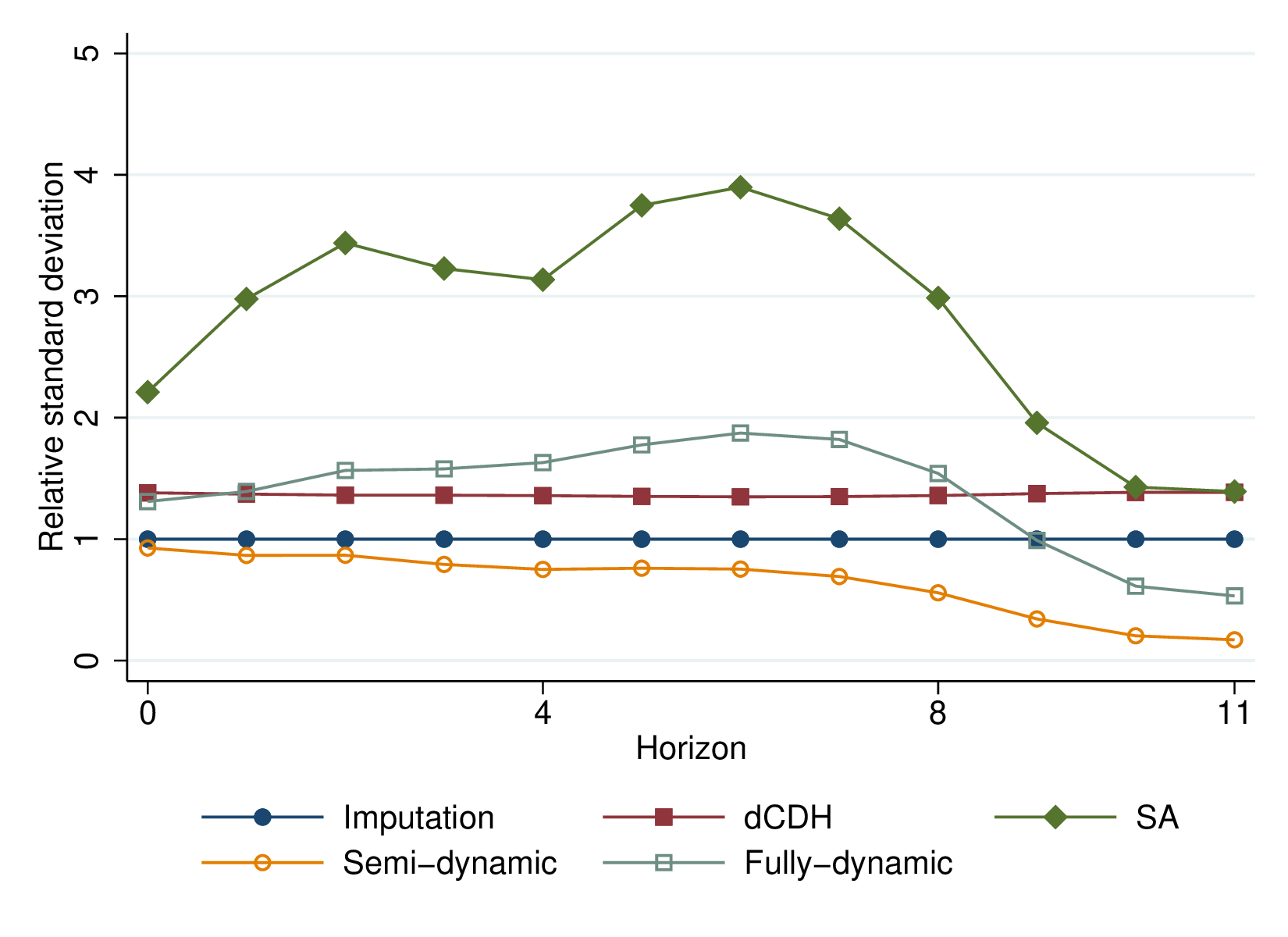} & \includegraphics[width=0.32\textwidth]{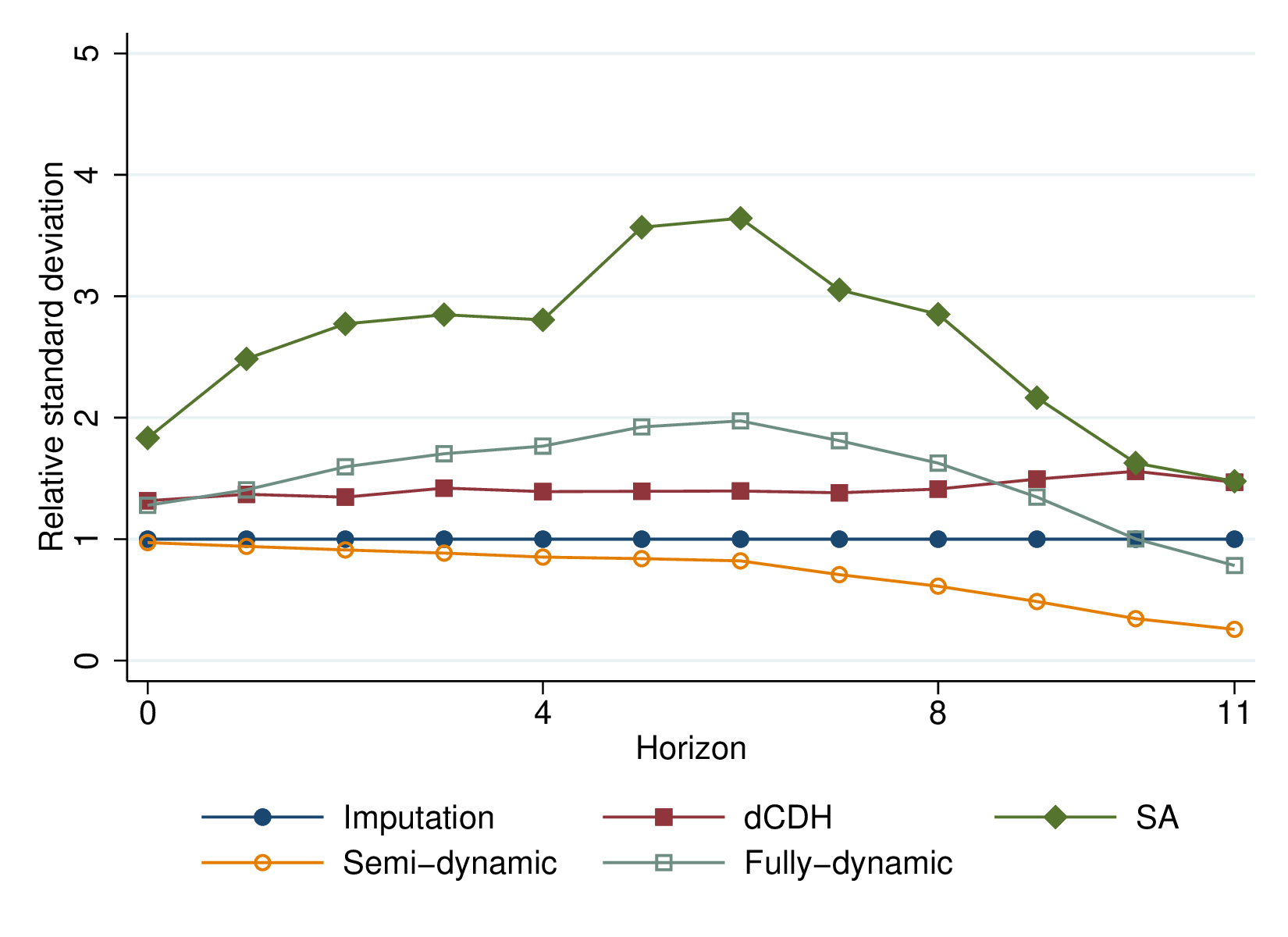} & \includegraphics[width=0.32\textwidth]{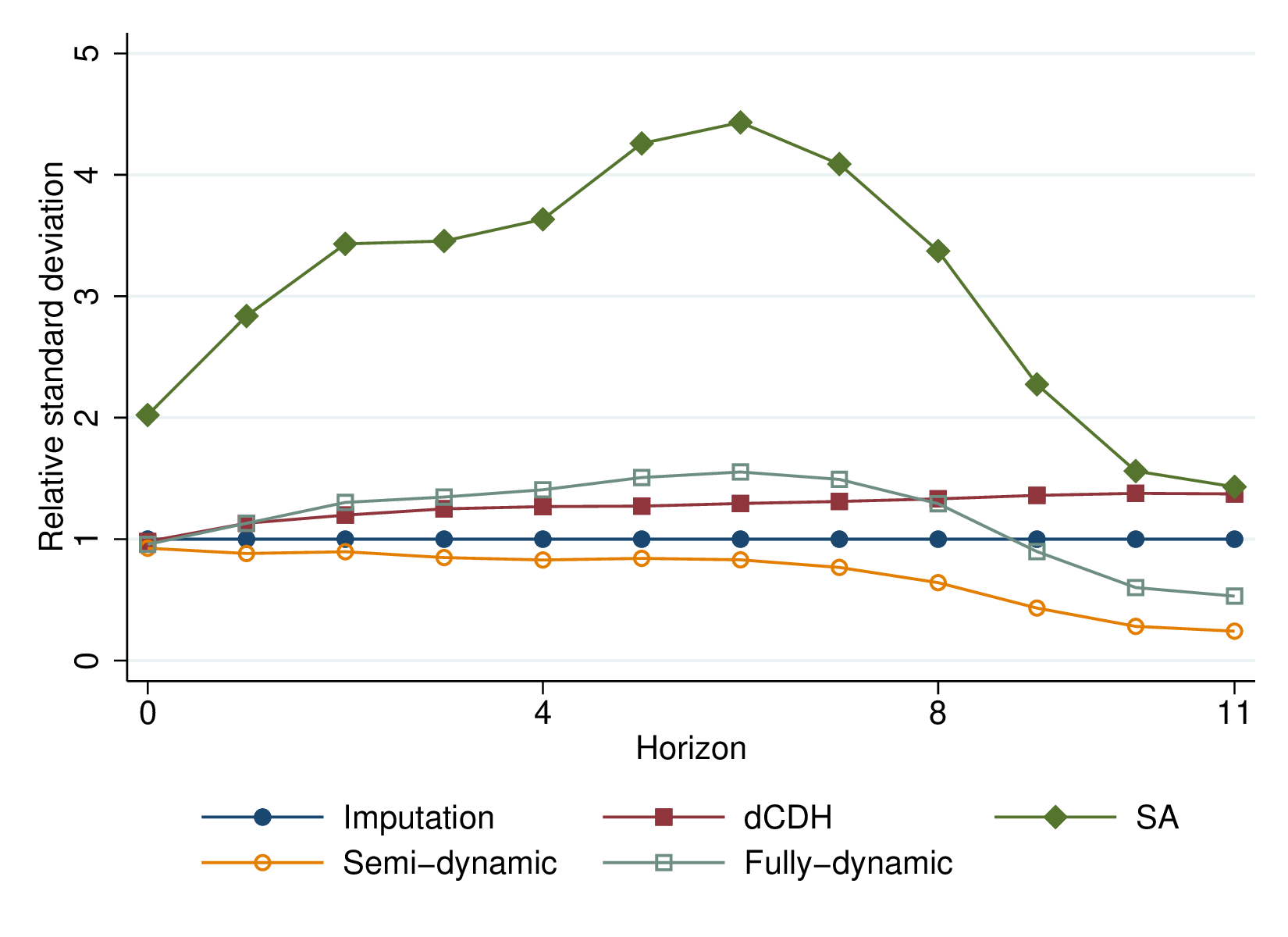} & \includegraphics[width=0.32\textwidth]{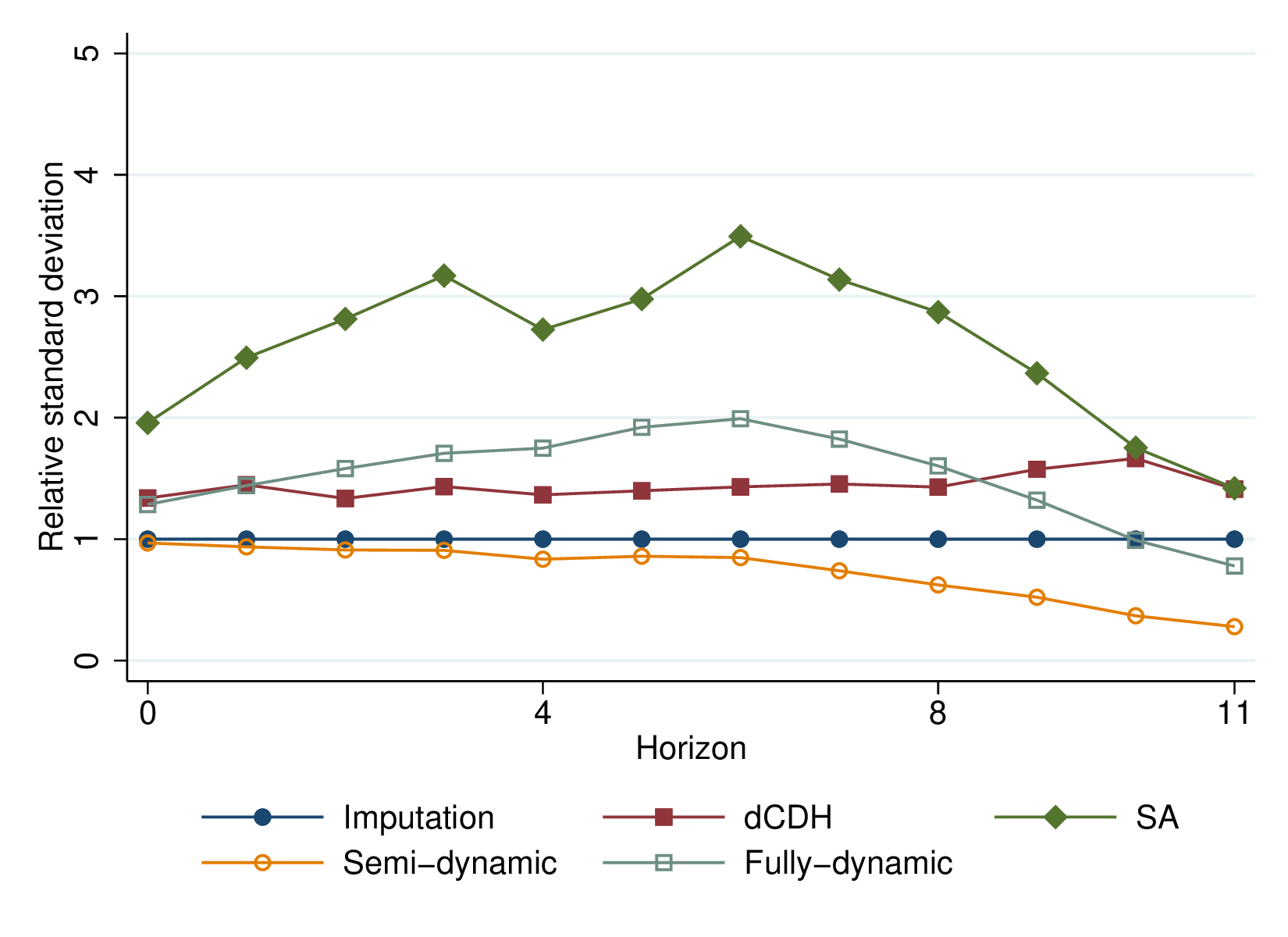}\tabularnewline
 &  &  & \tabularnewline
\multicolumn{3}{l}{\textbf{\small{}Other simulations}} & \tabularnewline
\multirow{1}{*}{{\small{}3: Bias}} & {\small{}4: RMSE, spherical errors} & \multirow{1}{*}{{\small{}5: Sensitivity to anticipation}} & \multirow{1}{*}{{\small{}6: Imputation estimator coverage}}\tabularnewline
\includegraphics[width=0.32\textwidth]{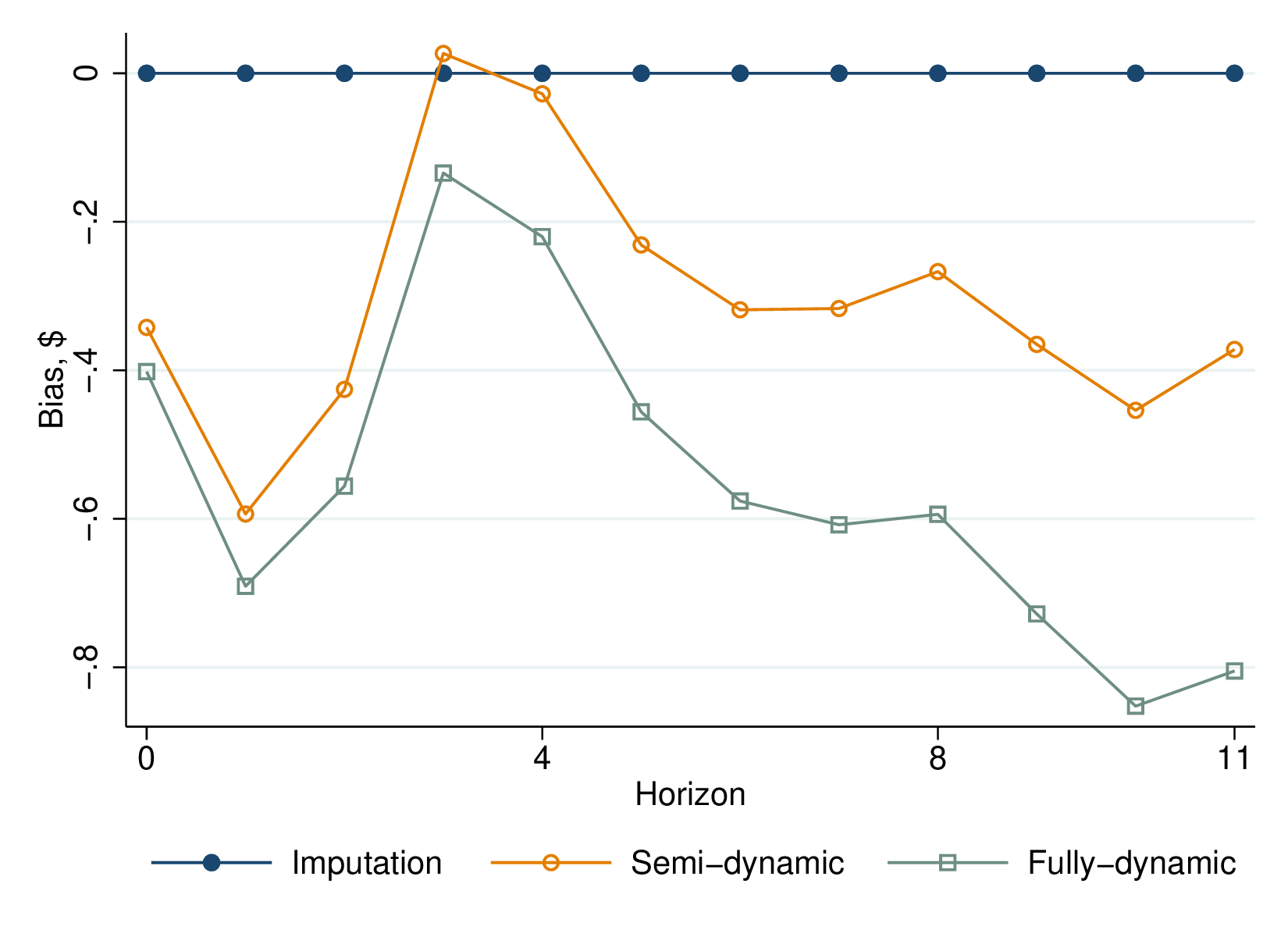} & \includegraphics[width=0.32\textwidth]{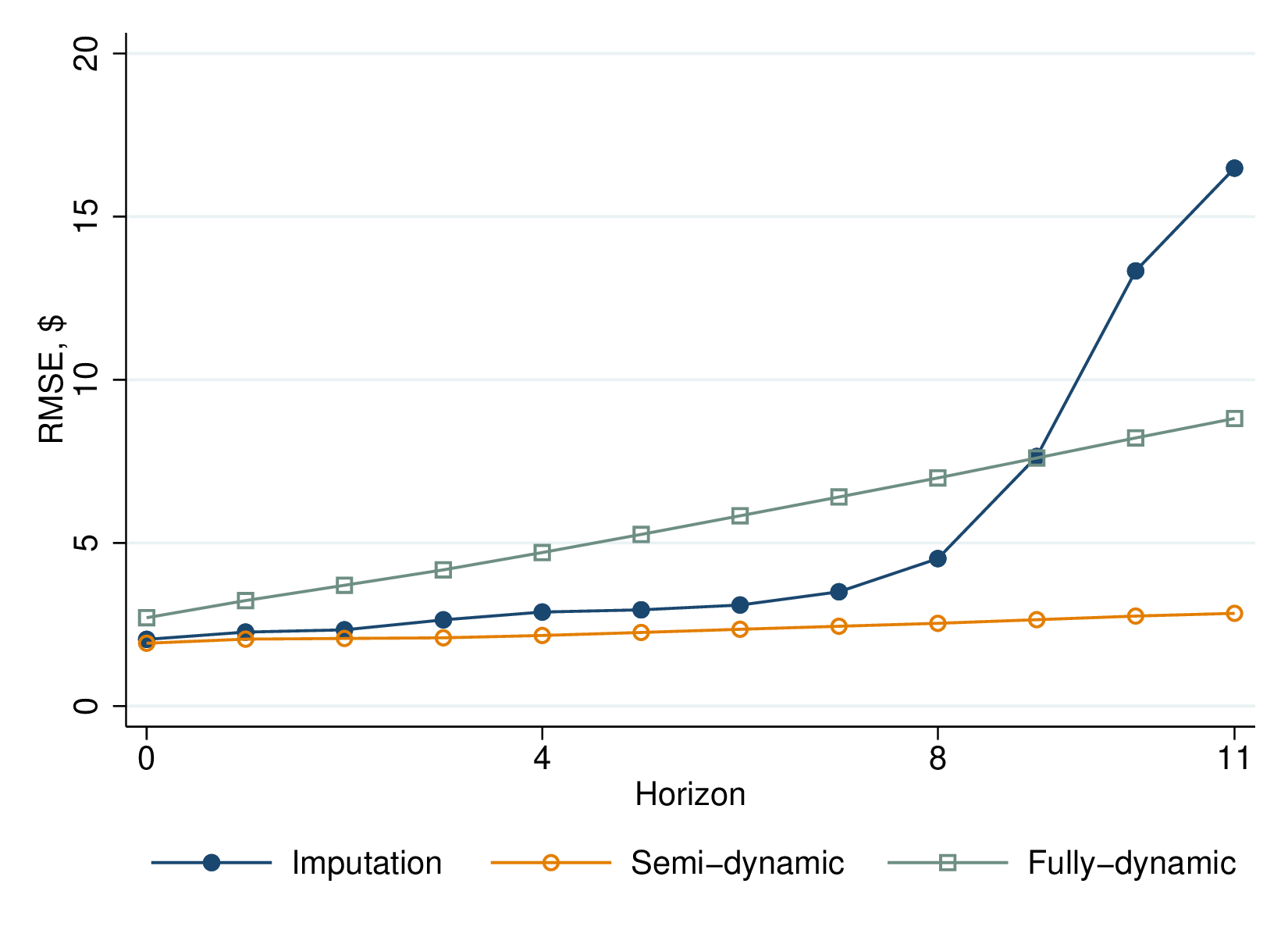} & \includegraphics[width=0.32\textwidth]{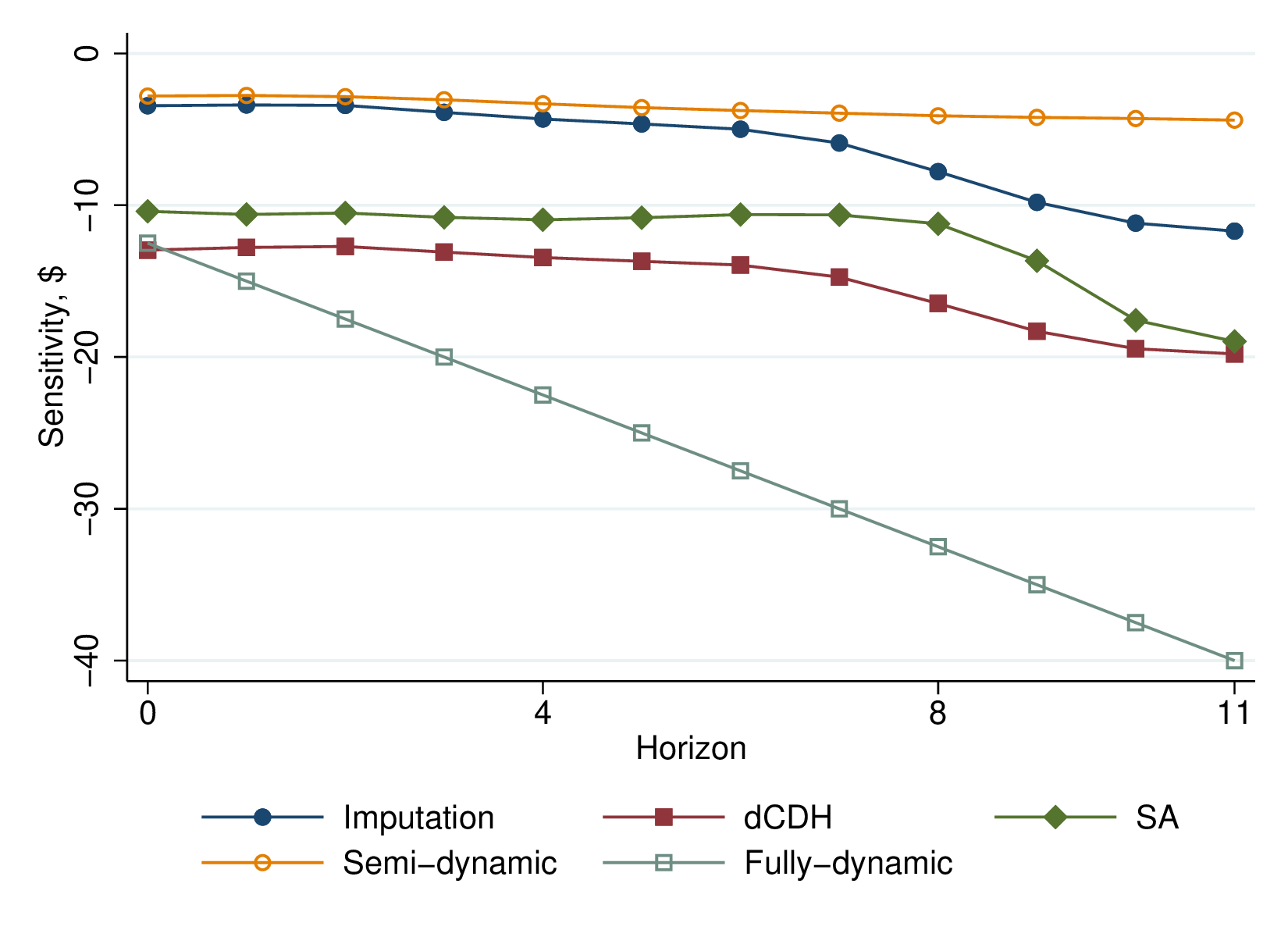} & \includegraphics[width=0.32\textwidth]{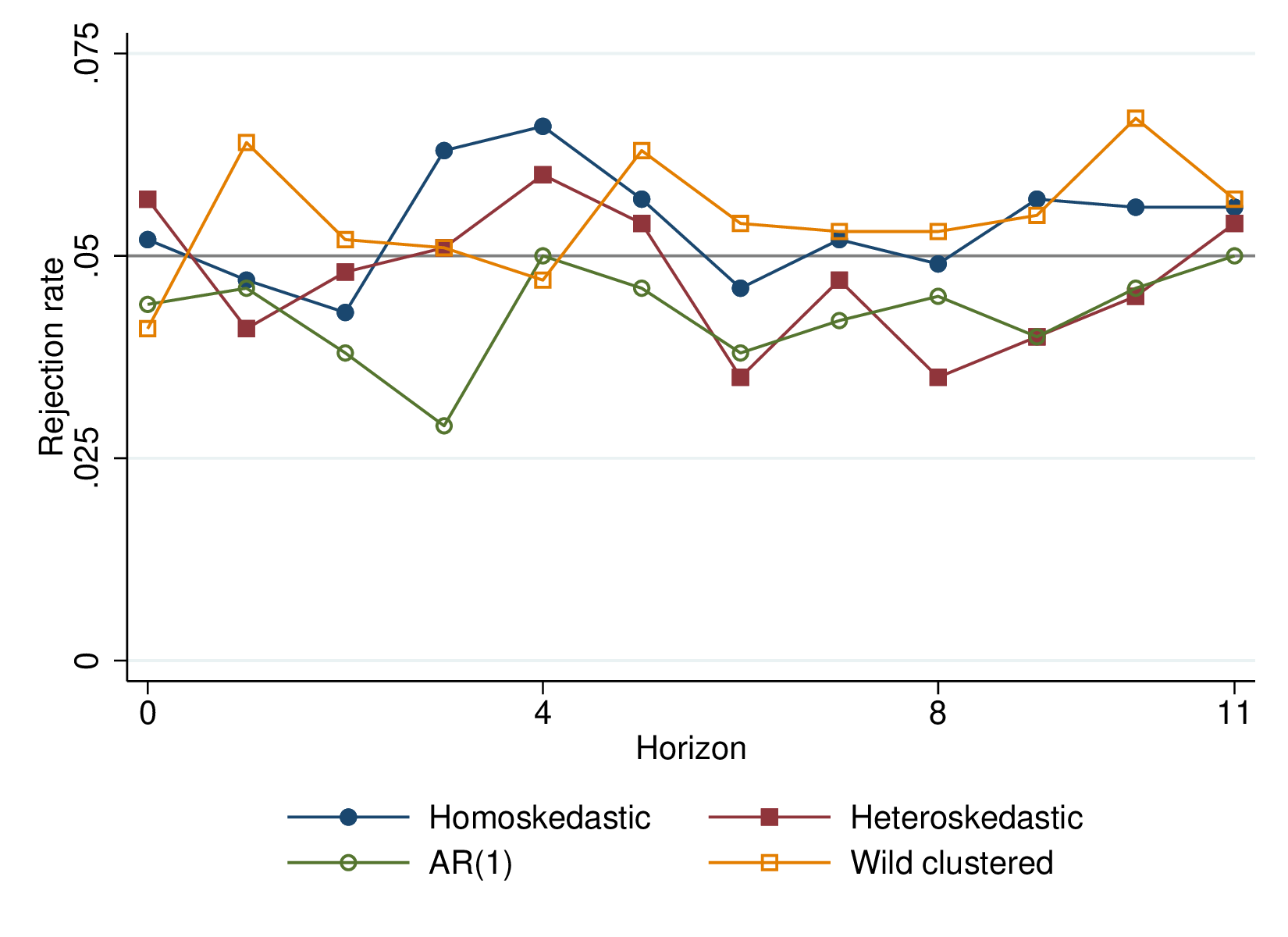}\tabularnewline
\end{tabular}
\par\end{centering}
\begin{centering}
\vspace{2mm}
\par\end{centering}
\emph{\small{}Notes:}{\small{} See \ref{subsec:Monte-Carlo-BP-Appx-NEW}
for a detailed description of the data-generating processes and reported
statistics.}{\small\par}
\end{figure}
\end{landscape}
\end{document}